%% file: main.tex
\documentclass{article}
\usepackage{graphicx} 
\input{preamble}

\title{Strong converse exponent of channel interconversion}

\author[]{Aadil Oufkir}
\author[]{Yongsheng Yao}
\author[]{Mario Berta}

\affil[]{\small{Institute for Quantum Information,
  RWTH Aachen University,
  Aachen, Germany}}

\begin{document}

\maketitle

\begin{abstract}
In their seminal work, Bennett {\it et al.}~[IEEE Trans.~Inf.~Theory (2002)] showed that, with sufficient shared randomness, one noisy channel can simulate another at a rate equal to the ratio of their capacities. We establish that when coding above this channel interconversion capacity, the exact strong converse exponent is characterized by a simple optimization involving the difference of the corresponding R\'enyi channel capacities with H\"older dual parameters. We further extend this result to the entanglement-assisted interconversion of classical–quantum channels, showing that the strong converse exponent is likewise determined by differences of sandwiched R\'enyi channel capacities. The converse bound is obtained by relaxing to non-signaling assisted codes and applying H\"older duality together with the data processing inequality for R\'enyi divergences. Achievability is proven by concatenating refined channel coding and simulation protocols that go beyond first-order capacities, attaining an exponentially small conversion error, remaining robust under small variations in the
input distribution, and tolerating a sublinear gap between the conversion rates.
\end{abstract}


\section{Overview}

\paragraph{Introduction.} In his seminal work, Shannon \cite{shannon1948mathematical} established that the maximum achievable rate for reliable classical information transmission through a noisy channel is characterized by a single parameter $C$, called the capacity. This fundamental quantity defines two distinct operational regimes:
For rates below capacity ($r=\frac{1}{n}\log M<C$), messages can be transmitted with exponentially vanishing error probability as the blocklength $n$ increases \cite{shannon59,shannon1967lower,Fano1961Nov}. Conversely, any attempt to communicate above capacity ($r=\frac{1}{n}\log M>C$) inevitably leads to an error probability that approaches one exponentially fast in $n$~\cite{han89}. This latter phenomenon is known as the \textit{strong converse}, with the exponential decay rate of the success probability referred to as the \textit{strong converse exponent}.
Both the capacity concept and the strong converse property were later extended by  \cite{Holevo1998Jan, Schumacher1997Jul} and \cite{Winter2002Aug,Ogawa2002Aug,MosonyiOgawa2017strong} to coding over channels with classical input and quantum output, so-called classical-quantum channels.

Dual to the noisy channel coding problem, channel simulation aims to reproduce a noisy channel’s behavior using noiseless communication. Remarkably, when shared randomness is available (shared entanglement for classical-quantum channels), the channel capacity $C$ serves as the same fundamental threshold \cite{BSST2002entanglement,BDHSW2014quantum,BCR2011the}. As in coding—though operationally inverted—simulation exhibits a sharp transition at capacity: successful simulation with exponentially vanishing error is achievable when the perfect communication rate $r=\frac{1}{n}\log M$ exceeds capacity $(r>C)$ \cite{LiYao2021reliable,Oufkir2024OctTV, Li2024OctLarge,Oufkir2024OctPur}. The strong converse occurs for rates below capacity $(r<C)$, where the simulation error is guaranteed to converge to one exponentially fast \cite{BSST2002entanglement,BDHSW2014quantum,Oufkir2024OctTV,Li2024OctLarge,Oufkir2024OctPur}. The exponential rate of this convergence defines the strong converse exponent.

Generalizing both channel coding and channel simulation, the channel interconversion problem seeks to simulate a possibly noisy target channel using communication over another possibly noisy channel. When shared randomness is allowed (shared entanglement for classical-quantum channels), the first-order asymptotic rate for this task is given by the ratio of the capacities of the channels involved \cite{BSST2002entanglement,cao2024channel}. In fact, this can be derived by concatenating the known protocols for channel coding and simulation, pleasantly leading to first order reversibility. However, to the best of our knowledge, little—if anything—is known about higher-order refinements for this problem, partly exactly because concatenating protocols in a black box like fashion will in general not be optimal. In this work, we take the first step by studying the strong converse exponents for the channel interconversion problem. Next, we present an overview of our results. The precise definitions of the quantities involved are given in the subsequent section.

\paragraph{Overview of results.} In this paper, we study the strong converse exponent for channel interconversion under the purified and total variation distances. It represents the optimal rate of exponential convergence of the conversion error towards the worst-case performance. For classical-quantum channels, the strong converse exponent for the conversion $W^{\otimes n}\rightarrow T^{\otimes \ceil{rn}}$ under the \textit{purified distance} takes the form (see Theorem \ref{thm:maincq})
\begin{align}\label{intro-res:pur}
    \sup_{1< p < 2}\frac{2-p}{p}\left(r\cdot \widetilde{I}_{\frac{p}{2}}(T)-\widetilde{I}_{\frac{p}{2(p-1)}}(W)\right)
\end{align}
where $\widetilde{I}_{\alpha}(W)$ is the channel's sandwiched R\'enyi mutual information of order $\alpha$, i.e.,
\begin{equation}
    \widetilde{I}_{\alpha}(W)=\inf_{\sigma \,:\, \text{state}} \max_{x} \widetilde{D}_{\alpha}(W_x\|\sigma) \qquad \forall \alpha \in (0,\infty),
\end{equation}
 and $\widetilde{D}_{\alpha}$ is the sandwiched R\'enyi divergence of order $\alpha$, defined as 
\begin{equation}
    \widetilde{D}_{\alpha}(\rho \| \sigma)=\frac{1}{\alpha-1} \log \operatorname{Tr}\big({\sigma}^{\frac{1-\alpha}{2\alpha}} \rho {\sigma}^{\frac{1-\alpha}{2\alpha}}\big)^\alpha \qquad \forall \alpha \in (0,\infty). 
\end{equation}
Notably, the exact characterization of the strong converse exponent in \eqref{intro-res:pur} is the same under the entanglement and non-signaling assistance.  Furthermore, if the channel $T$ is classical, shared randomness assistance is sufficient to achieve such performance. Finally, by either setting $T = \id$ or $W = \id$, we neatly recover the strong converse exponent of channel coding and channel simulation, respectively.

We further characterize the strong converse exponent for the conversion of \textit{classical} channels $W^{\times n} \rightarrow T^{\times \ceil{rn}}$ under the \textit{total-variation distance} for any assistance ranging from shared randomness to non-signaling. This exponent takes the  form (see Theorem~\ref{thm:classical})
\begin{align}\label{intro-res:tv}
    \sup_{0< \alpha < 1}\sup_{1<p<\frac{1}{1-\alpha}}	\frac{1-p+\alpha p}{p}\left(r\cdot I_{(1-\alpha)p}(T)-I_{\frac{\alpha p}{p-1}}(W)\right).
\end{align}
Further, we prove the strong converse exponent under the most general R\'enyi divergence as a notion of measure. In particular, we quantify the strong converse exponent under $\mathtt{d}_{\alpha}(P\|Q)=1-(\sum_xP_x^\alpha Q_x^{1-\alpha})^{\frac{1}{1-\alpha}}$ for $\alpha \in (0, 1)$
\begin{align}
\sup_{1<p<\frac{1}{\alpha}}	\frac{1-\alpha p}{(1-\alpha) p}\left(r\cdot I_{\alpha p}(T)-I_{\frac{(1-\alpha) p}{p-1}}(W)\right).
\end{align}
Similarly, by either setting $T = \id$ or $W = \id$, we recover the strong converse exponent of the channel coding and the channel simulation, respectively.
We refer to Table \ref{tab:results} for a summary of our results. 

\paragraph{Technical challenges.} In channel coding as well as channel simulation problems, higher-order refinements are typically obtained via  relating the task performance to so-called meta-converses \cite{tomamichel2015second, WWY2014strong, gupta2015multiplicativity, Matthews2014Sep,Wang2019Feb,cao2024channel}. The latter can then be further tightly controlled using entropic divergences and variances in the block-length regime \cite{li2014second,tomamichel2012framework,Nagaoka2006Nov,hayashi2007error,MosonyiOgawa2015quantum,li2022tight,li2024operational}. In contrast, in the case of direct channel interconversion, such refinements are fundamentally more challenging due to the absence of a unified meta-converse that simultaneously captures the properties of both channels involved. This lack of a tight meta-converse framework makes it unclear how to directly relate interconversion performance to entropic quantities using standard methods and addressing this challenge requires the development of novel techniques beyond the traditional meta-converse-based approach. The channel interconversion problem might also be seen as the channel-superchannel analogue to the state-channel conversion problem of quantum dichotomies (cf.~the latest \cite{Berta2024Oct} and references therein). Nevertheless, determining the strong converse exponents for the former is significantly more challenging than for the latter.

\paragraph{Technical contributions.} We derive the achievability part of \eqref{intro-res:pur} and \eqref{intro-res:tv} starting from ideas first outlined in \cite{Dueck1979Jan} and \cite[Problem 2.5.16]{csiszar2011information} for classical channel coding. The methodology was enhanced to prove the strong converse exponent for classical-quantum \cite{MosonyiOgawa2017strong} and quantum \cite{Li2023Nov} channels, as well as quantum  dichotomies \cite{Berta2024Oct}. This starts by analyzing separate cases based on whether the rate exceeds the ratio of the channel mutual information. However, contrary to the aforementioned works, a standard first order code is insufficient here due to the bipartite nature of channel interconversion. Consequently, we derive a refined achievability result that goes beyond the first-order bound obtained by simply concatenating standard channel simulation and coding protocols. Namely, our new result attains an exponentially small conversion error, remains robust under small variations in the input distribution, and tolerates a sublinear gap between the conversion rates. The protocol is carefully designed using a tight finite block-length approach, incorporating the random coding error exponent \cite{Fano1961Nov}, tight channel simulations \cite{Oufkir2024OctMeta}, Taylor expansions of the mutual information \cite{hayashi2016correlation}, and relevant continuity bounds.

Unlike the purified distance (and R\'enyi divergences), the strong converse exponent under the TV distance \eqref{intro-res:tv} involves an unusual and technically demanding double optimization over two parameters $\alpha$ and $p$. Conjecturing such an expression is rather non-trivial, even though it correctly reduces to the known strong converse exponents for channel coding and simulation when one of the channels is the identity.  A key technical tool  that enables this double optimization\,---\,and underpins our achievability result for the TV distance\,---\,is the following inequality. For any probability distributions $P$, $Q$, $V$, and any $\eta \in (0,\frac{1}{2})$, we have that
\begin{equation}
    -\log(1-\operatorname{TV}(P, Q)) \le \log\left(\frac{1}{1-2\eta}\right) + \sup_{0 < \alpha < 1} \Big(\alpha D_{\max}^{\eta}(V \| P) + (1-\alpha) D_{\max}^{\eta}(V \| Q) \Big),
\end{equation}
where $D_{\max}^{\eta}$ denotes the smooth max-divergence. This inequality admits a generalization to the quantum setting, and we believe it is of independent interest in both classical and quantum information theory. The relation resembles that for the purified distance, as proved in \cite{li2024operational}, which we also use in this work. For any states $\rho$, $\sigma$ and $\tau$, we have that
\[
 -\log(1-\operatorname{Pur}(\rho, \sigma)) \le   \log 2+  D(\tau\| \rho) + D(\tau\| \sigma).
\]
The main difference between the total variation distance relation and the purified distance relation lies in that, the former requires an optimization over $\alpha \in (0,1)$, while the latter corresponds to a fixed value of $\alpha = \frac{1}{2}$. This connection is further reflected by the strong converse exponent characterizations: one can obtain \eqref{intro-res:pur} (up to a normalization by a factor of $2$) by setting $\alpha = \frac{1}{2}$ in \eqref{intro-res:tv}.

The converse of \eqref{intro-res:pur} and \eqref{intro-res:tv} is proved via the most general non-signaling strategies, employing Hölder's inequality, and the data-processing inequality for the R\'enyi divergences.

\paragraph{Related results.} The strong converse exponents for channel simulation have recently been obtained for classical channels under the total variation distance \cite{Oufkir2024OctTV}, under measures based on R\'enyi divergences \cite{Li2024OctLarge}, and for classical-quantum channels under the purified distance \cite{Oufkir2024OctPur}. The works \cite{Oufkir2024OctTV,Oufkir2024OctPur} proceed by showing the strong converse achievability with non-signaling strategies, then invoke rounding results \cite{Oufkir2024OctMeta} to show the existence of shared randomness-assisted or entanglement-assisted strategies. As outlined above, our approach is conceptually different. As such, our work shows a unification of the proof methods for establishing strong converse exponents for the two different tasks of channel simulation and channel coding.

\paragraph{Outline.} The rest of this manuscript is structured as follows. Before formally describing the channel interconversion problem in Section \ref{sec:problem-form}, we introduce the necessary notation and preliminaries in Section \ref{sec:prel}. In Section \ref{sec:classical}, we then present the strong converse exponent of classical channel interconversion under both the total variation and the purified distance. The classical-quantum setting is studied for the purified distance in Section \ref{sec:CQ} before concluding with some open problems in Section \ref{sec:conclusion}. We note that each section is split into a classical and a quantum part, such that it is straightforward to comprehend the classical results without having to digest any quantum notation.

\begin{table}[t!]
    \centering
    \renewcommand{\arraystretch}{1.6}
    \begin{tabular}{>{\centering\arraybackslash}m{2.2cm}|>{\centering\arraybackslash}m{7.5cm}|>{\centering\arraybackslash}m{5.5cm}}
        \hline
        & \textbf{Classical Channel Interconversion} \quad \qquad (SR, EA, NS) & \textbf{CQ Channel Interconversion} (EA, NS) \\
        \hline
        \textbf{Purified Distance} 
        & $\displaystyle\sup_{ 1< p < 2}\frac{2-p}{p}\left(r\cdot {I}_{\frac{p}{2}}(T)-{I}_{\frac{p}{2(p-1)}}(W)\right)$ 
        & $\displaystyle\sup_{1< p < 2}\frac{2-p}{p}\left(r\cdot \widetilde{I}_{\frac{p}{2}}(T)-\widetilde{I}_{\frac{p}{2(p-1)}}(W)\right)$ \\
        \hline
        \textbf{Total Variation Distance} 
        & $\displaystyle \sup_{\substack{0< \alpha < 1\\1<p< (1-\alpha)^{-1} }}\!\frac{1-p+\alpha p}{p}\left(\!r\cdot I_{(1-\alpha)p}(T)-I_{\frac{\alpha p}{p-1}}(W)\!\right)$ 
        & - \\
        \hline
        \textbf{R\'enyi   $\alpha\in (0,1)$} & $\sup_{1<p<\frac{1}{\alpha}}	\frac{1-\alpha p}{(1-\alpha) p}\left(r\cdot I_{\alpha p}(T)-I_{\frac{(1-\alpha) p}{p-1}}(W)\right)$ & -
        \\
        \hline
    \end{tabular}
    \caption{Overview of Results: Exponents of classical and classical-quantum (CQ) channel interconversion $W^{\otimes n}\rightarrow T^{\otimes \ceil{rn}}$ under the purified distance, the total variation distance and R\'enyi-based measure and various forms of assistance—shared randomness (SR), entanglement assistance (EA), and non-signaling (NS).}
    \label{tab:results}
\end{table}


\section{Preliminaries}\label{sec:prel}
\subsection{Notation}

\paragraph{Classical setting.}
In this work, we only consider finite alphabets $\cX, \cY, \cS$ and $\cZ$. The cardinality of $\cX$ is denoted $|\cX|$. $\bm{x}^n = (x_1, \dots, x_n)$ is an element of $\cX^n$.  The set of probability distributions on $\cX$ is denoted $\cP(\cX)$.  The set of all types of alphabet $\cX$ and length $n$ is denoted $\cT_n(\cX)$.  
$X$ denotes a  random variable distributed according to the probability distribution $P_{X}$ on $\cX$. $W_{Y|X}$ is a discrete memoryless channel from $\cX$ to $\cY$. For any input element $x \in \cX$,  $W_x(\cdot) = W_{Y|X=x}(\cdot) = W_{Y|X}(\cdot|x)$ is the conditional output probability distribution over $\cY$. For an input probability distribution $P_X$, the joint probability distribution of the input and output is denoted by $P_{X} W_{Y|X} (x,y) = P_{X}(x) W_{Y|X}(y|x)$.  The output probability distribution on $\cY$ given the input  distribution $P_{X}$ is denoted  by $PW_Y(\cdot) = \sum_{x\in \cX} P_{X}(x) W_{Y|X}(\cdot|x)$.  In the blocklength setting, we consider the iid channel $W_{\bm{Y|X}}^{\times n} = W_{Y_1|X_1}\times W_{Y_2|X_2} \times \cdots \times W_{Y_n|X_n}$ from $\cX^n$ to $\cY^n$. Finally, $P_{X} P_Y$ denotes a product  distribution on $\cX\times \cY$.\\

The total variation (TV) distance between two probability distributions $P, Q \in \mathcal{P}(\mathcal{X})$ is defined as follows:
\begin{align}
    \operatorname{TV}(P, Q) &= \frac{1}{2}\sum_{x\in \cX} \left|P(x) - Q(x)\right|.
\end{align}
The Kullback–Leibler divergence 
is defined  for $P, Q \in \mathcal{P}(\mathcal{X})$ as follows:
\begin{align}
    D(P\|Q) = \sum_{x\in \cX} P(x)\log \frac{P(x)}{Q(x)}
\end{align}
if $\supp(P)\subset \supp(Q)$ and $+\infty$ otherwise. 
\\The R\'enyi divergence is defined  for $P, Q \in \mathcal{P}(\mathcal{X})$ and $\alpha\in (0,1)\cup (1, \infty)$ as follows:
\begin{align}
    D_\alpha(P\|Q) = \frac{1}{\alpha -1} \log \sum_{x\in \cX} P(x)^{\alpha} Q(x)^{1-\alpha}
\end{align}
The max-divergence is defined  for $P, Q \in \mathcal{P}(\mathcal{X})$ as follows:
\begin{align}
    D_{\max}(P\| Q) = \inf\{\lambda : P \mle \exp(\lambda) Q\}.
\end{align}
The $\varepsilon$-smooth  relative max-divergence is defined  for $P, Q \in \mathcal{P}(\mathcal{X})$  and $\varepsilon\ge 0$ as follows:
\begin{align}
    D_{\max}^{\varepsilon}(P \| Q) = \inf_{\substack{\widetilde{P}\in \cP(\cX) \\ \mathrm{TV}(\widetilde{P}, P)\le \varepsilon}}  D_{\max}(\widetilde{P}\|Q).
\end{align}
The channel mutual information of the channel $W_{Y|X}$ and input probability distribution $P_X$ is defined as:
\begin{align}
    I(P_{X}, W_{Y|X}) &= \inf_{Q_Y\in \cP(\cY)} D\left( P_X W_{Y|X}\| P_X Q_Y\right). 
\end{align}
The channel capacity of the channel $W_{Y|X}$ is defined as:
\begin{align}
    C( W_{Y|X}) &=\sup_{P_{X}\in \cP(\cX)} I(P_{X}, W_{Y|X}) = \sup_{P_{X}\in \cP(\cX)}\inf_{Q_Y\in \cP(\cY)} D\left( P_X W_{Y|X}\| P_X Q_Y\right). 
\end{align}
The notion of the channel mutual information can be generalized in two different ways. Given a channel $W_{Y|X}$, an input probability distribution $P_X$ and a parameter $\alpha \in(0, \infty)$, the Sibson $\alpha$-mutual information is defined as \cite{Sibson1969Jun}
\begin{align}\label{eq:I(P,W)}
    I_{\alpha}(P_{X}, W_{Y|X}) &= \inf_{Q_Y\in \cP(\cY)} D_{\alpha}\left( P_X W_{Y|X}\| P_X Q_Y\right).
\end{align}
The Augustin--Csisz\'ar $\alpha$-mutual information is defined as \cite{Augustin78,csiszar1995generalized}
\begin{align}
    I^{\rm{ac}}_{\alpha}(P_{X}, W_{Y|X}) &= \inf_{Q_Y\in \cP(\cY)}  \mathbb{E}_{x\sim P_X} \left[D_{\alpha}\left(  W_{Y|X=x}\|  Q_Y\right)\right].
\end{align}
These notions of mutual information are different in general. However, they induce the same $\alpha$-channel capacity defined for $\alpha\in (0,\infty)$ \cite[Proposition 1]{csiszar1995generalized}:
\begin{align}
    I_{\alpha}(W_{Y|X}) = \sup_{P_{X}\in \cP(\cX)}I_{\alpha}(P_{X}, W_{Y|X}) =\sup_{P_{X}\in \cP(\cX)} I^{\rm{ac}}_{\alpha}(P_{X}, W_{Y|X}).
\end{align}
The mutual information variance of the channel $W_{Y|X}$ and the input probability distribution $P_X$ is defined as follows:
\begin{align}
    \var(P_X, W_{Y|X}) &= \sum_{x\in \cX} \sum_{y\in \cY} P_X(x)W_{Y|X}(y|x)\left( \log\Big(\frac{W(y|x)}{PW(y)}\Big) -D(P_XW_{Y|X}\| P_X PW_Y) \right)^2.
\end{align}
Taking the optimal  input probability distribution $P_X$, we obtain the mutual information variance of the channel $W_{Y|X}$:
\begin{align}
    \var(W_{Y|X}) &= \sup_{P_X\in \cP(\cX)}\var(P_X, W_{Y|X}).
\end{align}
\paragraph{Classical-quantum setting.} 
For a finite-dimensional Hilbert space $\mc{H}$, we let $\mc{L}(\mc{H})$ be the set of all linear operators on $\mc{H}$. We denote  the set of the  positive semi-definite  operators  on $\mc{H}$ as $\mc{P}(\mc{H})$. The set of quantum states on $\mc{H}$ is denoted as $\mathcal{S}(\mc{H})$. We use the notation $|\mc{H}|$ for the dimension of $\mc{H}$.  When $\mc{H}$ is associated with a system $A$, the above notations $\mc{L}(\mc{H})$, $\mc{P}(\mc{H})$, $\mathcal{S}(\mc{H})$, $|\mc{H}|$ and $I_{\mc{H}}$ also can be written as $\mc{L}(A)$, $\mc{P}(A)$, $\mathcal{S}(A)$, and $|A|$, respectively. The support of an operator $X$ is denoted by $\supp(X)$. For $A, B \in \mathcal{P}(\mathcal{H})$, we denote by $\{A \geq B\}$ the projection
onto the subspace spanned by the eigenvectors corresponding to the non-negative eigenvalues of $A-B$. The non-commutative minimum of  $A, B \in \mathcal{P}(\mathcal{H})$ is denoted $A\wedge B = A - (A-B)\{A \geq B\}$. 
A quantum channel $\mc{N}_{A \rightarrow B}$ is a completely positive and trace-preserving linear map from $\mc{L}(A)$ to $\mc{L}(B)$. 
A classical-quantum channel $T_{X \rightarrow Y}$ is a quantum channel with the following form:
\[
T_{X \rightarrow Y }(\cdot)=\sum_{x} \bra{x}(\cdot)\ket{x} \rho_Y^x,
\]
where $\{\ket{x}\}$ is an orthonormal
 basis of the underlying Hilbert space $\mc{H}_X$ and $\{\rho_Y^x\} \in \mc{S}(Y)$.\\

For $\rho , \sigma \in \mathcal{P}(\mathcal{H})$, the quantum relative
entropy of $\rho$ and $\sigma$ is defined as~\cite{Umegaki1954conditional}
\begin{equation}
D(\rho\|\sigma)= \begin{cases}
\tr{\rho(\log\rho-\log\sigma)} & \text{ if }\supp(\rho)\subseteq\supp(\sigma), \\
+\infty                        & \text{ otherwise.}
                  \end{cases}
\end{equation}
The log-Euclidean R\'enyi divergence and the sandwiched R\'enyi divergence were introduced in~\cite{MosonyiOgawa2017strong} and~\cite{MDSFT2013on, WWY2014strong}, respectively, which have been found wide applications so far.
Let $\alpha\in(0,+\infty)\setminus\{1\}$, $\rho\in\mathcal{S}(\mathcal{H})$ and $\sigma\in\mathcal{P}(\mathcal{H})$.
When $\alpha >1$ and $\supp(\rho)\subseteq\supp(\sigma)$ or $\alpha\in (0,1)$ and $\supp(\rho)\not\perp\supp(\sigma)$, the sandwiched R{\'e}nyi divergence of order $\alpha $
is defined as
\begin{align}
\widetilde{D}_{\alpha}(\rho \| \sigma)=\frac{1}{\alpha-1} \log \widetilde{Q}_{\alpha}(\rho \| \sigma),
\quad\text{with}\ \
\widetilde{Q}_{\alpha}(\rho \| \sigma)=\tr {({\sigma}^{\frac{1-\alpha}{2\alpha}} \rho {\sigma}^{\frac{1-\alpha}{2\alpha}})^\alpha}; \\
\end{align}
otherwise, we set $\widetilde{D}_{\alpha}(\rho \| \sigma)=+\infty$. When $\alpha > 1$ and $\supp(\rho)\subseteq\supp(\sigma)$  or  $\alpha \in (0,1)$ and $\supp(\rho)\cap\supp(\sigma)\neq\{0\}$, the log-Euclidean R{\'e}nyi divergence of order
$\alpha $ is defined as
\beq
D_{\alpha}^{\flat}(\rho \| \sigma)=\frac{1}{\alpha-1}\log Q_{\alpha}^{\flat}(\rho \| \sigma),
\quad\text{with}\ \
Q_{\alpha}^{\flat}(\rho \| \sigma)=\tr{2^{\alpha \log \rho +(1-\alpha) \log \sigma}};
\eeq
otherwise, we set $D_{\alpha}^{\flat}(\rho \| \sigma)=+\infty$. 

The $1/2$-sandwiched R\'enyi divergence is closely related to the fidelity and purified distance which are defined for two states $ \rho, \sigma \in \mathcal{S}(\mathcal{H})$ as follows:
\begin{align}
    F(\rho, \sigma) &= \left(\tr{\sqrt{\sqrt{\sigma}\rho \sqrt{\sigma}}}\right)^2,
    \\ \operatorname{Pur}(\rho, \sigma) &= \sqrt{1-F(\rho,\sigma)}.
\end{align}
For a quantum channel
 $\mc{N}_{A \rightarrow B}$, its entanglement-assisted classical capacity is defined as
 \begin{align}
 C_E(N_{A \rightarrow B})=\sup_{\rho_A \in \mc{S}(A)}
 \inf_{\sigma_B \in \mc{S}(B)}D(\mc{N}_{A \rightarrow B}(\rho_{RA})\| \rho_R \ox \sigma_B),
 \end{align}
where $\rho_{RA}$ is a purification of $\rho_A$.
For a classical-quantum channel $T_{X \rightarrow Y}$,
The channel mutual information of the channel $T_{X \rightarrow Y}$ and input probability distribution $P_X$ is defined as:
\begin{align}
    I(P_{X}, T_{X \rightarrow Y}) &= \inf_{\sigma_Y\in \mc{S}(Y)} \sum_{x} P_X(x) D\left( T_x\| \sigma_Y\right), 
\end{align}
where $T_x=T(\ket{x}\bra{x})$.

The mutual information variance of the classical-quantum channel $T_{X \rightarrow Y}$ and the input probability distribution $P_X$ is defined as follows:
\begin{align}
    \var(P_X, T_{X \rightarrow Y}) &= \sum_{x} P_X(x) \tr {T_x\left(\log T_x -\log(\sum_x P_X(x)W_x) \right)^2} .
\end{align}
Taking the optimal  input probability distribution $P_X$, we obtain the mutual information variance of the channel $T_{X \rightarrow Y}$:
\begin{align}
    \var(T_{X \rightarrow Y}) &= \sup_{P_X\in \cP(\cX)}\var(P_X, T_{X \rightarrow Y}).
\end{align}
Similar to the classical setting, for a classical-quantum channel $W_{X \rightarrow Y}$, we can define its sandwiched R\'enyi $\alpha$-channel capacity and log-Euclidean R\'enyi $\alpha$-channel capacity for $\alpha \in (0, \infty)$ respectively as
\begin{align}
\widetilde{I}_\alpha(W_{X \rightarrow Y}):&=\sup_{P_X \in \mathcal{P}(\mathcal{X})} \inf_{\sigma_Y \in \mc{S}(Y)} \widetilde{D}_\alpha(P_XW_{X \rightarrow Y}\| P_X\ox \sigma_Y) \\
&=\sup_{P_X \in \mathcal{P}(\mathcal{X})} \inf_{\sigma_Y \in \mc{S}(Y)} \mathbb{E}_{x \sim P_X}\widetilde{D}_\alpha(W_x\|\sigma_Y), \\
I_\alpha^\flat(W_{X \rightarrow Y}):&=\sup_{P_X \in \mathcal{P}(\mathcal{X})} \inf_{\sigma_Y \in \mc{S}(Y)} D_\alpha^\flat(P_XW_{X \rightarrow Y}\| P_X\ox \sigma_Y) \\
&=\sup_{P_X \in \mathcal{P}(\mathcal{X})} \inf_{\sigma_Y \in \mc{S}(Y)} \mathbb{E}_{x \sim P_X}D_\alpha^\flat(W_x\|\sigma_Y).
\end{align}

In the following proposition, we collect some important properties of these quantum R\'enyi divergences.
\begin{proposition}
\label{prop:mainpro}
For $\rho \in \mathcal{S}(\mathcal{H})$, $\sigma \in \mc{P}(\mc{H})$, the log-Euclidean R\'enyi divergence and the sandwiched R{\'e}nyi
divergence satisfy the following properties.
\begin{enumerate}[(i)]
  \item Monotonicity in $\sigma$~\cite{MDSFT2013on,MosonyiOgawa2017strong}: if $\sigma' \geq \sigma$, then $D^{(t)}_{\alpha}(\rho \| \sigma') \leq D^{(t)}_{\alpha}(\rho \| \sigma)$,
      for $(t)=\flat$, $\alpha \in (0,+\infty)$ and $(t)=\widetilde{}$,  $\alpha \in [\frac{1}{2},+\infty)$;
   \item  Data processing inequality \cite{Petz1986quasi,  Beigi2013sandwiched, MDSFT2013on, WWY2014strong, MosonyiOgawa2017strong}: for any quantum channel $\mc{N}$ from $\mc{L}(\mc{H})$ to $\mc{L}(\mc{H}')$, we have
      \begin{equation}
      D_{\alpha}^{(\rm{t})}(\mc{N}(\rho) \| \mc{N}(\sigma)) \leq D_{\alpha}^{(\rm{t})}(\rho \| \sigma),
      \end{equation}
      for $(t)=\flat$, $\alpha \in [0,1]$ and $(t)=\widetilde{}$, $\alpha \in [\frac{1}{2},+\infty)$;
\item Approximation by pinching~\cite{MosonyiOgawa2015quantum,hayashi2016correlation}:
        for $\alpha\geq 0$, we have
        \begin{equation}
    \widetilde{D}_\alpha(\mc{P}_\sigma(\rho)\|\sigma)
        \leq \widetilde{D}_\alpha(\rho\|\sigma)
        \leq \widetilde{D}_\alpha(\mc{P}_\sigma(\rho)\|\sigma)+2\log v(\sigma);
        \end{equation}
\item Additivity of the quantum R\'enyi mutual information~\cite{hayashi2016correlation}:
        for $\rho_{AB}\in\mc{S}(AB)$, $\sigma_{A'B'} \in \mc{S}(A'B')$, we have
        \begin{equation}
        \widetilde{I}_{\alpha}(AA':BB')_{\rho \ox \sigma}=\widetilde{I}_{\alpha}(A:B)_\rho+\widetilde{I}_{\alpha}(A':B')_\sigma,
        \end{equation}
   for $\alpha \in [\frac{1}{2},+\infty)$;  
\item  Variational expression~\cite{MosonyiOgawa2017strong}:
      when $\rho$ commutes with $\sigma$, we have
      \begin{equation}
          \widetilde{D}_{\alpha}(\rho \| \sigma)= \begin{cases}
         \min\limits_{\tau \in \mc{S}_\rho(\mc{H})} \big\{D(\tau \| \sigma)
         -\frac{\alpha}{\alpha-1}D(\tau \| \rho)\big\}, & \alpha \in (0,1), \\
         \max\limits_{\tau \in \mc{S}_\rho(\mc{H})} \big\{D(\tau \| \sigma)
         -\frac{\alpha}{\alpha-1}D(\tau \| \rho)\big\}, & \alpha \in (1,+\infty),
      \end{cases} 
      \end{equation}
      where $\mc{S}_\rho(\mc{H})=\{\tau~|~\tau \in \mc{P}(\mc{H}),~\supp(\tau) \subseteq \supp(\rho),~\tau~commutes~with~
      \rho~and~\sigma\}$.
\item Additivity of the sandwihced R\'enyi $\alpha$-channel capacity~\cite{WWY2014strong}:
        for any classical-quantum channels $W_{X \rightarrow}$ and $T_{S \rightarrow Z}$, we have
        \begin{equation}
        \widetilde{I}_{\alpha}(W_{X\rightarrow Y} \otimes T_{S \rightarrow Z} )=\widetilde{I}_{\alpha}(W_{X\rightarrow Y} )+\widetilde{I}_{\alpha}(T_{S \rightarrow Z} )
        \end{equation}
   for $\alpha \in [\frac{1}{2},+\infty)$;  
\end{enumerate}
\end{proposition}
For $S_n$ being the symmetric group of the permutations of $n$ elements, we define  the set of symmetric states on system $A^n$ as
\beq
\mc{S}_{\rm{sym}}(A^n)=\left\{\sigma_{A^n} \in \mc{S}(A^n)~|~
W^{\pi}_{A^n} \sigma_{A^n}W^{\pi * }_{A^n}=\sigma_{A^n}, \ \forall\ \pi \in S_n\right\}.
\eeq
where $W^{\pi}_{A^n}: \ket{\psi_1} \ox \ldots \ox \ket{\psi_n}
\mapsto  \ket{\psi_{\pi^{-1}(1)}} \ox \ldots \ox \ket{\psi_{\pi^{-1}(n)}}$ is the natural representation of $\pi \in S_n$. There exists a  symmetric state that dominates all the other symmetric states, as stated in the following Lemma~\ref{lem:sym}.  The paper~\cite{Hayashi2009universal} and~\cite{CKR2009postselection} give two different constructions of this symmetric state, respectively. 
\begin{lemma}
\label{lem:sym}
For a finite-dimensional system $A$ and any $n\in\mathbb{N}$, there exists a  symmetric state $\sigma_{A^n}^u \in \mc{S}_{\rm{sym}}(A^n)$ such that every symmetric state $\sigma_{A^n} \in \mc{S}_{\rm{sym}}(A^n)$ is dominated as
\begin{equation}
\sigma_{A^n} \leq v_{n,|A|}\sigma_{A^n}^u,
\end{equation}
where $v_{n,|A|} \leq (n+1)^{\frac{(|A|+2)(|A|-1)}{2}}$ is a polynomial of $n$. The number of distinct eigenvalues of $\sigma_{A^n}^u$ is upper bounded by $v_{n,|A|}$ as well.
\end{lemma}


\section{Problem formulation}\label{sec:problem-form}

\paragraph{Classical channel interconversion.} Given two classical channels $W_{Y|X}$ and $T_{Z|S}$, a shared-randomness scheme to convert the channel $W$ into the channel $T$ consists of an encoding channel $\cE : \cS\times \cR \rightarrow \cX$, a 
decoding channel $\cD : \cY\times \cR \rightarrow \cZ$ and a shared random variable $P_R$. The synthesized conversion channel $\widetilde{T}$ is then 
\begin{align}
    \widetilde{T}(z|s) = \sum_{r\in \cR} P_R(r) \sum_{x\in \cX}\sum_{y\in \cY} \cE(x|s,r) W(y|x) \cD(z|y,r).
\end{align}
We denote the set of the shared randomness schemes by $\mathrm{SR}(W\rightarrow T)$.

Similarly, we can define the set of the non-signaling schemes $\mathrm{NS}(W\rightarrow T)$ which consists of non signaling channels $N: \cS\times \cY \rightarrow \cX \times \cZ$ satisfying:
\begin{align}
\sum_{x \in \mathcal{X}} N(x,z|s,y)&=N(z|y),  &\forall z\in \cZ,\; \forall s\in \cS, \; \forall y \in \cY, \\
\sum_{z \in \mathcal{Z}} N(x,z|s,y)&=N(x|s), &\forall x\in \cX,\; \forall s \in \cS,\; \forall y\in \cY.
\end{align}
The synthesized conversion channel is denoted:
\begin{align}
   (N\circ W)(z|s)=  \sum_{x\in \cX}\sum_{y\in \cY} N(x,z|s,y) W(y|x).
\end{align}
$N$ is said to have marginal $P_X$ if $N_{X|S=s} = P_X$ for all $s\in \cS$. 

For some specified distance ‘$\mathrm{dist}$', the optimal performance among all conversion schemes in some set $\cA$ is then defined as:
\begin{align}
     \mathrm{Succ}^{\cA}_{\mathrm{dist}}(W\rightarrow T) = \sup_{N\in \cA} 1-\mathrm{dist}(N\circ W, T).
\end{align}
We will focus on the total-variation distance and the purified distance. The total-variation distance between two channels $T$ and $T'$ is defined as follows:
\begin{align}
    \operatorname{TV}(T',T) = \sup_{s\in \cS} \operatorname{TV}(T'_s,T_s).
\end{align}
Similarly, the purified distance between two channels $T$ and $T'$ is defined as follows:
\begin{align}
    \operatorname{Pur}(T',T) = \sup_{s\in \cS} \operatorname{Pur}(T'_s,T_s).
\end{align}

 Shannon's theorem~\cite{shannon1948mathematical} and the reverse Shannon's theorem~\cite{BSST2002entanglement} imply that the ratio of the channel capacities $C(T)$ and $C(W)$ determines exactly when asymptotic channel interconversion is possible.  When the interconversion rate $r>\frac{C(W)}{C(T)}$, the strong converse property holds, i.e., the optimal performance $ \mathrm{Succ}_{\mathrm{dist}}^{\mathrm{SR}}(W^{\times n}\rightarrow T^{\times \ceil{rn}})$ or $ \mathrm{Succ}_{\mathrm{dist}}^{\mathrm{NS}}(W^{\times n}\rightarrow T^{\times \ceil{rn}})$ converges to $0$ as $n$ grows to infinity. The exact exponential decay rate is called the strong converse exponent of shared randomness classical channel interconversion~(resp. non-signaling assisted classical channel interconversion),
\begin{align*}
       & \lim_{n \rightarrow \infty} -\frac{1}{n} \log \mathrm{Succ}_{\mathrm{dist}}^{\mathrm{SR}}(W^{\times n}\rightarrow T^{\times \ceil{rn}}) \qquad  \left(\text{resp. }
        \lim_{n \rightarrow \infty} -\frac{1}{n} \log \mathrm{Succ}_{\mathrm{dist}}^{\mathrm{NS}}(W^{\times n}\rightarrow T^{\times \ceil{rn}}) \right).
\end{align*}
The goal of this paper is to study these strong converse exponents. 

\paragraph{Classical-quantum channel interconversion.} Let $W_{X \rightarrow Y}$ and $T_{S \rightarrow Z}$ be two quantum channels from Alice to Bob, an entanglement-assisted  scheme to convert  $W$ into $T$ consists of using a shared bipartite 
entangled state $\phi_{\tilde{S}\tilde{Y}}$, applying local quantum channel $\mathcal{E}_{S\tilde{S} \rightarrow X}$ at Alice's side, sending system $X$ from Alice to Bob through $W_{X \rightarrow Y}$, and at last applying local quantum channel
$\mathcal{D}_{Y\tilde{Y} \rightarrow Z}$ at Bob's side, such a  scheme is usually denoted as a triple $\Pi=\{\phi_{\tilde{S}\tilde{Y}},\mathcal{E}_{S\tilde{S} \rightarrow X},\mathcal{D}_{Y\tilde{Y} \rightarrow Z}\}$. The performance of 
this scheme is measured by the purified distance between the final channel $\Pi \circ W_{X \rightarrow Y}=\mathcal{D}_{Y\tilde{Y} \rightarrow Z} \circ W_{X \rightarrow Y}\circ \mathcal{E}_{S\tilde{S} \rightarrow X}((\cdot) \otimes \phi_{\tilde{S}\tilde{Y}})$ and $T_{S \rightarrow Z}$. The optimal performance among all entanglement-assisted conversion schemes is defined as
\[
 \mathrm{Succ}_{\mathrm{Pur}}^{\mathrm{EA}}(W\rightarrow T)=\sup_{\Pi \in \mathrm{EA}(W \rightarrow T)} 1-\mathrm{Pur}(\Pi \circ W, T),
\]
where $\mathrm{EA}(W \rightarrow T)$ is the set of all entanglement-assisted schemes for converting $W$ to $T$.

Similarly, a non-signaling scheme that converts $W$ to $T$ consists of a super channel $\Pi_{SY \rightarrow XZ}$ which satisfies the following restrictions:
\begin{align*}
    \begin{split}
   &J_{\Pi} \geq 0, ~\mathrm{Tr}_{XZ} J_{\Pi}=I_{SY} \\
   &\mathrm{Tr}_{X}= \frac{I_S}{|S|}\mathrm{Tr}_{XS}  J_{\Pi} \\
   &\mathrm{Tr}_{Z}= \frac{I_Y}{|Y|}\mathrm{Tr}_{ZY}  J_{\Pi}.
    \end{split}
\end{align*}
 The performance of 
this scheme is measured by the purified distance between the final channel $\Pi \circ W_{X \rightarrow Y}$ and $T_{S \rightarrow Z}$. The optimal performance among all non-signaling assisted conversion schemes is defined as
\[
 \mathrm{Succ}_{\mathrm{Pur}}^{\mathrm{NS}}(W\rightarrow T)=\sup_{\Pi \in \mathrm{NS}(W \rightarrow T)} 1-\mathrm{Pur}(\Pi \circ W, T),
\]
where $\mathrm{NS}(W \rightarrow T)$ is the set of all non-signaling assisted schemes for  converting $W$ to $T$.

The quantum Shannon's second theorem~\cite{BSST2002entanglement} and quantum reverse Shannon's theorem~\cite{BDHSW2014quantum, BCR2011the} imply that for asymptotically perfect channel interconversion, a conversion rate which is equal to the ratio between their entanglement-assisted classical capacity $C_E(W_{X \rightarrow Y})$ and
$C_E(T_{S \rightarrow Z})$ is necessary and sufficient. When the interconversion rate $r>\frac{C_E(W_{X\rightarrow Y})}{C_E(T_{S \rightarrow Z})}$, the strong converse property holds, i.e., the optimal performance $ \mathrm{Succ}_{\mathrm{Pur}}^{\mathrm{EA}}(W^{\otimes n}\rightarrow T^{\otimes \ceil{rn}})$ or $ \mathrm{Succ}_{\mathrm{Pur}}^{\mathrm{NS}}(W^{\otimes n}\rightarrow T^{\otimes \ceil{rn}})$ converges to $0$ as $n$ grows to infinity. The exact exponential decay rate is called the strong converse exponent of entanglement-assisted classical-quantum channel interconversion~(resp. non-signaling assisted classical-quantum channel interconversion),
\begin{align*}
       & \lim_{n \rightarrow \infty} -\frac{1}{n} \log \mathrm{Succ}_{\mathrm{Pur}}^{\mathrm{EA}}(W^{\otimes n}\rightarrow T^{\otimes \ceil{rn}}) \qquad  \left(\text{resp. }
        \lim_{n \rightarrow \infty} -\frac{1}{n} \log \mathrm{Succ}_{\mathrm{Pur}}^{\mathrm{NS}}(W^{\otimes n}\rightarrow T^{\otimes \ceil{rn}}) \right).
\end{align*}
The goal of this paper is to analyze these strong converse exponents. 


\section{Classical channels}\label{sec:classical}

In this section, we characterize the strong converse exponent for channel interconversion under the worst-case total variation and purified distance metrics.  

\begin{theorem}\label{thm:classical}
Let $W_{Y|X}$ and $T_{Z|S}$ be two classical channels over finite alphabets. We have that for all rates $r\ge 0$:  
\begin{align}
\lim_{n \rightarrow \infty}-\frac{1}{n}\log \mathrm{Succ}_{\rm{TV}}^{\cA}\left( W^{\times n} \rightarrow  T^{\times \ceil{rn}}\right)
&= \sup_{0< \alpha < 1}\sup_{1<p<\frac{1}{1-\alpha}}	\frac{1-p+\alpha p}{p}\left(r\cdot I_{(1-\alpha)p}(T)-I_{\frac{\alpha p}{p-1}}(W)\right),
\\\lim_{n \rightarrow \infty}-\frac{1}{n}\log \mathrm{Succ}_{\rm{Pur}}^{\cA}\left( W^{\times n} \rightarrow  T^{\times \ceil{rn}}\right)
&=\sup_{1< p < 2}\frac{2-p}{p}\left(r\cdot I_{\frac{p}{2}}(T)-I_{\frac{p}{2(p-1)}}(W)\right), \label{eq:pur}
\end{align}
where $\cA \in \{\mathrm{SR, EA, NS}\}$.  
\end{theorem}

The proof of this theorem has two parts: the converse and the achievability bounds. The
converse is proven for  non-signaling  strategies, which subsume the other forms of assistance. The achievability uses a shared-randomness strategy, which is the weakest form of assistance we consider in this theorem.

We remark that our proof methods can be used to prove the strong converse exponent under the most general R\'enyi divergence as a notion of measure. The $\alpha$-R\'enyi based measure between two channels $T' = N\circ W$ and $T$ is defined as follows:
\begin{align}
    \mathtt{d}_{\alpha}(T\|  T') = 1-\inf_{s\in \cS}    \exp(-D_{\alpha}(T_s\|  T'_s)). 
\end{align}
Recently, \cite{Li2024OctLarge} characterizes exponents of channel simulation under this notion of measure. Similarly, we find the strong converse exponent under $\mathtt{d}_{\alpha}$ for $\alpha \in (0, 1)$.
\begin{theorem}\label{thm:SCE-Renyi}
    Let $W_{Y|X}$ and $T_{Z|S}$ be two classical channels over finite alphabets. We have that for all $\alpha \in (0,1)$ and all rates  $r\ge 0$:  
    \begin{align}
\lim_{n \rightarrow \infty}-\frac{1}{n}\log \mathrm{Succ}_{\mathtt{d}_\alpha}^{\cA}\left( W^{\times n} \rightarrow  T^{\times \ceil{rn}}\right)
&=\sup_{1<p<\frac{1}{\alpha}}	\frac{1-\alpha p}{(1-\alpha) p}\left(r\cdot I_{\alpha p}(T)-I_{\frac{(1-\alpha) p}{p-1}}(W)\right), 
\end{align}
where $\cA\in \{\rm{SR, EA, NS}\}$.
\end{theorem}


\subsection{Converse}

In this section, we prove the converse  of the strong converse exponent for channel interconversion. Mainly, we prove the lower bound inequalities of Theorems~\ref{thm:classical} \& \ref{thm:SCE-Renyi} with  NS assistance. To that end, we  first prove the following proposition based on Hölder's inequality. 

\begin{proposition}\label{prop:classical-converse-u}
Let $W_{Y|X}$ and $T_{Z|S}$ be two classical channels over finite alphabets. We have that for all rates $r\ge 0$,  $n \in \mathbb{N}$, and $\alpha\in (0,1)$:  
\begin{align}
&\inf_{N\in \rm{NS}} -  \frac{1}{n}\log \inf_{\bm{s}^{\ceil{rn}}}\mathrm{Tr}\left[(T_{\bm{s}^{\ceil{rn}}}^{\times \ceil{rn}})^{(1-\alpha)}({N}\circ W^{\times n})_{\bm{s}^{\ceil{rn}}}^{\alpha } \right]  
\ge \sup_{1<p<\frac{1}{1-\alpha}}	\frac{1-p+\alpha p}{p}\left(r\cdot I_{(1-\alpha)p}(T)-I_{\frac{\alpha p}{p-1}}(W)\right).
\end{align}
\end{proposition}
\begin{proof}[Proof of Proposition \ref{prop:classical-converse-u}]
For $N\in \rm{NS}$ and denote  $\widetilde{T}^{ \ceil{rn}} = {N}\circ W^{\times n} $.
Observe that 
\begin{align}  
 &\inf_{N\in \rm{NS}} -  \log \inf_{\bm{s}^{\ceil{rn}}}\mathrm{Tr}\left[(T_{\bm{s}^{\ceil{rn}}}^{\times \ceil{rn}})^{(1-\alpha)}(\widetilde{T}_{\bm{s}^{\ceil{rn}}}^{ \ceil{rn}})^{\alpha } \right]  
= \inf_{N\in \rm{NS}} \sup_{P_{\bm{S}^n}} - \log\mathrm{Tr}\left[P_{\bm{S}^n}(T_{\bm{Z}|\bm{S}}^{\times \ceil{rn}})^{(1-\alpha)}(\widetilde{T}_{\bm{Z}|\bm{S}}^{ \ceil{rn}})^{\alpha } \right].
\end{align}
By Hölder's inequality, we have that for all $1<p<\frac{1}{1-\alpha}$:
\begin{align}
&\log\mathrm{Tr}\left[P_{\bm{S}^n}(T_{\bm{Z}|\bm{S}}^{\times \ceil{rn}})^{(1-\alpha)}(\widetilde{T}_{\bm{Z}|\bm{S}}^{ \ceil{rn}})^{\alpha } \right]  
\\&= \inf_{P_{\bm{Z}^n}}\log\mathrm{Tr}\left[P_{\bm{S}^n}^{\frac{1}{p}}(T_{\bm{Z}|\bm{S}}^{\times \ceil{rn}})^{(1-\alpha)}(P_{\bm{Z}^n})^{\frac{1}{p}-(1-\alpha)} P_{\bm{S}^n}^{\frac{p-1}{p}}(\widetilde{T}_{\bm{Z}|\bm{S}}^{ \ceil{rn}})^{\alpha }(P_{\bm{Z}^n})^{\frac{p-1}{p}-\alpha } \right]  
\\&\le \inf_{P_{\bm{Z}^n}}\log \left(\mathrm{Tr}\left[P_{\bm{S}^n}(T_{\bm{Z}|\bm{S}}^{\times \ceil{rn}})^{(1-\alpha)p}(P_{\bm{Z}^n})^{1-(1-\alpha)p} \right]\right)^{\frac{1}{p}}\left( \mathrm{Tr}\left[P_{\bm{S}^n}(\widetilde{T}_{\bm{Z}|\bm{S}}^{ \ceil{rn}})^{\frac{\alpha p}{p-1}}(P_{\bm{Z}^n})^{1-\frac{\alpha p}{p-1}} \right] \right)^{\frac{p-1}{p}}
\\&=\inf_{P_{\bm{Z}^n}}-\frac{1-p+\alpha p}{p}\left(D_{(1-\alpha)p}\left( P_{\bm{S}^n}T_{\bm{Z}|\bm{S}}^{\times \ceil{rn}}\middle\|P_{\bm{S}^n}P_{\bm{Z}^n}\right)-D_{\frac{\alpha p}{p-1}}\left( P_{\bm{S}^n}\widetilde{T}^{ \ceil{rn}}_{\bm{Z}|\bm{S}}\middle\|P_{\bm{S}^n}P_{\bm{Z}^n}\right)\right).
\end{align}
Moreover, by the data-processing inequality  (Lemma \ref{lem:DPI-SC}) we have that 
\begin{align}
I_{\frac{\alpha p}{p-1}}\left( W_{\bm{Y}|\bm{X}}^{\times n}\right) 
&=\sup_{P_{X^n}}\inf_{P_{Y^n}}D_{\frac{\alpha p}{p-1}}\left(P_{X^n}  W_{\bm{Y}|\bm{X}}^{\times n} \middle\| P_{X^n}P_{Y^n}\right)
\\& \ge \inf_{P_{Y^n}}D_{\frac{\alpha p}{p-1}}\left(P_{\bm{S}^n} \cdot  (N\circ W_{\bm{Y}|\bm{X}}^{\times n})_{Z|S} \middle\| P_{\bm{S}^n}\cdot (N\circ P_{Y^n})_{Z|S}\right)
\\&\ge \inf_{P_{\bm{Z}^n}}D_{\frac{\alpha p}{p-1}}\left( P_{\bm{S}^n}\widetilde{T}^{ \ceil{rn}}_{\bm{Z}|\bm{S}}\middle\|P_{\bm{S}^n}P_{\bm{Z}^n}\right).
\end{align}
where we used that $(N\circ P_{Y^n})_{Z|S} = (N\circ P_{Y^n})_{Z}$ since $N$ is non-signaling. Therefore 
\begin{align}
&
\inf_{N\in \rm{NS}} -  \log \inf_{\bm{s}^{\ceil{rn}}}\mathrm{Tr}\left[(T_{\bm{s}^{\ceil{rn}}}^{\times \ceil{rn}})^{(1-\alpha)}(\widetilde{T}_{\bm{s}^{\ceil{rn}}}^{\ceil{rn}})^{\alpha } \right]  
\\&= \inf_{N\in \rm{NS}} \sup_{P_{\bm{S}^n}} - \log\mathrm{Tr}\left[P_{\bm{S}^n}(T_{\bm{Z}|\bm{S}}^{\times \ceil{rn}})^{(1-\alpha)}(\widetilde{T}_{\bm{Z}|\bm{S}}^{ \ceil{rn}})^{\alpha } \right]
\\&\ge \inf_{N\in \rm{NS}} \sup_{P_{\bm{S}^n}} \sup_{P_{\bm{Z}^n}}\frac{1-p+\alpha p}{p}\left(D_{(1-\alpha)p}\left( P_{\bm{S}^n}T_{\bm{Z}|\bm{S}}^{\times \ceil{rn}}\middle\|P_{\bm{S}^n}P_{\bm{Z}^n}\right)-D_{\frac{\alpha p}{p-1}}\left( P_{\bm{S}^n}\widetilde{T}^{ \ceil{rn}}_{\bm{Z}|\bm{S}}\middle\|P_{\bm{S}^n}P_{\bm{Z}^n}\right)\right)
\\&\ge \inf_{N\in \rm{NS}} \sup_{P_{\bm{S}^n}} \frac{1-p+\alpha p}{p}\left(\inf_{P_{\bm{Z}^n}}D_{(1-\alpha)p}\left( P_{\bm{S}^n}T_{\bm{Z}|\bm{S}}^{\times \ceil{rn}}\middle\|P_{\bm{S}^n}P_{\bm{Z}^n}\right)-\inf_{P_{\bm{Z}^n}}D_{\frac{\alpha p}{p-1}}\left( P_{\bm{S}^n}\widetilde{T}^{ \ceil{rn}}_{\bm{Z}|\bm{S}}\middle\|P_{\bm{S}^n}P_{\bm{Z}^n}\right)\right)
\\&\ge \inf_{N\in \rm{NS}} \sup_{P_{\bm{S}^n}} \frac{1-p+\alpha p}{p}\left(\inf_{P_{\bm{Z}^n}}D_{(1-\alpha)p}\left( P_{\bm{S}^n}T_{\bm{Z}|\bm{S}}^{\times \ceil{rn}}\middle\|P_{\bm{S}^n}P_{\bm{Z}^n}\right)-I_{\frac{\alpha p}{p-1}}\left( W_{\bm{Y}|\bm{X}}^{\times n}\right)\right)
\\&=  \frac{1-p+\alpha p}{p}\left(I_{(1-\alpha)p}\left( T_{\bm{Z}|\bm{S}}^{\times \ceil{rn}}\right)-I_{\frac{\alpha p}{p-1}}\left( W_{\bm{Y}|\bm{X}}^{\times n}\right)\right).
\end{align}
Using the additivity of the $\alpha$-mutual information (see e.g.~\cite{Gallager}), we obtain
\begin{align}
\inf_{N\in \rm{NS}} -  \frac{1}{n}\log \inf_{\bm{s}^{\ceil{rn}}}\mathrm{Tr}\left[(T_{\bm{s}^{\ceil{rn}}}^{\times \ceil{rn}})^{(1-\alpha)}(\widetilde{T}_{\bm{s}^{\ceil{rn}}}^{ \ceil{rn}})^{\alpha } \right]  
&\ge\frac{1}{n} \cdot \frac{1-p+\alpha p}{p}\left(I_{(1-\alpha)p}\left( T_{\bm{Z}|\bm{S}}^{\times \ceil{rn}}\right)-I_{\frac{\alpha p}{p-1}}\left( W_{\bm{Y}|\bm{X}}^{\times n}\right)\right)
\\&=   \frac{1-p+\alpha p}{p}\left(\frac{\ceil{rn}}{n}\cdot I_{(1-\alpha)p}\left( T_{Z|S}\right)-I_{\frac{\alpha p}{p-1}}\left( W_{Y|X}\right)\right).
\end{align}
Since this inequality holds for all  $p\in (1, \frac{1}{1-\alpha})$, we conclude that 
\begin{align}
&\inf_{N\in \rm{NS}} -  \frac{1}{n}\log \inf_{\bm{s}^{\ceil{rn}}}\mathrm{Tr}\left[(T_{\bm{s}^{\ceil{rn}}}^{\times \ceil{rn}})^{(1-\alpha)}({N}\circ W^{\times n})_{\bm{s}^{\ceil{rn}}}^{\alpha } \right] 
\ge\sup_{1<p<\frac{1}{1-\alpha}}	\frac{1-p+\alpha p}{p}\left(r\cdot I_{(1-\alpha)p}(T)-I_{\frac{\alpha p}{p-1}}(W)\right).
\end{align}
\end{proof}
Proposition \ref{prop:classical-converse-u} implies  directly  the converse of Theorem \ref{thm:SCE-Renyi}. Moreover, using Proposition \ref{prop:classical-converse-u}, we can prove the converse of Theorem~\ref{thm:classical} with the NS assistance, which we restate:
\begin{proposition}\label{prop:classical-converse}
Let $W_{Y|X}$ and $T_{Z|S}$ be two classical channels over finite alphabets. We have that for all rates $r\ge 0$, for all $n \in \mathbb{N}$:  
\begin{align}
-\frac{1}{n}\log \mathrm{Succ}_{\rm{TV}}^{\rm{NS}}\left( W^{\times n} \rightarrow  T^{\times \ceil{rn}}\right)
&\ge\sup_{0< \alpha < 1}\sup_{1<p<\frac{1}{1-\alpha}}	\frac{1-p+\alpha p}{p}\left(r\cdot I_{(1-\alpha)p}(T)-I_{\frac{\alpha p}{p-1}}(W)\right),
\\ -\frac{1}{n}\log \mathrm{Succ}_{\rm{Pur}}^{\rm{NS}}\left( W^{\times n} \rightarrow  T^{\times \ceil{rn}}\right)
&\ge \sup_{1< p < 2}\frac{2-p}{p}\left(r\cdot I_{\frac{p}{2}}(T)-I_{\frac{p}{2(p-1)}}(T)\right).
\end{align}
\end{proposition}
\begin{proof}[Proof of Proposition \ref{prop:classical-converse}]
For $N\in \rm{NS}$ and denote  $\widetilde{T}^{ \ceil{rn}} = {N}\circ W^{\times n} $. 

\noindent\textbf{TV-distance.}
For all $\alpha \in(0,1)$ we have that:   
\begin{align}
-\frac{1}{n}\log \mathrm{Succ}_{\rm{TV}}^{\rm{NS}}\left( W^{\times n} \rightarrow  T^{\times \ceil{rn}}\right)
&=  \inf_{N\in \rm{NS}}-  \frac{1}{n}\log \inf_{\bm{s}^{\ceil{rn}}}\mathrm{Tr}\left[(T_{\bm{s}^{\ceil{rn}}}^{\times \ceil{rn}})\wedge (\widetilde{T}_{\bm{s}^{\ceil{rn}}}^{\ceil{rn}})\right]  
\\&\ge \inf_{N\in \rm{NS}} -  \frac{1}{n}\log \inf_{\bm{s}^{\ceil{rn}}}\mathrm{Tr}\left[(T_{\bm{s}^{\ceil{rn}}}^{\times \ceil{rn}})^{(1-\alpha)}(\widetilde{T}_{\bm{s}^{\ceil{rn}}}^{ \ceil{rn}})^{\alpha } \right]  
\\&\ge \sup_{1<p<\frac{1}{1-\alpha}}	\frac{1-p+\alpha p}{p}\left(r\cdot I_{(1-\alpha)p}(T)-I_{\frac{\alpha p}{p-1}}(W)\right)
\end{align}
where we used Proposition \ref{prop:classical-converse-u} in the last inequality.

Since this inequality holds for all $\alpha \in (0,1)$, we conclude that 
\begin{align}
&-\frac{1}{n}\log \mathrm{Succ}_{\rm{TV}}^{\rm{NS}}\left( W^{\times n} \rightarrow  T^{\times \ceil{rn}}\right)
\ge\sup_{0< \alpha < 1}\sup_{1<p<\frac{1}{1-\alpha}}	\frac{1-p+\alpha p}{p}\left(r\cdot I_{(1-\alpha)p}(T)-I_{\frac{\alpha p}{p-1}}(W)\right).
\end{align}
\textbf{Purified distance.} Since $1-\operatorname{Pur}(P,Q)\le F(P,Q)$, we have that 
\begin{align}
-\frac{1}{n}\log \mathrm{Succ}_{\rm{Pur}}^{\rm{NS}}\left( W^{\times n} \rightarrow  T^{\times \ceil{rn}}\right)
&=\inf_{N\in \rm{NS}}-\frac{1}{n}\log \inf_{\bm{s}^{\ceil{rn}}}\;  1-\operatorname{Pur}\left( \widetilde{T}^{ \ceil{rn}}_{\bm{s}^{\ceil{rn}}} , T^{\times \ceil{rn}}_{\bm{s}^{\ceil{rn}}}\right)
\\&\ge \inf_{N\in \rm{NS}}-\frac{2}{n}\log \inf_{\bm{s}^{\ceil{rn}}} \tr{ (T^{\times \ceil{rn}}_{\bm{s}^{\ceil{rn}}})^{1/2} (\widetilde{T}^{ \ceil{rn}}_{\bm{s}^{\ceil{rn}}})^{1/2}  } 
\\&\ge \sup_{1<p<2}	\frac{2-p}{p}\left(r\cdot I_{\frac{p}{2}}(T)-I_{\frac{ p}{2(p-1)}}(W)\right)
\end{align}
where we used Proposition \ref{prop:classical-converse-u} in the last inequality with $\alpha = \frac{1}{2}$.
\end{proof}


\subsection{Achievability} 

In this section, we move to prove the achievability bound of Theorem \ref{thm:classical}. Both bounds for the total variation  and the purified distances rely on the following improved first order achievability for channel interconversion with an exponentially vanishing conversion error. 
\begin{proposition}\label{prop:1order-improved}
Let $T_{Z|S}$ and $W_{Y|X}$ be two channels over finite alphabets. Let $n\in \mathbb{N}$, $r\ge 0$,    $P_{S}$ be a type on $\cS^{\ceil{rn}}$ and $P_{X}\in \cP(\cX)$.  Let $\delta_{n} = 2\frac{\sqrt{\log(n) r \var({T}_{Z|S} )} }{n^{1/4}}$ and  $\gamma_{n} =2\frac{\sqrt{\log(n) \var( W_{Y|X})} }{n^{1/4}}$   such that:
\begin{align}
r\cdot I(P_{S}, {T}_{Z|S} ) +\delta_{n}\le I(P_{X}, {W}_{Y|X}) - \gamma_{n}.
\end{align}
Then, there is a shared-randomness strategy converting $W^{\times n}$ to $T^{\times \ceil{rn}}$ under any input of type close to $P_{S}$ and achieving a conversion TV-error at most $\exp(-\sqrt{n}\log(n))$. That is, for all $A\ge 0$, there is  $n_{0}(A)$ such that for all $n\ge n_{0}$:
\begin{align}
\inf_{N\in \mathrm{SR}} \sup_{\substack{\bm{s}^{\ceil{rn}}\sim P_{S}' \\ \mathrm{TV}(P_{S}, P_{S}')\le \frac{A}{n}}} \mathrm{TV}\left(({N}\circ W^{\times n})_{\bm{s}^{\ceil{rn}}} , T^{\times \ceil{rn}}_{\bm{s}^{\ceil{rn}}}\right) \le \exp(-\sqrt{n}\log(n)).
\end{align}
Furthermore, $N$ has marginal $P_{\bm{X}}^{\times n}$.
\end{proposition}

The proof of this result is detailed in Appendix \ref{app:1order-improved}. Additionally, as we will see later, both notion of distances involve a strong converse exponent that can be written using  variational formulas as shown in the following proposition.
\begin{proposition}\label{prop:ach-var-u}
Let $\mathbb{A} \subset [0,1]$ be a convex set and let $W_{Y|X}$ and $T_{Z|S}$ be two classical channels over finite alphabets. 
We have that:
\begin{align}
 &\sup_{\alpha\in \mathbb{A}}\sup_{1<p<\frac{1}{1-\alpha}}	\frac{1-p+\alpha p}{p}\left(r\cdot I_{(1-\alpha)p}(T)-I_{\frac{\alpha p}{p-1}}(W)\right).
 \\&=\inf_{P_{X}} \sup_{P_{S}} \inf_{\widetilde{W}_{Y|X}} \inf_{\widetilde{T}_{Z|S} }\sup_{\alpha\in \mathbb{A}} \;\; \alpha D\left(P_{X}\widetilde{W}_{Y|X}\middle\|  P_{X} W_{Y|X}\right) +  r(1-\alpha)D\left( P_{S}\widetilde{T}_{Z|S}\middle\| P_{S} T_{Z|S}\right) \\&\qquad\qquad\qquad\qquad\qquad +\alpha\left|r \cdot  I(P_{S}, \widetilde{T}_{Z|S} )-  I(P_{X}, \widetilde{W}_{Y|X} ) \right|_{+}. \label{eq:var-form}
\end{align}
\end{proposition}
The proof of this proposition can be found in Appendix \ref{app:variational}. This proposition reduces the problem of proving the achievability part of the strong converse to proving the upper bound using the variational formulas \eqref{eq:var-form}.  This can be done by considering  cases where the rate satisfies $r \cdot  I(P_{S}, \widetilde{T}_{Z|S} ) \le   I(P_{X}, \widetilde{W}_{Y|X} )$ or not for arbitrary fixed channels $\widetilde{T}_{Z|S}$ and $\widetilde{W}_{Y|X}$. Observe that this condition  resembles the condition of Proposition \ref{prop:1order-improved} which guarantees that  a channel interconversion strategy exists.  Although the high-level proof strategy is the same and relies on Propositions \ref{prop:1order-improved} and \ref{prop:ach-var-u}, we need to separate the proofs for the total variation distance and the purified distance.


\subsubsection{Total variation distance}

We start with the TV distance and prove the following proposition. 
\begin{proposition}\label{prop:ach-var-tv}
Let $W_{Y|X}$ and $T_{Z|S}$ be two classical channels over finite alphabets. We have that for all rates $r\ge 0$:  
\begin{align}
\lim_{n \rightarrow \infty}-\frac{1}{n}\log \mathrm{Succ}_{\rm{TV}}^{\rm{SR}}\left( W^{\times n} \rightarrow  T^{\times \ceil{rn}}\right)
&\le \inf_{P_{X}} \sup_{P_{S}} \inf_{\widetilde{W}_{Y|X}} \inf_{\widetilde{T}_{Z|S} }\sup_{0 < \alpha <1} \alpha\left|r \cdot  I(P_{S}, \widetilde{T}_{Z|S} )-  I(P_{X}, \widetilde{W}_{Y|X} ) \right|_{+}
\\&\;\;+\alpha  D\left(P_{X}\widetilde{W}_{Y|X}\middle\|  P_{X} W_{Y|X}\right) +  r(1-\alpha)D\left( P_{S}\widetilde{T}_{Z|S}\middle\| P_{S} T_{Z|S}\right),
\end{align}
where the minimization is over channels $\widetilde{W}_{Y|X}$ and $\widetilde{T}_{Z|S}$. 
\end{proposition}

Combining Proposition \ref{prop:ach-var-tv} and Proposition \ref{prop:ach-var-u} with $\mathbb{A}=(0,1)$ gives the achievability of Theorem \ref{thm:classical} for the TV distance. To prove Proposition \ref{prop:ach-var-tv}, the following inequality proves to be useful (see Appendix \ref{sec:lemmas} for a proof). 
\begin{lemma}\label{lem:reverse-chernoff}
    Let $P, Q \in \cP(\cX)$. We have that for all $\eta\in (0,\frac{1}{2})$:
    \begin{align}
        \tr{P \wedge Q} \ge (1-2\eta) \sup_{ V \in \cP(\cX)}\inf_{0 < \alpha < 1} \exp\left(-\alpha D_{\max}^{\eta}(V \| P) -(1-\alpha) D_{\max}^{\eta}(V \| Q) \right).
    \end{align}
\end{lemma}

\begin{proof}[Proof of Proposition~\ref{prop:ach-var-tv}]
Observe that by Sion's minimax theorem \cite{sion1958general}:
\begin{align}
\mathrm{Succ}_{\rm{TV}}^{\rm{SR}}\left( W^{\times n} \rightarrow  T^{\times \ceil{rn}}\right)
&= \sup_{N\in \rm{SR}}\inf_{\bm{s}^{\ceil{rn}}\in \cS^{\ceil{rn}}}  1-\mathrm{TV}\left(({N}\circ W^{\times n})_{\bm{s}^{\ceil{rn}}} , T^{\times \ceil{rn}}_{\bm{s}^{\ceil{rn}}}\right) 
 \\&\ge  \sup_{N\in \rm{SR}}\inf_{\bm{s}^{\ceil{rn}}\in \cS^{\ceil{rn}}}  \tr{({N}\circ W^{\times n})_{\bm{s}^{\ceil{rn}}} \wedge T^{\times \ceil{rn}}_{\bm{s}^{\ceil{rn}}}}
   \\ &=   \inf_{ P_{\bm{S}^{ \ceil{rn}}} } \sup_{N\in \rm{SR}}  \mathbb{E}_{\bm{s}^{\ceil{rn}} \sim P_{\bm{S}^{ \ceil{rn}}}} \tr{({N}\circ W^{\times n})_{\bm{s}^{\ceil{rn}}} \wedge T^{\times \ceil{rn}}_{\bm{s}^{\ceil{rn}}}}
. 
\end{align}
Furthermore, we can restrict the minimization on $P_{\bm{S}^{ \ceil{rn}}}$ to be on permutation invariant probability distributions. In this case, there is a type $t\in \cT_{\ceil{rn}}(\cS)$ such that $P_{\bm{S}^{ \ceil{rn}}}\mge (\ceil{rn}+1)^{-|\cS|} \frac{1}{|t|}\sum_{\bm{s}^{\ceil{rn}} \in t} \bid\{\bm{s}^{\ceil{rn}} \in t\}$ (see Appendix \ref{sec-types}). All elements in the type $t$ have the same empirical distribution denoted $P_{S}$. Hence for $\bm{s}^{\ceil{rn}}(t) \sim t$, we have that:
\begin{align}
\mathrm{Succ}_{\rm{TV}}^{\rm{SR}}\left( W^{\times n} \rightarrow  T^{\times \ceil{rn}}\right)
&\ge   \inf_{ P_{\bm{S}^{ \ceil{rn}}} } \sup_{N\in \rm{SR}}  \mathbb{E}_{\bm{s}^{\ceil{rn}} \sim P_{\bm{S}^{ \ceil{rn}}}}\tr{({N}\circ W^{\times n})_{\bm{s}^{\ceil{rn}}} \wedge T^{\times \ceil{rn}}_{\bm{s}^{\ceil{rn}}}}
\\&\ge  (\ceil{rn}+1)^{-|\cS|}  \inf_{ t\in \cT_{\ceil{rn}}(\cS)} \sup_{N\in \rm{SR}}  \tr{({N}\circ W^{\times n})_{\bm{s}^{\ceil{rn}}} \wedge T^{\times \ceil{rn}}_{\bm{s}^{\ceil{rn}}}}.
\end{align}
\sloppy Fixing $t$ and $P_{S}$ (the empirical distribution of $t$), we make cases depending on the comparison between $r \cdot  I(P_{S}, \widetilde{T}_{Z|S} )$ and $ I(P_{X}, \widetilde{W}_{Y|X} ) $. Let  $P_Z$ such that   $I(P_{S}, \widetilde{T}_{Z|S} ) = D(P_{S} \widetilde{T}_{Z|S} \| P_{S}P_Z)$. 
Let $\delta_{n} = \frac{2}{n^{1/4}}\sqrt{\log(n) \var(\widetilde{T}_{Z|S} )} $ and  $\gamma_{n} =\frac{2}{n^{1/4}} \sqrt{\log(n) \var(\widetilde{W}_{Y|X})}$. 

\textbf{Case $1$. $r\cdot I(P_{S}, \widetilde{T}_{Z|S} ) +\delta_{n} \le I(P_{X}, \widetilde{W}_{Y|X}) - \gamma_{n}.$}
			
By Proposition~\ref{prop:1order-improved}, for sufficiently large  $n$, $\widetilde{W}_{\bm{Y}|\bm{X}}^{\times n}$ can be converted to $\widetilde{T}_{\bm{Z}|\bm{S}}^{\times \ceil{rn}}$ up to an error $\exp(-\sqrt{n}\log(n))$ (under  input  $\bm{s}^{\ceil{rn}}\in t$)  using a SR strategy $N$ of marginal $P_{\bm{X}}^{\times n}$. Hence, we have that for constant  $\eta \in (0, \frac{1}{2})$ and $T_{\min} = \min_{s, z}T(z|s)\bid\{T(z|s)>0\}$:
\begin{align}
&-\log  \tr{({N}\circ W^{\times n})_{\bm{s}^{\ceil{rn}}} \wedge T^{\times \ceil{rn}}_{\bm{s}^{\ceil{rn}}}} + \log(1-2\eta)
\\&\overset{(a)}\le \sup_{0< \alpha <1}  \alpha D_{\max}^{\eta}\left(  (N\circ \widetilde{W}^{\times n})_{\bm{s}^{\ceil{rn}}(t)} \middle\|  (N\circ W^{\times n})_{\bm{s}^{\ceil{rn}}(t)}\right) + (1-\alpha) D_{\max}^{\eta}\left(  (N\circ \widetilde{W}^{\times n})_{\bm{s}^{\ceil{rn}}(t)}\middle\|  T^{\times \ceil{rn}}_{\bm{s}^{\ceil{rn}}(t)} \right)
\\&\overset{(b)}\le \sup_{0< \alpha <1}  \alpha D_{\max}^{\eta}\left(  P_{\bm{X}}^{\times n}\widetilde{W}_{\bm{Y}|\bm{X}}^{\times n} \middle\|  P_{\bm{X}}^{\times n} W^{\times n}_{\bm{Y}|\bm{X}}\right) + (1-\alpha) D_{\max}^{\eta}\left(  (N\circ \widetilde{W}^{\times n})_{\bm{s}^{\ceil{rn}}(t)}\middle\|  T^{\times \ceil{rn}}_{\bm{s}^{\ceil{rn}}(t)} \right)
\\&\overset{(c)}\le \sup_{0< \alpha <1} n\alpha D\left(P_{X} \widetilde{W}_{Y|X}\middle\|  P_{X} W_{Y|X}\right) + (1-\alpha) \sum_{t=1}^{\ceil{rn}}D\left(   \widetilde{T}_{s_{t}}\middle\|  T_{s_{t}} \right) +\cO(\sqrt{n}) +(\sqrt{n}+1)e^{-\sqrt{n}\log(n)}T_{\min}^{-\sqrt{n}}
\\&=\sup_{0< \alpha < 1}  n  \alpha D\left(P_{X} \widetilde{W}_{Y|X}\middle\| P_{X} W_{Y|X}\right) +  \ceil{rn}(1-\alpha) D\left( P_{S} \widetilde{T}_{Z|S} \middle\|  P_{S} T_{Z|S} \right) +\cO(\sqrt{n})
\end{align}
where in $(a)$ we used Lemma~\ref{lem:reverse-chernoff}; in $(b)$ we used the data-processing inequality (Lemma \ref{lem:DPI-SC}, note that $N$ is  of marginal $P_{\bm{X}}^{\times n}$); in $(c)$ we used Lemma \ref{lem:iid} and Lemma \ref{lem:al-iid}.

\textbf{Case $2$. $r\cdot I(P_{S}, \widetilde{T}_{Z|S} ) +\delta_{n} \ge I(P_{X}, \widetilde{W}_{Y|X}) - \gamma_{n}.$}

Let $r' \le r $ such that $r'\cdot I(P_{S}, \widetilde{T}_{Z|S} ) +\delta_{n} = I(P_{X}, \widetilde{W}_{Y|X}) - \gamma_{n}$. By Proposition~\ref{prop:1order-improved},   for sufficiently large $n$, $\widetilde{W}_{\bm{Y}|\bm{X}}^{\times n}$ can be converted to $\widetilde{T}_{\bm{Z}|\bm{S}}^{\times \ceil{r'n}}$ up to an error $\exp(-\sqrt{n}\log(n))$  (under any input $\bm{s}^{\ceil{r'n}}$ of type $P_{S}'$ satisfying $\mathrm{TV}(P_{S}',P_{S})\le \cO\left(\frac{1}{n}\right)$) using a SR strategy $N$ of marginal $P_{\bm{X}}^{\times n}$.

Let  $P_Z$ be a probability distribution satisfying  $D(P_{S} \widetilde{T}_{Z|S} \|P_{S}P_Z) = \inf_{P_Z}D(P_{S} \widetilde{T}_{Z|S} \|P_{S}P_Z)$ and  $\bm{s}^{\ceil{r'n}}$ a sub-string of $\bm{s}^{\ceil{rn}}$ of type $P_{S}'$  such that $\mathrm{TV}(P_S', P_S)\le \cO\left(\frac{1}{n}\right)$. Let $\bm{s}^{\ceil{rn}-\ceil{r'n}} = \bm{s}^{\ceil{rn}}\setminus \bm{s}^{\ceil{r'n}}$ and define   
 \begin{align}
 \widetilde{N}(\bm{z}^{\ceil{rn}}\bm{x}^{n}|\bm{y}^{n}\bm{s}^{\ceil{rn}}) = {N}(\bm{z}^{\ceil{r'n}}\bm{x}^{n}|\bm{y}^{n}\bm{s}^{\ceil{r'n}})\times P_{\bm{Z}}^{\times \ceil{rn}-\ceil{r'n}}(\bm{z}^{\ceil{rn}-\ceil{r'n}}). 
\end{align}
Note that $\widetilde{N}$ corresponds to a valid SR strategy. Let $P''_{S}$ be the empirical distribution of $\bm{s}^{\ceil{rn}-\ceil{r'n}} = \bm{s}^{\ceil{rn}}\setminus \bm{s}^{\ceil{r'n}}$, we have also $\mathrm{TV}(P_{S}'',P_{S})\le \cO\left(\frac{1}{n}\right)$. So by the same reasoning as in Case $1$, we obtain for constant $\eta\in (0, \frac{1}{2})$:
 
\begin{align}
&-\log  \tr{(\widetilde{N}\circ W^{\times n})_{\bm{s}^{\ceil{rn}}} \wedge T^{\times \ceil{rn}}_{\bm{s}^{\ceil{rn}}}} + \log(1-2\eta)
\\&\overset{(a)}\le \sup_{0< \alpha <1}  \alpha D_{\max}^{\eta}\left(  (N\circ \widetilde{W}^{\times n})_{\bm{s}^{\ceil{r'n}}} \times \widetilde{T}^{\times \ceil{rn}-\ceil{r'n}}_{\bm{s}^{\ceil{rn}-\ceil{r'n}}}\middle\|  (N\circ W^{\times n})_{\bm{s}^{\ceil{r'n}}}\times P_{\bm{Z}}^{\times \ceil{rn}-\ceil{r'n}}\right) 
\\&\qquad\qquad  + (1-\alpha) D_{\max}^{\eta}\left(  (N\circ \widetilde{W}^{\times n})_{\bm{s}^{\ceil{r'n}}}\times \widetilde{T}^{\times \ceil{rn}-\ceil{r'n}}_{\bm{s}^{\ceil{rn}-\ceil{r'n}}}\middle\|  T^{\times \ceil{rn}}_{\bm{s}^{\ceil{rn}}} \right)
\\&\overset{(b)}\le \sup_{0< \alpha <1}  \alpha D_{\max}^{\eta}\left(  P_{\bm{X}}^{\times n}\widetilde{W}_{\bm{Y}|\bm{X}}^{\times n} \times \widetilde{T}^{\times \ceil{rn}-\ceil{r'n}}_{\bm{s}^{\ceil{rn}-\ceil{r'n}}}\middle\|  P_{\bm{X}}^{\times n} W^{\times n}_{\bm{Y}|\bm{X}}\times P_{\bm{Z}}^{\times \ceil{rn}-\ceil{r'n}}\right) 
\\&\qquad \qquad + (1-\alpha) D_{\max}^{\eta}\left(  (N\circ \widetilde{W}^{\times n})_{\bm{s}^{\ceil{r'n}}}\times \widetilde{T}^{\times \ceil{rn}-\ceil{r'n}}_{\bm{s}^{\ceil{rn}-\ceil{r'n}}}\middle\|  T^{\times \ceil{rn}}_{\bm{s}^{\ceil{rn}}} \right)
\\&\overset{(c)}\le \sup_{0< \alpha <1} n\alpha D\left(P_{X} \widetilde{W}_{Y|X}\middle\|  P_{X} W_{Y|X}\right)  + (\ceil{rn}-\ceil{r'n})\alpha D\left(P_{S}'' \widetilde{T}_{Z|S}\middle\| P_{S}'' P_Z\right)
\\&\qquad \qquad + \ceil{rn}(1-\alpha) D\left( P_{S} \widetilde{T}_{Z|S} \middle\|  P_{S} T_{Z|S} \right)  +\cO(\sqrt{n}) +(\sqrt{n}+1)e^{-\sqrt{n}\log(n)}T_{\min}^{-\sqrt{n}}
\\&\overset{(d)}=\sup_{0< \alpha < 1}  n\alpha D\left(P_{X} \widetilde{W}_{Y|X}\middle\|  P_{X} W_{Y|X}\right) +\ceil{rn}(1-\alpha) D\left( P_{S} \widetilde{T}_{Z|S} \middle\|  P_{S} T_{Z|S} \right) 
\\&\qquad\qquad +\ceil{rn}\alpha D\left(P_{S} \widetilde{T}_{Z|S}\middle\| P_{S} P_Z\right) - n\alpha D\left(P_{X} \widetilde{W}_{Y|X}\middle\|  P_{X} P_{Y}\right) +\cO(n^{3/4}\sqrt{\log(n)})
\end{align}
where in $(a)$ we used Lemma~\ref{lem:reverse-chernoff}; in $(b)$ we used the data-processing inequality (Lemma \ref{lem:DPI-SC}, note that $N$ is  of marginal $P_{\bm{X}}^{\times n}$); in $(c)$ we used Lemma \ref{lem:iid} and Lemma \ref{lem:al-iid}; in $(d)$  we used $\mathrm{TV}(P_{S}'',P_{S})\le \cO\left(\frac{1}{n}\right)$,  $r'\cdot I(P_{S}, \widetilde{T}_{Z|S} ) +\delta_{n} = I(P_{X}, \widetilde{W}_{Y|X}) - \gamma_{n}$ and $n\alpha (\gamma_{n}+\delta_{n}) = \cO(n^{3/4}\sqrt{\log(n)})$ (the variance terms are bounded by constants independent of $n$ \cite[Corollary III.5]{Dupuis2019Jul}).
 
We have the non-asymptotic upper bound in both cases: 
\begin{align}
    &-\frac{1}{n}\log \mathrm{Succ}_{\rm{TV}}^{\rm{SR}}\left( W^{\times n} \rightarrow  T^{\times \ceil{rn}}\right)    \\&\le \sup_{ t\in \cT_{\ceil{rn}}(\cS) }  \inf_{N\in \rm{SR}}-\frac{1}{n}\log  \tr{({N}\circ W^{\times n})_{\bm{s}^{\ceil{rn}}(t)} \wedge T^{\times \ceil{rn}}_{\bm{s}^{\ceil{rn}}(t)}} +\frac{|\cS|\log(\ceil{rn}+1)}{n}
    \\&\le    \sup_{P_{S}} \inf_{\widetilde{T}, \widetilde{W}} \sup_{0< \alpha < 1} \alpha D\left(P_{X}\widetilde{W}_{Y|X}\middle\|  P_{X} W_{Y|X}\right) +  r(1-\alpha)D\left( P_{S}\widetilde{T}_{Z|S}\middle\| P_{S} T_{Z|S}\right) 
    \\&\qquad \qquad\qquad\qquad +\alpha \left|r\cdot I(P_{S}, \widetilde{T}_{Z|S} ) - I(P_{X}, \widetilde{W}_{Y|X})\right|_+
 + \cO\left(\frac{\sqrt{\log(n)}}{n^{1/4}}\right).
\end{align}
Finally, taking the limit $n\rightarrow \infty$ and optimizing over $P_{X}$ we obtain the required bound:
\begin{align}
\lim_{n \rightarrow \infty}-\frac{1}{n}\log \mathrm{Succ}_{\rm{TV}}^{\rm{SR}}\left( W^{\times n} \rightarrow  T^{\times \ceil{rn}}\right)
&\le \inf_{P_{X}} \sup_{P_{S}} \inf_{\widetilde{T}, \widetilde{W} }\sup_{0< \alpha < 1} \alpha\left|r \cdot  I(P_{S}, \widetilde{T}_{Z|S} )-  I(P_{X}, \widetilde{W}_{Y|X} ) \right|_{+}
  \\&\;\;\;+\alpha D\left(P_{X}\widetilde{W}_{Y|X}\middle\|  P_{X} W_{Y|X}\right) +  r(1-\alpha)D\left( P_{S}\widetilde{T}_{Z|S}\middle\| P_{S} T_{Z|S}\right).
\end{align}
\end{proof}


\subsubsection{Purified distance}

Now, we move to the purified distance and prove the following proposition:

\begin{proposition}\label{prop:ach-var-pur}
Let $W_{Y|X}$ and $T_{Z|S}$ be two classical channels over finite alphabets. We have that for all rates $r\ge 0$:  
\begin{align}
&\lim_{n \rightarrow \infty}-\frac{1}{n}\log \mathrm{Succ}_{\rm{Pur}}^{\rm{SR}}\left( W^{\times n} \rightarrow  T^{\times \ceil{rn}}\right)
\\&\le \inf_{P_{X}} \sup_{P_{S}} \inf_{\widetilde{W}_{Y|X}} \inf_{\widetilde{T}_{Z|S} } D\left(P_X\widetilde{W}_{Y|X}\middle\|  p_XW_{Y|X}\right) + r\cdot   D\left(P_S \widetilde{T}_{Z|S}\middle\|P_S T_{Z|S}\right) +\left|r \cdot  I(P_{S}, \widetilde{T}_{Z|S} )-  I(P_{X}, \widetilde{W}_{Y|X} ) \right|_{+},
\end{align}
where the minimization is over channels $\widetilde{W}_{Y|X}$ and $\widetilde{T}_{Z|S}$. 
\end{proposition}

Combining Proposition \ref{prop:ach-var-pur} and Proposition \ref{prop:ach-var-u} with $\mathbb{A}=\{\frac{1}{2}\}$ gives the achievability of Theorem \ref{thm:classical} for the purified distance. 

\begin{proof}[Proof of Proposition~\ref{prop:ach-var-pur}]
Observe that by the inequality $  1-\operatorname{Pur}(P,Q)\ge \frac{1}{2}F(P,Q)$ and Sion's minimax theorem \cite{sion1958general} we have that:
\begin{align}
\mathrm{Succ}_{\rm{Pur}}^{\rm{SR}}\left( W^{\times n} \rightarrow  T^{\times \ceil{rn}}\right)
&= \sup_{N\in \rm{SR}}\inf_{\bm{s}^{\ceil{rn}}\in \cS^{\ceil{rn}}}  1-\operatorname{Pur}\left(({N}\circ W^{\times n})_{\bm{s}^{\ceil{rn}}} , T^{\times \ceil{rn}}_{\bm{s}^{\ceil{rn}}}\right) 
 \\&\ge   \frac{1}{2}\sup_{N\in \rm{SR}}\inf_{\bm{s}^{\ceil{rn}}\in \cS^{\ceil{rn}}}  F\left(({N}\circ W^{\times n})_{\bm{s}^{\ceil{rn}}} , T^{\times \ceil{rn}}_{\bm{s}^{\ceil{rn}}}\right) 
   \\ &=  \frac{1}{2} \inf_{ P_{\bm{S}^{ \ceil{rn}}} } \sup_{N\in \rm{SR}}  \mathbb{E}_{\bm{s}^{\ceil{rn}} \sim P_{\bm{S}^{ \ceil{rn}}}}F\left(({N}\circ W^{\times n})_{\bm{s}^{\ceil{rn}}} , T^{\times \ceil{rn}}_{\bm{s}^{\ceil{rn}}}\right). 
\end{align}
Furthermore, we can restrict the minimization on $P_{\bm{S}^{ \ceil{rn}}}$ to be on permutation invariant probability distributions. In this case, there is a type $t\in \cT_{\ceil{rn}}(\cS)$ such that $P_{\bm{S}^{ \ceil{rn}}}\mge (\ceil{rn}+1)^{-|\cS|} \frac{1}{|t|}\sum_{\bm{s}^{\ceil{rn}} \in t} \bid\{\bm{s}^{\ceil{rn}} \in t\}$ (see Appendix \ref{sec-types}). All elements in the type $t$ have the same empirical distribution denoted $P_{S}$. Hence for $\bm{s}^{\ceil{rn}}(t) \sim t$, we have that:
\begin{align}
\mathrm{Succ}_{\rm{Pur}}^{\rm{SR}}\left( W^{\times n} \rightarrow  T^{\times \ceil{rn}}\right)
&\ge   \frac{1}{2} \inf_{ P_{\bm{S}^{ \ceil{rn}}} } \sup_{N\in \rm{SR}}  \mathbb{E}_{\bm{s}^{\ceil{rn}} \sim P_{\bm{S}^{ \ceil{rn}}}}F\left(({N}\circ W^{\times n})_{\bm{s}^{\ceil{rn}}} , T^{\times \ceil{rn}}_{\bm{s}^{\ceil{rn}}}\right)
\\&\ge  \frac{1}{2}(\ceil{rn}+1)^{-|\cS|}  \inf_{ t\in \cT_{\ceil{rn}}(\cS)} \sup_{N\in \rm{SR}}  F\left(({N}\circ W^{\times n})_{\bm{s}^{\ceil{rn}}(t)} , T^{\times \ceil{rn}}_{\bm{s}^{\ceil{rn}}(t)}\right).
\end{align}
Fixing $t$ and $P_{S}$ (the empirical distribution of $t$), we make cases depending on the comparison between $r \cdot  I(P_{S}, \widetilde{T}_{Z|S} )$ and $ I(P_{X}, \widetilde{W}_{Y|X} ) $. Let  $P_Z$ such that   $I(P_{S}, \widetilde{T}_{Z|S} ) = D(P_{S} \widetilde{T}_{Z|S} \| P_{S}P_Z)$. 

\textbf{Case $1$. $r\cdot I(P_{S}, \widetilde{T}_{Z|S} ) +\delta_{n} \le I(P_{X}, \widetilde{W}_{Y|X}) - \gamma_{n}.$}

Recall the notation $T_{\min} = \min_{s, z}T(z|s)\bid\{T(z|s)>0\}$. 
By Proposition~\ref{prop:1order-improved}, for sufficiently large  $n$, $\widetilde{W}_{\bm{Y}|\bm{X}}^{\times n}$ can be converted to $\widetilde{T}_{\bm{Z}|\bm{S}}^{\times \ceil{rn}}$ up to an error $\exp(-\sqrt{n}\log(n))$ (under  input  $\bm{s}^{\ceil{rn}}\in t$)  using a SR strategy $N$ of marginal $P_{\bm{X}}^{\times n}$. Hence, we have that 
\begin{align}
&-\log  F\left( (N\circ W^{\times n})_{\bm{s}^{\ceil{rn}}(t)} , T^{\times \ceil{rn}}_{\bm{s}^{\ceil{rn}}(t)} \right)
\\&\le  D\left(  (N\circ \widetilde{W}^{\times n})_{\bm{s}^{\ceil{rn}}(t)} \middle\|  (N\circ W^{\times n})_{\bm{s}^{\ceil{rn}}(t)}\right) + D\left(  (N\circ \widetilde{W}^{\times n})_{\bm{s}^{\ceil{rn}}(t)}\middle\|  T^{\times \ceil{rn}}_{\bm{s}^{\ceil{rn}}(t)} \right)
\\&\le  \mathbb{E}_{\bm{x}^{n}\sim P_{\bm{X}}^{\times n}}D\left(\widetilde{W}^{\times n}_{\bm{x}^{n}}\middle\|  W^{\times n}_{\bm{x}^{n}}\right) + D\left(  (N\circ \widetilde{W}^{\times n})_{\bm{s}^{\ceil{rn}}(t)}\middle\|  T^{\times \ceil{rn}}_{\bm{s}^{\ceil{rn}}(t)} \right)
\\&\le n  D\left(P_X\widetilde{W}_{Y|X}\middle\|  p_XW_{Y|X}\right) + \ceil{rn}  D\left(P_S \widetilde{T}_{Z|S}\middle\|P_S T_{Z|S}\right) + \ceil{rn} \exp(-\sqrt{n}\log(n)) \log(T_{\min}^{-1}) +2\log(2),
\end{align}
where we use Lemma~\ref{Lem:fid-D} in the first inequality; we use  the data-processing inequality (Lemma \ref{lem:DPI-SC}, note that $N$ is  of marginal $P_{\bm{X}}^{\times n}$) in the second inequality; in the third inequality we use the continuity bound for the relative entropy \cite[Corollary 5.9]{Bluhm2023May} (see e.g., \href{https://arxiv.org/pdf/2208.00922}{Cor. 5.9}).

\textbf{Case $2$. $r\cdot I(P_{S}, \widetilde{T}_{Z|S} ) +\delta_{n} \ge I(P_{X}, \widetilde{W}_{Y|X}) - \gamma_{n}.$}

Let $r' \le r $ such that $r'\cdot I(P_{S}, \widetilde{T}_{Z|S} ) +\delta_{n} = I(P_{X}, \widetilde{W}_{Y|X}) - \gamma_{n}$. By Proposition~\ref{prop:1order-improved},   for sufficiently large $n$, $\widetilde{W}_{\bm{Y}|\bm{X}}^{\times n}$ can be converted to $\widetilde{T}_{\bm{Z}|\bm{S}}^{\times \ceil{r'n}}$ up to an error $\exp(-\sqrt{n}\log(n))$  (under any input $\bm{s}^{\ceil{r'n}}$ of type $P_{S}'$ satisfying $\mathrm{TV}(P_{S}',P_{S})\le \cO\left(\frac{1}{n}\right)$) using a SR strategy $N$ of marginal $P_{\bm{X}}^{\times n}$.

Let  $P_Z$ be a probability distribution satisfying  $D(P_{S} \widetilde{T}_{Z|S} \|P_{S}P_Z) = \inf_{P_Z}D(P_{S} \widetilde{T}_{Z|S} \|P_{S}P_Z)$ and let  $\bm{s}^{\ceil{r'n}}$ be  a substring of $\bm{s}^{\ceil{rn}}$ of type $P_{S}'$,  such that $\mathrm{TV}(P_S', P_S)\le \cO\left(\frac{1}{n}\right)$. Let $\bm{s}^{\ceil{rn}-\ceil{r'n}} = \bm{s}^{\ceil{rn}}\setminus \bm{s}^{\ceil{r'n}}$ and define   
 \begin{align}
 \widetilde{N}(\bm{z}^{\ceil{rn}}\bm{x}^{n}|\bm{y}^{n}\bm{s}^{\ceil{rn}}) = {N}(\bm{z}^{\ceil{r'n}}\bm{x}^{n}|\bm{y}^{n}\bm{s}^{\ceil{r'n}})\times P_{\bm{Z}}^{\times \ceil{rn}-\ceil{r'n}}(\bm{z}^{\ceil{rn}-\ceil{r'n}}).
\end{align}
Note that $\widetilde{N}$ corresponds to a valid SR strategy. 
 Let $P''_{S}$ be the empirical distribution of $\bm{s}^{\ceil{rn}-\ceil{r'n}} = \bm{s}^{\ceil{rn}}\setminus \bm{s}^{\ceil{r'n}}$, we have also $\mathrm{TV}(P_{S}'',P_{S})\le \cO\left(\frac{1}{n}\right)$. So by the same reasoning as in Case $1$, we obtain 
\begin{align}
&-\log F\left((\widetilde{N}\circ W^{\times n})_{\bm{s}^{\ceil{rn}}(t)} , T^{\times \ceil{rn}}_{\bm{s}^{\ceil{rn}}(t)}\right)
\\&=  -\log F\left( ({N}\circ W^{\times n})_{\bm{s}^{\ceil{r'n}}} , T^{\times \ceil{r'n}}_{\bm{s}^{\ceil{r'n}}}\right)-\log F\left(P_{\bm{Z}}^{\times \ceil{rn}-\ceil{r'n}} , T^{\times \ceil{rn}-\ceil{r'n}}_{\bm{s}^{\ceil{rn}-\ceil{r'n}}}\right)
\\&\overset{(a)}{\le}  D\left( (N\circ \widetilde{W}^{\times n})_{\bm{s}^{\ceil{r'n}}}\middle\| (N\circ W^{\times n})_{\bm{s}^{\ceil{r'n}}}\right) + D\left( (N\circ \widetilde{W}^{\times n})_{\bm{s}^{\ceil{r'n}}}\middle\| T^{\times \ceil{r'n}}_{\bm{s}^{\ceil{r'n}}}\right) 
\\&\qquad +D\left( \widetilde{T}^{\times \ceil{rn}-\ceil{r'n}}_{\bm{s}^{\ceil{rn}-\ceil{r'n}}}\middle\| T^{\times \ceil{rn}-\ceil{r'n}}_{\bm{s}^{\ceil{rn}-\ceil{r'n}}}\right) +D\left( \widetilde{T}^{\times \ceil{rn}-\ceil{r'n}}_{\bm{s}^{\ceil{rn}-\ceil{r'n}}}\middle\| P_{\bm{Z}}^{\times \ceil{rn}-\ceil{r'n}}\right) 
\\&\overset{(b)}{\le} \mathbb{E}_{\bm{x}^{n}\sim P_{\bm{X}}^{\times n}}D\left(\widetilde{W}^{\times n}_{\bm{x}^{n}}\middle\|  W^{\times n}_{\bm{x}^{n}}\right) + D\left(  \widetilde{T}^{\times \ceil{r'n}}_{\bm{s}^{\ceil{r'n}}}\middle\| T^{\times \ceil{r'n}}_{\bm{s}^{\ceil{r'n}}}\right) +\ceil{r'n} \exp(-\sqrt{n}\log(n))\log(T_{\min}^{-1})
\\&\qquad +\mathbb{E}_{s\sim P_{S}''}(\ceil{rn}-\ceil{r'n})D\left( \widetilde{T}_{s}\middle\| T_{s}\right) +(\ceil{rn}-\ceil{r'n})D\left( P_S\widetilde{T}_{Z|S}\middle\| P_S P_Z\right) +2\log(2)
\\&\overset{(c)}{\le} n D\left(P_X\widetilde{W}_{Y|X}\middle\|  p_XW_{Y|X}\right) +  \ceil{r'n}D\left( P_S\widetilde{T}_{Z|S}\middle\| P_ST_{Z|S}\right) +(\ceil{rn}-\ceil{r'n})D\left(P_S\widetilde{T}_{Z|S}\middle\|P_ST_{Z|S}\right) 
\\&\qquad +(\ceil{rn}-\ceil{r'n})D\left( P_S\widetilde{T}_{Z|S}\middle\| P_S P_Z\right) 
+\cO(1) 
\\&\overset{(d)}{\le} n  D\left(P_X\widetilde{W}_{Y|X}\middle\|  p_XW_{Y|X}\right) +  \ceil{rn}D\left( P_S\widetilde{T}_{Z|S}\middle\| P_ST_{Z|S}\right) +\ceil{rn}D\left( P_S\widetilde{T}_{Z|S}\middle\| P_S P_Z\right) 
\\&\qquad -n I(P_{X}, \widetilde{W}_{Y|X})
+ \cO(n^{3/4}\sqrt{\log(n)}) 
\end{align}
where in $(a)$ we used Lemma~\ref{Lem:fid-D}; in $(b)$ we used the data-processing inequality for non-signaling channels (Lemma \ref{lem:DPI-SC}) and the continuity bound of \cite[Corollary 5.9]{Bluhm2023May}; in $(c)$ we used the fact that $\mathrm{TV}(P_S', P_S)\le \cO\left(\frac{1}{n}\right)$ and $\mathrm{TV}(P_{S}'',P_{S})\le \cO\left(\frac{1}{n}\right)$; in $(d)$ we used  the definition of $r'\cdot I(P_{S}, \widetilde{T}_{Z|S} ) +\delta_{n} = I(P_{X}, \widetilde{W}_{Y|X}) - \gamma_{n}$ and $n (\gamma_{n}+\delta_{n}) = \cO(n^{3/4}\sqrt{\log(n)})$ (the variance terms are bounded by constants independent of $n$ \cite[Corollary III.5]{Dupuis2019Jul}).
  
We have the non-asymptotic upper bound in both cases: 
\begin{align}
    &-\frac{1}{n}\log \mathrm{Succ}_{\rm{Pur}}^{\rm{SR}}\left( W^{\times n} \rightarrow  T^{\times \ceil{rn}}\right)    \\&\le \sup_{ t\in \cT_{\ceil{rn}}(\cS) }  \inf_{N\in \rm{SR}}-\frac{1}{n}\log  F\left(({N}\circ W^{\times n})_{\bm{s}^{\ceil{rn}}(t)} , T^{\times \ceil{rn}}_{\bm{s}^{\ceil{rn}}(t)}\right) +\frac{|\cS|\log(\ceil{rn}+1)}{n}+\frac{\log 2}{n}
    \\&\le    \sup_{P_{S}}\!\inf_{\widetilde{T}, \widetilde{W}} D\left(\!P_{X}\widetilde{W}_{Y|X}\middle\|  P_{X} W_{Y|X}\!\right) +  rD\left(\! P_{S}\widetilde{T}_{Z|S}\middle\| P_{S} T_{Z|S}\!\right) +\left|r\cdot I(P_{S}, \widetilde{T}_{Z|S} ) - I(P_{X}, \widetilde{W}_{Y|X})\right|_+
 + \widetilde{\cO}\left(\!n^{-\frac{1}{4}}\!\right).
\end{align}
Finally, taking the limit $n\rightarrow \infty$ and optimizing over $P_{X}$ we obtain the required bound:
\begin{align}
&\lim_{n \rightarrow \infty}-\frac{1}{n}\log \mathrm{Succ}_{\rm{Pur}}^{\rm{SR}}\left( W^{\times n} \rightarrow  T^{\times \ceil{rn}}\right)
\\&\le \inf_{P_{X}} \sup_{P_{S}} \inf_{\widetilde{W}_{Y|X}} \inf_{\widetilde{T}_{Z|S} } D\left(P_X\widetilde{W}_{Y|X}\middle\|  p_XW_{Y|X}\right) + r\cdot   D\left(P_S \widetilde{T}_{Z|S}\middle\|P_S T_{Z|S}\right) +\left|r \cdot  I(P_{S}, \widetilde{T}_{Z|S} )-  I(P_{X}, \widetilde{W}_{Y|X} ) \right|_{+}\!.
\end{align}
\end{proof}
Finally, the achievability result of Theorem \ref{thm:SCE-Renyi} can be proved similarly using the variational property of the R\'enyi divergence \cite{vanErven2014Jun}: 
\begin{align}
   (1-\alpha) D_{\alpha}(P\|Q) = \inf_{R\in \cP(\cX)} \alpha D(R\|P) + (1-\alpha) D(R\|Q). 
\end{align}


\section{Classical-quantum channels}\label{sec:CQ}

In this section, we determine the strong converse exponents of entanglement-assisted and non-signaling assisted  classical-quantum channel interconversion. The main result of this section is stated in the following theorem.

\begin{theorem}
\label{thm:maincq}
For any classical-quantum channels $W_{X \rightarrow Y}$, $T_{S \rightarrow Z}$ and real $r \geq 0$, we have
\begin{equation}
\lim_{n \rightarrow \infty} -\frac{1}{n} \log \mathrm{Succ}_{\mathrm{Pur}}^{\cA}(W^{\otimes n}\rightarrow T^{\otimes  \ceil{rn}})
=\sup_{1<p<2} \frac{2-p}{p}\big(r\widetilde{I}_{\frac{p}{2}}(T_{S\rightarrow Z})-\widetilde{I}_{\frac{p}{2(p-1)}}(W_{X\rightarrow Y})   \big),
\end{equation}
where $\cA \in \{\mathrm{EA, NS}\}$. 
\end{theorem}

We divide the proof of this theorem into the converse part and the achievability part. In the determination of the converse part, we mainly use the operator  Hölder's inequality. In the establishment of the achievability part, we first derive an upper bound in terms of the log-Euclidean R\'enyi divergence and then improve it to the sandwiched one by using the pinching technique.


\subsection{Converse}

In this section, we prove the converse  of the strong converse exponents of entanglement-assisted and non-signaling assisted classical-quantum channel interconversion.

\begin{proposition}
\label{prop:maincqcon}
For any classical-quantum channels $W_{X \rightarrow Y}$, $T_{S \rightarrow Z}$ and $r \geq 0$, we have
\begin{equation}
\begin{split}
\lim_{n \rightarrow \infty} -\frac{1}{n} \log \mathrm{Succ}_{\mathrm{Pur}}^{\mathrm{EA}}(W^{\otimes n}\rightarrow T^{\otimes  \ceil{rn}}) \geq & \lim_{n \rightarrow \infty} -\frac{1}{n} \log \mathrm{Succ}_{\mathrm{Pur}}^{\mathrm{NS}}(W^{\otimes n}\rightarrow T^{\otimes  \ceil{rn}}) \\
\geq &\sup_{1<p<2} \frac{2-p}{p}\big(r\widetilde{I}_{\frac{p}{2}}(T_{S\rightarrow Z})-\widetilde{I}_{\frac{p}{2(p-1)}}(W_{X\rightarrow Y})   \big).
\end{split}
\end{equation}
\end{proposition}

\begin{proof}
By definition, the first inequality holds trivially. Hence, we only need to prove the second inequality. For any non-signaling assisted conversion scheme $\Pi_n$  converting  $W^{\otimes n}_{X \rightarrow Y}$ to $T^{\otimes \ceil{rn}}_{S \rightarrow Z}$, quantum states $\sigma_{Y^n}$ and $\rho_{S^{\ceil{rn}}R_n}$, $1<p<2$, we let $\mathcal{R}^{\sigma_{Y^n}}_{X^n\rightarrow Y^n}$  represent the replacer channel associated with $\sigma_{Y^n}$. Then, we have
\begin{equation}
\label{equ:cqcon}
\begin{split}
&\frac{p}{2-p}\log F(\Pi_n \circ W^{\otimes n}(\rho_{S^{\ceil{rn}}R_n} ), T^{\otimes \ceil{rn}}(\rho_{S^{\lceil nr \rceil} R_n})) \\
\overset{(a)}\leq &-\widetilde{D}_{\frac{p}{2}}(T^{\otimes \ceil{rn}}(\rho_{S^{\ceil{rn}}R_n})\| \Pi_n \circ \mathcal{R}^{\sigma_{Y^n}}(\rho_{S^{\ceil{rn}}R_n}))+\widetilde{D}_{\frac{p}{2(p-1)}}(\Pi_n \circ W^{\otimes n}(\rho_{S^{\ceil{rn}}R_n} )\|\Pi_n \circ \mathcal{R}^{\sigma_{Y^n}}(\rho_{S^{\ceil{rn}}R_n})) \\
\overset{(b)}= &-\widetilde{D}_{\frac{p}{2}}(T^{\otimes \ceil{rn}}(\rho_{S^{\ceil{rn}}R_n})\| \rho_{R_n} \otimes \Pi_n \circ \mathcal{R}^{\sigma_{Y^n}}(\rho_{S^{\ceil{rn}}}))+\widetilde{D}_{\frac{p}{2(p-1)}}(\Pi_n \circ W^{\otimes n}(\rho_{S^{\ceil{rn}}R_n} )\|\Pi_n \circ \mathcal{R}^{\sigma_{Y^n}}(\rho_{S^{\ceil{rn}}R_n})) \\
\leq &-\widetilde{I}_{\frac{p}{2}}(R_n:Z^{\ceil{rn}})_{T^{\otimes \ceil{rn}}(\rho_{S^{\ceil{rn}}R_n})}+\sup_{\tau_{S^{\ceil{rn}}R_n}}\widetilde{D}_{\frac{p}{2(p-1)}}(\Pi_n \circ W^{\otimes n}(\tau_{S^{\ceil{rn}}R_n} )\|\Pi_n \circ \mathcal{R}^{\sigma_{Y^n}}(\tau_{S^{\ceil{rn}}R_n})) \\
\overset{(c)}\leq&-\widetilde{I}_{\frac{p}{2}}(R_n:Z^{\ceil{rn}})_{T^{\otimes \ceil{rn}}(\rho_{S^{\ceil{rn}}R_n})}+\sup_{\tau_{X^{n}R_n}}\widetilde{D}_{\frac{p}{2(p-1)}}( W^{\otimes n}(\tau_{X^{n}R_n} )\|\mathcal{R}^{\sigma_{Y^n}}(\tau_{X^{n}R_n})) ,
\end{split}
\end{equation}
where  (a) is due to Lemma~\ref{lem:hof} in Appendix,  (b) uses  the fact that a  non-signaling superchannel sends a replacer channel to a replacer channel, (c) follows from the data processing inequality~\cite{wang2019resource}.

Because Eq.~(\ref{equ:cqcon}) holds for any non-signaling assisted conversion schemes, $\sigma_{Y^n}$ and input state $\rho_{S^{\ceil{rn}}R_n}$, we have
\begin{equation}
\label{equ:confian}
\begin{split}
&\frac{p}{2-p}\log \max_{\Pi_n \in \mathrm{NS}(W^{\otimes n} \rightarrow T^{\otimes \ceil{rn}}) }\inf_{\rho_{S^{\ceil{rn}}R_n}}F(\Pi_n \circ W^{\otimes n}(\rho_{S^{\ceil{rn}}R_n} ), T^{\otimes \ceil{rn}}(\rho_{S^{\ceil{rn}}R_n})) \\
\leq &-\sup_{\rho_{S^{\ceil{rn}}R_n}}\widetilde{I}_{\frac{p}{2}}(R_n:Z^{\ceil{rn}})_{T^{\otimes \ceil{rn}} (\rho_{S^{\ceil{rn}}R_n})}+\inf_{\sigma_{Y^n}}\sup_{\tau_{X^{n}R_n}}\widetilde{D}_{\frac{p}{2(p-1)}}( W^{\otimes n}(\tau_{X^{n}R_n} )\|\mathcal{R}^{\sigma_{Y^n}}_{X^n\rightarrow Y^n}(\tau_{X^{n}R_n})) \\
=&-\widetilde{I}_{\frac{p}{2}}(T^{\otimes \ceil{rn}})+\widetilde{I}_{\frac{p}{2(p-1)}}(W^{\otimes n}) \\
\overset{(a)}=&-\ceil{rn}\widetilde{I}_{\frac{p}{2}}(T)+n\widetilde{I}_{\frac{p}{2(p-1)}}(W),
\end{split}
\end{equation}
where (a) comes from Proposition~\ref{prop:mainpro}~(\romannumeral6). Eq.~(\ref{equ:confian}) gives
\begin{equation}
\label{equ:cq}
\begin{split}
&\lim_{n \rightarrow \infty} -\frac{1}{n} \log \mathrm{Succ}_{\mathrm{Pur}}^{\mathrm{NS}}(W^{\otimes n}\rightarrow T^{\otimes  \ceil{rn}}) \\
\geq & \lim_{n \rightarrow \infty} -\frac{1}{n} \log \max_{\Pi_n \in \mathrm{NS}(W^{\otimes n} \rightarrow T^{\otimes \ceil{rn}}) }\inf_{\rho_{S^{nr}R_n}}F(\Pi_n \circ W^{\otimes n}(\rho_{S^{\ceil{rn}}R_n} ), T^{\otimes \ceil{rn}}(\rho_{S^{\ceil{rn}}R_n}))   \\
\ge &\frac{2-p}{p}\big(r\widetilde{I}_{\frac{p}{2}}(T_{Y\rightarrow B})-\widetilde{I}_{\frac{p}{2(p-1)}}(W_{X\rightarrow A})   \big).
\end{split}
\end{equation}
 Because Eq.~(\ref{equ:cq}) holds for any $1<p<2$, we complete the proof of the converse part.
\end{proof}


\subsection{Achievability}

In this section, we prove the achievability  of the strong converse exponents of entanglement-assisted and non-signaling assisted classical-quantum channel interconversion. We first construct an intermediate quantity and derive its variational expression, then prove that it is an upper bound of the strong converse exponents.

\begin{proposition}
\label{prop:var}
Let $W_{X \rightarrow Y}$, $T_{S \rightarrow Z}$ be two classical-quantum channels and $r \geq 0$, we have 
\begin{equation}
\label{equ:maincq}
\begin{split}
&\sup_{1<p<2} \frac{2-p}{p}\big(rI^{\flat}_{\frac{p}{2}}(T_{S\rightarrow Z})-I_{\frac{p}{2(p-1)}}^{\flat}(W_{X\rightarrow Y})   \big) \\
=&\inf_{P_X \in \mathcal{P}(\mathcal{X})}\sup_{P_S \in \mathcal{P}(\mathcal{S})}\inf_{\{\widetilde{T}_s\} \in \mathcal{F}_T}\inf_{\{\widetilde{W}_x\} \in \mathcal{F}_W}
\big(\mathbb{E}_{x\sim P_X}D(\widetilde{W}_x\|W_x)+ r\mathbb{E}_{s \sim P_S}D(\widetilde{T}_s\|T_s)\\
&+|rI(P_{S},\widetilde{T}_{S \rightarrow Z})-I(P_{X},\widetilde{W}_{X \rightarrow Y})|_+ \big),
\end{split}
\end{equation}
where $\frac{1}{\alpha}+\frac{1}{\beta}=2$, $\mathcal{F}_W=\{ \{\widetilde{W}_x\}~|~\widetilde{W}_x \in \mathcal{S}_{W_x},~\forall x\in \mathcal{X} \}$ and $\mathcal{F}_T=\{ \{\widetilde{T}_s\}~|~\widetilde{T}_s \in \mathcal{S}_{T_s},~\forall s\in \mathcal{S} \}$.
\end{proposition}

\begin{proof}
By direct calculation, we have
\begin{equation}
\begin{split}
&\sup_{1<p<2} \frac{2-p}{p}\big(rI^{\flat}_{\frac{p}{2}}(T_{S\rightarrow Z})-I_{\frac{p}{2(p-1)}}^{\flat}(W_{X\rightarrow Y})   \big) \\
=&\sup_{1<p<2} \frac{2-p}{p}\left(r\sup_{P_S \in \mathcal{P}(\mathcal{S})}\inf_{\sigma_Z\in \mathcal{S}(Z)}\sum_{s \in \mathcal{S}}P_S(s)D_{\frac{p}{2}}^{\flat}(T_s\|\sigma_Z)-
\sup_{P_X \in \mathcal{P}(X)}\inf_{\sigma_Y \in \mathcal{S}(Y)} \sum_{x \in \mathcal{X}} P_X(x)D_{\frac{p}{2(p-1)}}^\flat(W_x\|\sigma_Y) \right) \\
\overset{(a)}=&\sup_{1<p<2}\sup_{P_S \in \mathcal{P}(\mathcal{S})}\inf_{\sigma_Z\in \mathcal{S}(Z)}\inf_{P_X \in \mathcal{P}(X)}\sup_{\sigma_Y \in \mathcal{S}(Y)}\inf_{\{\widetilde{T}_s\}\in \mathcal{F}_T}
\inf_{\{\widetilde{W}_x\}\in \mathcal{F}_W}\frac{2-p}{p} \bigg( r\sum_{s \in \mathcal{S}}P_S(s)D(\widetilde{T}_s\|\sigma_Z) \\
&+r\frac{p}{2-p}\sum_{s\in \mathcal{S}}P_S(s)D(\widetilde{T}_s\|T_s)- \sum_{x \in \mathcal{X}}P_X(x)D(\widetilde{W}_x\|\sigma_Y)+\frac{p}{2-p}\sum_{x \in \mathcal{X}}P_X(x)D(\widetilde{W}_x\|W_x)\bigg)\\
\overset{(b)}=&\sup_{1<p<2}\sup_{P_S \in \mathcal{P}(\mathcal{S})}\inf_{P_X \in \mathcal{P}(\mathcal{X})}\inf_{\{\widetilde{T}_s\}\in \mathcal{F}_T}\inf_{\{\widetilde{W}_x\}\in \mathcal{F}_W} \frac{2-p}{p}\bigg(rI(S:Z)_{P_S\widetilde{T}  }+r\frac{p}{2-p} \sum_{s \in \mathcal{S}} P_S(s)D(\widetilde{T}_s\|T_s)\\
&-I(X:A)_{P_X\widetilde{W} }+\frac{p}{2-p} \sum_{x \in \mathcal{X}}P_X(x)D(\widetilde{W}_x\|W_x)  \bigg)\\
=&\sup_{1<p<2}\sup_{P_S \in \mathcal{P}(\mathcal{S})}\inf_{P_X \in \mathcal{P}(\mathcal{X})}\inf_{\{\widetilde{T}_s\}\in \mathcal{F}_T}\inf_{\{\widetilde{W}_x\}\in \mathcal{F}_W} \bigg(\mathbb{E}_{x\sim P_X}D(\widetilde{W}_x\|W_x)  +r\mathbb{E}_{s \sim P_S}D(\widetilde{T}_s\|T_s) \\
&+\frac{2-p}{p}(rI(P_{S},\widetilde{T}_{S \rightarrow Z})-I(P_{X},\widetilde{W}_{X \rightarrow Y})_{P_X\widetilde{W}})\bigg),
\end{split}
\end{equation}
where $(a)$ comes from Proposition~\ref{prop:mainpro}~(\romannumeral5), (b) is due to Sion's minimax theorem, the convexity and the concavity of the objective function in
$\{\widetilde{W}_x\}$, $\{\widetilde{T}_y\}$ and $\sigma_A$ are obvious.

Because $r\mathbb{E}_{s \sim P_S}D(\widetilde{T}_s\|T_s)+\delta r I(P_{S},\widetilde{T}_{S \rightarrow Z})$ is convex in $\{\widetilde{T}_s\}$ and $\mathbb{E}_{x\sim P_X} D(\widetilde{W}_x\|W_x)-\delta I(P_{X},\widetilde{W}_{X \rightarrow Y})=\sup_{\sigma_Y \in \mathcal{S}(Y)}\{\mathbb{E}_{x\sim P_X} D(\widetilde{W}_x\|W_x)-\delta \mathbb{E}_{x\sim P_X} D(\widetilde{W}_x\|\sigma_Y)  \}$ is convex
in $P_X$ and $\{\widetilde{W}_x\}$, we can further apply Sion's minimax theorem to obtain
\begin{equation}
\begin{split}
&\sup_{1<p<2} \frac{2-p}{p}\big(rI^{\flat}_{\frac{p}{2}}(T_{S\rightarrow Z})-I_{\frac{p}{2(p-1)}}^{\flat}(W_{X\rightarrow Y})   \big) \\
=&\sup_{0<\delta<1}\sup_{P_S \in \mathcal{P}(\mathcal{S})}\inf_{P_X \in \mathcal{P}(X)}\inf_{\{\widetilde{T}_s\}\in \mathcal{F}_T}\inf_{\{\widetilde{W}_x\}\in \mathcal{F}_W} \big(\mathbb{E}_{x\sim P_X}D(\widetilde{W}_x\|W_x)  +r\mathbb{E}_{s \sim P_S}D(\widetilde{T}_s\|T_s) \\
&+\delta(rI(P_{S},\widetilde{T}_{S \rightarrow Z})-I(P_{X},\widetilde{W}_{X \rightarrow Y}))\big) \\
=&\sup_{P_S \in \mathcal{P}(\mathcal{S})}\inf_{P_X \in \mathcal{P}(\mathcal{X})}\inf_{\{\widetilde{T}_s\}\in \mathcal{F}_T}\inf_{\{\widetilde{W}_x\}\in \mathcal{F}_W}\sup_{0<\delta<1} \big(\mathbb{E}_{x\sim P_X}D(\widetilde{W}_x\|W_x)  +r\mathbb{E}_{s \sim P_S}D(\widetilde{T}_s\|T_s) \\
&+\delta(rI(P_{S},\widetilde{T}_{S \rightarrow Z})-I(P_{X},\widetilde{W}_{X \rightarrow Y}))\big) \\
=&\sup_{P_S \in \mathcal{P}(\mathcal{S})}\inf_{P_X \in \mathcal{P}(\mathcal{X})}\inf_{\{\widetilde{T}_s\}\in \mathcal{F}_T}\inf_{\{\widetilde{W}_x\}\in \mathcal{F}_W} \big(\mathbb{E}_{x\sim P_X}D(\widetilde{W}_x\|W_x)  +r\mathbb{E}_{s \sim P_S}D(\widetilde{T}_s\|T_s) \\
&+|rI(P_{S},\widetilde{T}_{S \rightarrow Z})-I(P_{X},\widetilde{W}_{X \rightarrow Y})|_+\big) \\
=&\inf_{P_X \in \mathcal{P}(\mathcal{X})}\sup_{P_Y \in \mathcal{P}(\mathcal{Y})}\inf_{\{\widetilde{T}_s\}\in \mathcal{F}_T}\inf_{\{\widetilde{W}_x\}\in \mathcal{F}_W} \big(\mathbb{E}_{x\sim P_X}D(\widetilde{W}_x\|W_x)  +r\mathbb{E}_{s \sim P_S}D(\widetilde{T}_s\|T_s) \\
&+|rI(P_{S},\widetilde{T}_{S \rightarrow Z})-I(P_{X},\widetilde{W}_{X \rightarrow Y})|_+\big). 
\end{split}
\end{equation}
\end{proof}

In the following, we establish an improved first order achievability for classical-quantum channel interconversion with an exponentially vanishing conversion error which serves as a tool to generalize the arguments in classical setting into classical-quantum setting.

\begin{proposition}
\label{prop:cqfir}
Let  $W_{X \rightarrow Y}$ and  $T_{S \rightarrow Z}$ be  two classical-quantum channels.  Let $n \in \mathbb{N}$, $P_S$ be a type on $\mathcal{S}^n$ and $P_X \in \mathcal{P}(\mathcal{X})$. Let
$\delta_{n} = 3\frac{\sqrt{\log(n) \var( {T}_{S \rightarrow Z} )} }{n^{1/4}}$,   $\gamma_{n} =2\frac{\sqrt{\log(n) \var( W_{X \rightarrow Y})} }{n^{1/4}}$ and $r \geq 0$ such that
\begin{equation}
rI(P_{S},T_{S \rightarrow Z})+\delta_n \leq I(P_{X},\widetilde{W}_{X \rightarrow Y})-\gamma_n.
\end{equation}
Then, there is an entanglement-assisted strategy converting $W^{\otimes n}$ to $T^{\otimes \ceil{rn}}$ under any input of type close to $P_S$ and achieving a conversion $Pur$-error at most $\exp(-\sqrt{n}\log(n))$\footnote{if the channel $T$ is classical, shared randomness assistance is sufficient to achieve the same result.}.
\end{proposition}

The proof of this Proposition is stated in Appendix~\ref{app:1order-improved-cq}. Because Eq.~(\ref{equ:maincq}) has the same variational expression as in the classical setting, by applying Proposition~\ref{prop:cqfir} and leveraging the same arguments in the classical setting~(the arguments in Proposition~\ref{prop:ach-var-pur}), we can prove that there exists a sequence of entanglement-assisted schemes whose performances decay to $0$ with an exponential decaying rate upper bounded by Eq.~(\ref{equ:maincq}). Hence, we can get the following proposition.

\begin{proposition}
\label{prop:cqinter}
For any classical-quantum channels $W_{X \rightarrow Y}$, $T_{S \rightarrow Z}$ and $r \geq 0$, we have
\begin{equation}
\begin{split}
\lim_{n \rightarrow \infty} -\frac{1}{n} \log \mathrm{Succ}_{\mathrm{Pur}}^{\mathrm{NS}}(W^{\otimes n}\rightarrow T^{\otimes  \ceil{rn}}) \leq &
\lim_{n \rightarrow \infty} -\frac{1}{n} \log \mathrm{Succ}_{\mathrm{Pur}}^{\mathrm{EA}}(W^{\otimes n}\rightarrow T^{\otimes  \ceil{rn}}) \\
\leq &\sup_{1<p<2} \frac{2-p}{p}\big(rI^{\flat}_{\frac{p}{2}}(T_{S\rightarrow Z})-I_{\frac{p}{2(p-1)}}^{\flat}(W_{X\rightarrow Y})   \big).
\end{split}
\end{equation}
\end{proposition}

\begin{remark}
    In the classical setting, we established the achievability part through the variational expression of the strong converse exponent. The analog quantum R\'enyi divergence with the variational expression is the log-Euclidean R\'enyi divergence. Hence, we use the log-Euclidean R\'enyi divergence to construct this intermediate quantity and derive its variational expression which is similar to the classical one in Proposition~\ref{prop:var}. By following the procedures in Proposition~\ref{prop:ach-var-pur}, we get such an upper bound in terms of log-Euclidean R\'enyi divergence.
\end{remark}

Next, we improve the upper bound in Proposition~\ref{prop:cqinter} to the correct one to accomplish the proof of the achievability part.

\begin{proof}
We first fix an integer $m$ and denote $\mathcal{P}_{\sigma_{Y^m}^u}\circ W_{X \rightarrow Y}^{\ox m}$ and $\mathcal{P}_{\sigma_{Z^m}^u}\circ T_{S \rightarrow Z}^{\ox m}$ as $W_{X^m \rightarrow Y^m}$ and $T_{S^m \rightarrow Z^m}$
respectively. By Proposition~\ref{prop:cqinter}, we can find a sequence of entanglement-assisted conversion schemes $$\{\Pi_k=\{\phi_{\tilde{S}_k\tilde{Y}_k}, \mathcal{E}_{\tilde{S}_k S^{m\ceil{kr}}\rightarrow X^{mk}},\mathcal{D}_{\tilde{Y}_kY^{mk} \rightarrow Z^{m\ceil{kr}}}\}  \}_{k \in \mathbb{N}}$$
such that
\begin{equation}
\label{equ:int}
\begin{split}
&-\lim_{k \rightarrow \infty} \frac{1}{k} \log \inf_{\rho_{R_kS^{m\ceil{kr}}}} F(\mathcal{D}\circ W^{\otimes k}\circ \mathcal{E}(\phi_{\tilde{S}_k\tilde{Y}_k} \otimes \rho_{R_kS^{m\ceil{kr}}}), T^{\otimes \ceil{kr}}_{S^m \rightarrow Z^m}(\rho_{R_kS^{m\ceil{kr}}})) \\
\leq & \sup_{1<p<2} \frac{2-p}{p}\big(rI_{\frac{p}{2}}^{\flat}(T_{S^m \rightarrow Z^m})-I_{\frac{p}{2(p-1)}}^{\flat}(W_{X^m \rightarrow Y^m}) \big).
\end{split}
\end{equation}
For integer $n=mk$, we can construct an entanglement-assisted conversion scheme $\tilde{\Pi}_n$ based on $\Pi_k$ as 
$$\tilde{\Pi}_n=\{ \phi_{\tilde{S}_k\tilde{Y}_k}, \mathcal{E}_{\tilde{S}_k S^{m\ceil{kr}}\rightarrow X^{mk}},\mathcal{P}^{\otimes \ceil{kr}}_{\sigma_{Z^m}^u}\circ \mathcal{D}_{\tilde{Y}_kY^{mk} \rightarrow Z^{m\ceil{kr}}}\circ \mathcal{P}^{\otimes k}_{\sigma_{Y^m}^u} \}$$
 and the performance of this protocol can be evaluated as
\begin{equation}
\label{equ:reint}
\begin{split}
&\inf_{\rho_{R_kS^{m\ceil{kr}}}} F\left(\mathcal{P}^{\otimes \ceil{kr}}_{\sigma_{Z^m}^u}\circ \mathcal{D} \circ W^{\otimes k}_{X^m \rightarrow Y^m}\circ \mathcal{E}(\phi_{\tilde{S}_k\tilde{Y}_k} \otimes \rho_{R_kS^{m\ceil{kr}}}), T^{\otimes m\ceil{kr}}_{S \rightarrow Z}(\rho_{R_kS^{m\ceil{kr}}})\right) \\
\overset{(a)}\geq &\frac{1}{v_{m,|B|}^{\ceil{kr}}}\inf_{\rho_{R_kS^{m\ceil{kr}}}} F\left(\mathcal{P}^{\otimes \ceil{kr}}_{\sigma_{Z^m}^u}\circ \mathcal{D} \circ W^{\otimes k}_{X^m \rightarrow Y^m}\circ \mathcal{E}(\phi_{\tilde{Y}_k\tilde{A}_k} \otimes \rho_{R_kS^{m\ceil{kr}}}), T^{\otimes\ceil{kr}}_{S^m \rightarrow Z^m}(\rho_{R_kS^{m\ceil{kr}}})\right) \\
\overset{(b)}\geq &\frac{1}{v_{m,|B|}^{\ceil{kr}}}\inf_{\rho_{R_kS^{m\ceil{kr}}}} F\left(\mathcal{D} \circ W^{\otimes k}_{X^m \rightarrow Y^m}\circ \mathcal{E}(\phi_{\tilde{Y}_k\tilde{A}_k} \otimes \rho_{R_kS^{m\ceil{kr}}}), T^{\otimes \ceil{kr}}_{S^m \rightarrow Z^m}(\rho_{R_kS^{m\ceil{kr}}})\right),
\end{split}
\end{equation}
where (a) is due to Lemma~\ref{lem:appen1} and (b) is from the data processing inequality of the fidelity function.

Eq.~(\ref{equ:int}),  Eq.~(\ref{equ:reint}) and Lemma \ref{lem:sym} imply that
\begin{equation}
\label{equ:fia}
\begin{split}
&\lim_{n \rightarrow \infty} -\frac{1}{n} \log \mathrm{Succ}_{\mathrm{Pur}}^{\mathrm{EA}}(W^{\otimes n}\rightarrow T^{\otimes m\ceil{kr}}) - \frac{\log v_{m, |Z|}}{m} \\
\leq &\sup_{1<p<2} \frac{2-p}{p m}\big(rI_{\frac{p}{2}}^{\flat}(T_{S^m \rightarrow Z^m})-I_{\frac{p}{2(p-1)}}^{\flat}(W_{X^m \rightarrow Y^m}) \big)  \\
\leq&\sup_{1<p<2} \frac{2-p}{p m}\big(r \sup_{P_{S^m} \in \mathcal{P}(\mathcal{S}^{m})} D_{\frac{p}{2}}^{\flat}\big(\sum_{s^m}P_{S^m}(s^m)\ket{s^m}\bra{s^m} \otimes
\mathcal{P}_{\sigma^u_{Z^m}}(T^{\otimes m}_{s^m})\|\sum_{s^m}P_{S^m}(s^m)\ket{s^m}\bra{s^m} \otimes \sigma^u_{Z^m} \big)  \\
&-\sup_{P_{X} \in \mathcal{P}(\mathcal{X})}\inf_{\sigma_{Y^m} \in \mathcal{S}(Y^m)}
D_{\frac{p}{2(p-1)}}^{\flat}\big(\sum_{x^m}P_{X^m}(x^m)\ket{x^m}\bra{x^m} \otimes
\mathcal{P}_{\sigma^u_{Y^m}}(W^{\otimes m}_{x^m})\|\sum_{x^m}P_{X^m}(x^m)\ket{x^m}\bra{x^m} \otimes \sigma_{Y^m} \big) \big)\\
\overset{(a)}=&\sup_{1<p<2} \frac{2-p}{p m}\big(r \sup_{P_{S^m} \in \mathcal{P}_{\mathrm{sym}}(\mathcal{S}^{m})} D_{\frac{p}{2}}^{\flat}\big(\sum_{s^m}P_{S^m}(s^m)\ket{s^m}\bra{s^m} \otimes
\mathcal{P}_{\sigma^u_{Z^m}}(T^{\otimes m}_{s^m})\|\sum_{s^m}P_{S^m}(s^m)\ket{s^m}\bra{s^m} \otimes \sigma^u_{Z^m} \big)  \\
&-\sup_{P_{X} \in \mathcal{P}(\mathcal{X})}\inf_{\sigma_{Y^m} \in \mathcal{S}(Y^m)}
D_{\frac{p}{2(p-1)}}^{\flat}\big(\sum_{x^m}P_{X^m}(x^m)\ket{x^m}\bra{x^m} \otimes
\mathcal{P}_{\sigma^u_{Y^m}}(W^{\otimes m}_{x^m})\|\sum_{x^m}P_{X^m}(x^m)\ket{x^m}\bra{x^m} \otimes \sigma_{Y^m} \big) \big),
\end{split}
\end{equation}
where $(a)$ is because that $\mathcal{P}_{\sigma_{Z^m}^u} \circ T^{\otimes m}$ and $\sigma_{Z^n}^u$ are permutation invariant, the supremum can be restricted to over all permutation invariant distributions. Then we can further upper bound 
Eq.~(\ref{equ:fia}) as
\begin{equation}
\begin{split}
&\lim_{n \rightarrow \infty} -\frac{1}{n} \log \mathrm{Succ}_{\mathrm{Pur}}^{\mathrm{EA}}(W^{\otimes n}\rightarrow T^{\otimes  m\ceil{kr}}) -\frac{\log v_{m, |Z|}}{m}  \\
\overset{(a)}\leq&\sup_{1<p<2} \frac{2-p}{p m}\big(r \sup_{P_{S^m} \in \mathcal{P}_{sym}(\mathcal{S}^{ m})}
 D_{\frac{p}{2}}^{\flat}(\sum_{s^m}P_{S^m}(x^m)\ket{s^m}\bra{s^m} \otimes
\mathcal{P}_{\sigma^u_{Z^m}}(T^{\otimes m}_{s^m})\|\sum_{s^m}P_{S^m}(s^m)\ket{s^m}\bra{s^m} \otimes \sigma^u_{Z^m} )  \\
&-\sup_{P_{X} \in \mathcal{P}(\mathcal{X})}\inf_{\sigma_{Y^m} \in \mathcal{S}_{sym}(Y^m)}
D_{\frac{p}{2(p-1)}}^{\flat}(\sum_{x^m}P^{\otimes m}_{X}(x^m)\ket{x^m}\bra{x^m} \otimes
\mathcal{P}_{\sigma^u_{Y^m}}(W^{\otimes m}_{x^m})\|\sum_{x^m}P^{\otimes m}_{X}(x^m)\ket{x^m}\bra{x^m} \otimes \sigma_{Y^m} ) \big)  \\
\leq &\sup_{1<p<2} \frac{2-p}{p m}\big(r \sup_{P_{Y^m} \in \mathcal{P}_{sym}(\mathcal{Y}^{ m})}
 D_{\frac{p}{2}}^{\flat}(\sum_{s^m}P_{S^m}(s^m)\ket{s^m}\bra{s^m} \otimes
\mathcal{P}_{\sigma^u_{Z^m}}(T^{\otimes m}_{s^m})\|\sum_{s^m}P_{S^m}(s^m)\ket{s^m}\bra{s^m} \otimes \sigma^u_{Z^m} )  \\
&-\sup_{P_{X} \in \mathcal{P}(\mathcal{X})}
D_{\frac{p}{2(p-1)}}^{\flat}(\sum_{x^m}P^{\otimes m}_{X}(x^m)\ket{x^m}\bra{x^m} \otimes
\mathcal{P}_{\sigma^u_{Y^m}}(W^{\otimes m}_{x^m})\|\sum_{x^m}P^{\otimes m}_{X}(x^m)\ket{x^m}\bra{x^m} \otimes v_{m,|Y|}\sigma^u_{Y^m} ) \big)  \\
\overset{(c)}\leq &\sup_{1<p<2} \frac{2-p}{p m}\big(r \sup_{P_{Y^m} \in \mathcal{P}_{sym}(\mathcal{Y}^{ m})}
 \widetilde{D}_{\frac{p}{2}}(\sum_{s^m}P_{S^m}(s^m)\ket{s^m}\bra{s^m} \otimes
\mathcal{P}_{\sigma^u_{Z^m}}(T^{\otimes m}_{s^m})\|\sum_{s^m}P_{S^m}(s^m)\ket{s^m}\bra{s^m} \otimes \sigma^u_{Z^m} )  \\
&-\sup_{P_{X} \in \mathcal{P}(\mathcal{X})}
\widetilde{D}_{\frac{p}{2(p-1)}}(\sum_{x^m}P^{\otimes m}_{X}(x^m)\ket{x^m}\bra{x^m} \otimes
\mathcal{P}_{\sigma^u_{Y^m}}(W^{\otimes m}_{x^m})\|\sum_{x^m}P^{\otimes m}_{X}(x^m)\ket{x^m}\bra{x^m} \otimes \sigma^u_{Y^m} ) \big) +\frac{\log v_{m, |Y|}}{m}, 
\end{split}
\end{equation}
where (a) is because that $\sum_{x^m}P^{\otimes m}_{X}(x^m)\ket{x^m}\bra{x^m} \otimes
\mathcal{P}_{\sigma^u_{Y^m}}(W^{\otimes m}_{x^m})$ is a permutation invariant state, the infimum can be restricted to over all symmetric states and (c) is due to $\sum_{x^m}P^{\otimes m}_{X}(x^m)\ket{x^m}\bra{x^m} \otimes
\mathcal{P}_{\sigma^u_{Y^m}}(W^{\otimes m}_{x^m})$ commutes with $\sum_{x^m}P^{\otimes m}_{X}(x^m)\ket{x^m}\bra{x^m} \otimes \sigma^u_{Y^m} $ and $\sum_{s^m}P_{S^m}(s^m)\ket{s^m}\bra{s^m} \otimes
\mathcal{P}_{\sigma^u_{Z^m}}(T^{\otimes m}_{s^m})$ commutes with $\sum_{s^m}P_{S^m}(s^m)\ket{s^m}\bra{s^m} \otimes \sigma^u_{Z^m}$. From the data processing inequality of the sandwiched R\'enyi  divergence and Proposition~\ref{prop:mainpro}~(\romannumeral3), we have
\begin{equation}
\label{equ:fian}
\begin{split}
&\lim_{n \rightarrow \infty} -\frac{1}{n} \log \mathrm{Succ}_{\mathrm{Pur}}^{\mathrm{EA}}(W^{\otimes n}\rightarrow T^{\otimes  m\ceil{kr}}) - \frac{3\log v_{m,|A|}}{m}-\frac{2\log v_{m, |B|}}{m} \\
\leq &\sup_{1<p<2} \frac{2-p}{p m}\big(\sup_{P_{S^m} \in \mathcal{P}_{\mathrm{sym}}(\mathcal{S}^{ m})}
\widetilde{D}_{\frac{p}{2}}(\sum_{s^m}P_{S^m}(s^m)\ket{s^m}\bra{s^m} \otimes
T^{\otimes m}_{s^m}\|\sum_{s^m}P_{S^m}(s^m)\ket{s^m}\bra{s^m} \otimes \sigma^u_{Z^m} )  \\
&-\sup_{P_{X} \in \mathcal{P}(\mathcal{X})}
\widetilde{D}_{\frac{p}{2(p-1)}}(\sum_{x^m}P^{\otimes m}_{X}(x^m)\ket{x^m}\bra{x^m} \otimes
W^{\otimes m}_{x^m}\|\sum_{x^m}P^{\otimes m}_{X}(x^m)\ket{x^m}\bra{x^m} \otimes \sigma^u_{Y^m} ) \big) -\frac{\log v_{m, |Z|}}{m}  \\
\overset{(a)}\leq &\sup_{1<p<2} \frac{2-p}{p m}\big(r \sup_{P_{S^m} \in \mathcal{P}_{\mathrm{sym}}(\mathcal{S}^{ m})} \inf_{\sigma_{Z^m}\in \mathcal{S}_{sym}(Z^m)}
\widetilde{D}_{\frac{p}{2}}(\sum_{s^m}P_{S^m}(s^m)\ket{s^m}\bra{s^m} \otimes
T^{\otimes m}_{s^m}\|\sum_{s^m}P_{S^m}(s^m)\ket{s^m}\bra{s^m} \otimes \sigma_{Z^m} )  \\
&-\sup_{P_{X} \in \mathcal{P}(\mathcal{X})}\inf_{\sigma_{Y^m}\in \mathcal{S}(Y^m)}
\widetilde{D}_{\frac{p}{2(p-1)}}(\sum_{x^m}P^{\otimes m}_{X}(x^m)\ket{x^m}\bra{x^m} \otimes
W^{\otimes m}_{x^m}\|\sum_{x^m}P^{\otimes m}_{X}(x^m)\ket{x^m}\bra{x^m} \otimes \sigma_{Y^m} ) \big) \\
= & \sup_{1<p<2} \frac{2-p}{p m}\big(r \sup_{P_{S^m} \in \mathcal{P}_{\mathrm{sym}}(\mathcal{S}^{ m})} \widetilde{I}_{\frac{p}{2}}(P_{S^m}:T^{\otimes m})- 
\sup_{P_{X} \in \mathcal{P}(\mathcal{X})}\widetilde{I}_{\frac{p}{2(p-1)}}(P_X^{\otimes m}:W^{\otimes m} )
\big)   \\
\overset{(b)}\leq &\sup_{1<p<2} \frac{2-p}{p m}\big(r \widetilde{I}_{\frac{p}{2}}(T^{\otimes m})- 
m\widetilde{I}_{\frac{p}{2(p-1)}}(W)
\big) +\frac{3\log v_{m,|A|}}{m}+\frac{2\log v_{m, |B|}}{m} \\
\overset{(c)}= &\sup_{1<p<2} \frac{2-p}{p }\big(r \widetilde{I}_{\frac{p}{2}}(T)- 
\widetilde{I}_{\frac{p}{2(p-1)}}(W)
\big) ,
\end{split}
\end{equation}
where (a) is from Lemma~\ref{lem:sym} and Proposition~\ref{prop:mainpro}~(\romannumeral1), (b) comes from Proposition~\ref{prop:mainpro}~(\romannumeral4) and (c) is from Proposition~\ref{prop:mainpro}~(\romannumeral6).

Noticing that Eq.~(\ref{equ:fian}) holds for any integer $m$, by letting $m \rightarrow \infty$, we have
\begin{equation}
\lim_{n \rightarrow \infty} -\frac{1}{n} \log \mathrm{Succ}_{\mathrm{Pur}}^{\mathrm{EA}}(W^{\otimes n}\rightarrow T^{\otimes  m\ceil{kr}}) 
\leq \sup_{1<p<2} \frac{2-p}{p}\big(r\widetilde{I}_{\frac{p}{2}}(T_{Y\rightarrow B})-\widetilde{I}_{\frac{p}{2(p-1)}}(W_{X\rightarrow A})   \big).
\end{equation}
\end{proof}


\section{Conclusion}\label{sec:conclusion}

In this work, we initiated the study of high-order refinements in the context of the channel interconversion problem. In particular, we determined the exact strong converse exponent under the purified distance for the conversion of classical-quantum channels, with natural extensions to metrics based on R\'enyi divergences. We further derived the exact strong converse exponent under the total variation distance in the fully classical setting. Notably, our results apply to all rates below and above the capacities of the involved channels.

Our analysis is starting from the strong converse exponent methods in \cite{Dueck1979Jan, csiszar2011information, MosonyiOgawa2017strong, Li2023Nov, Berta2024Oct}, combined with an improved first-order achievability scheme. However, several open questions remain. In particular, understanding the small deviation regime and the precise characterization of error exponents appear to require new techniques. We leave these as interesting directions for future research.


\section*{Acknowledgment}

AO, YY, and MB acknowledge funding by the European Research Council (ERC Grant Agreement No.~948139) and MB acknowledges support from the Excellence Cluster - Matter and Light for Quantum Computing (ML4Q).

\printbibliography


\appendix

\section{Improved first order achievability for classical channel interconversion}\label{app:1order-improved}

In this section, we prove Proposition~\ref{prop:1order-improved} which provides a first order achievability for channel interconversion with a vanishing conversion error. 

\begin{proposition}[Restatement of Proposition~\ref{prop:1order-improved}]
Let $T_{Z|S}$ and $W_{Y|X}$ be two channels over finite alphabets. Let $n\in \mathbb{N}$, $r\ge 0$,    $P_{S}$ be a type on $\cS^{\ceil{rn}}$ and $P_{X}\in \cP(\cX)$.  Let $\delta_{n} = 2\frac{\sqrt{\log(n) r \var( {T}_{Z|S} )} }{n^{1/4}}$ and  $\gamma_{n} =2\frac{\sqrt{\log(n) \var( W_{Y|X})} }{n^{1/4}}$    such that:
\begin{align}
r\cdot I(P_{S}, {T}_{Z|S} ) +\delta_{n}\le I(P_{X}, {W}_{Y|X}) - \gamma_{n}.
\end{align}
Then, there is a shared-randomness strategy converting $W^{\times n}$ to $T^{\times \ceil{rn}}$ under any input of type close to $P_{S}$ and achieving a conversion TV-error at most $\exp(-\sqrt{n}\log(n))$. That is, for all $A\ge 0$, there is  $n_{0}(A)$ such that for all $n\ge n_{0}$:
\begin{align}
\inf_{N\in \mathrm{SR}} \sup_{\substack{\bm{s}^{\ceil{rn}}\sim P_{S}' \\ \mathrm{TV}(P_{S}, P_{S}')\le \frac{A}{n}}} \mathrm{TV}\left(({N}\circ W^{\times n})_{\bm{s}^{\ceil{rn}}} , T^{\times \ceil{rn}}_{\bm{s}^{\ceil{rn}}}\right) \le \exp(-\sqrt{n}\log(n)).
\end{align}
Furthermore, $N$ has marginal $P_{\bm{X}}^{\times n}$.
\end{proposition}
\begin{proof}[Proof of Proposition~\ref{prop:1order-improved}]
 It suffices to consider a fixed type $P_{S}$ since by the continuity bound of \cite{Fannes1973Dec}, we have that:
\begin{align}
    \mathrm{TV}(P_{S}, P_{S}')\le \frac{A}{n} \Rightarrow I(P_{S}', {T}_{Z|S} ) \le I(P_{S}, {T}_{Z|S} ) +  \cO\left(\frac{1}{n}\right).
    \end{align}
Let $R$ such that 
\begin{align}\label{eq:R}
r\cdot I(P_{S}, {T}_{Z|S} ) +\delta_{n}\le R \le  I(P_{X}, {W}_{Y|X}) - \gamma_{n}.
\end{align}
Let $\id_2$ be the binary identity channel. In order to convert $T$ to $W$, we will use a concatenation between channel coding and channel simulation. More precisely, we use channel coding to simulate $\id_2^{\times \ceil{Rn}}$ using $W^{\times n}$.  Then, the obtained $\id_2^{\times \ceil{Rn}}$ is used to simulate $T^{\times \ceil{rn}}$. That is 
\begin{align}
 W^{\times n} \xrightarrow{\text{channel coding}} \id_2^{\times \ceil{Rn}}\xrightarrow{\text{channel simulation}}T^{\times \ceil{rn}}.
\end{align}
It is not difficult to see that the global conversion error can be bounded by the sum of the errors of channel coding and simulation. The inequalities that $R$ satisfies in \eqref{eq:R} ensure that both tasks have small error. \\
\textbf{Channel coding.} Let $M=2^{\ceil{Rn}}$. Using the random coding achievability error exponent \cite{Fano1961Nov}, we have that for all $\alpha \in (0, \frac{1}{2})$
\begin{align}
\epsilon_{\text{coding}}^{\text{SR}}\le \exp\left(-n\frac{\alpha}{1-\alpha}\left( I_{1-\alpha}(P_{X}, W_{Y|X}) -R\right)\right).
\end{align}
 We  choose $\alpha = \frac{\gamma_{n}}{\var(W_{Y|X})  }$. By Taylor expansion \cite[Proposition 11]{hayashi2016correlation}, we have that 
\begin{align}
I_{1-\alpha}(P_{X}, W_{Y|X})&= I(P_{X}, W_{Y|X}) -\frac{1}{2}\alpha \var(P_{X}, W_{Y|X}) +o(\alpha)
\ge  I(P_{X}, W_{Y|X}) - \frac{\gamma_{n}}{2} + o(\gamma_{n}).
\end{align}
Since $R\le I(P_{X}, W_{Y|X})  -\gamma_{n}$, we deduce that 
\begin{align}
\epsilon_{\text{coding}}^{\text{SR}} &\le \exp\left(-n\frac{\alpha}{1-\alpha}\left( I_{1-\alpha}(P_{X}, W_{Y|X}) -R\right)\right)
\\&\le \exp\left(-n \frac{\gamma_{n}}{\var( W_{Y|X})  }\left(  I(P_{X}, W_{Y|X}) - \frac{\gamma_{n}}{2} + o(\gamma_{n}) -R\right)\right)
\\&\le \exp\left(-n \frac{\gamma_{n}}{\var( W_{Y|X})  }\left(   \frac{\gamma_{n}}{2} + o(\gamma_{n}) \right)\right)
\\&\le \frac{1}{3}\exp\left( -\sqrt{n}\log(n) \right)
\end{align}
where we use in the last inequality $\gamma_{n} = 2\frac{\sqrt{\log(n) \var(W_{Y|X})} }{n^{1/4}}$ and $n$ is sufficiently large.

\textbf{Channel simulation.} Since we consider inputs $\bm{s}^{\ceil{rn}}$ of fixed type $P_{S}$, it is not clear how to apply the worst-case channel simulation achievability of \cite{cao2024channel}. Let $P_Z = \sum_{s\in \cS} P_{S}(s) T_{Z|S}(\cdot|s)$. 

Let $\cG = \left\{ \bm{z}^{\ceil{rn}} : T^{\times \ceil{rn}}(\bm{z}^{\ceil{rn}}|\bm{s}^{\ceil{rn}}) \le \frac{1}{n}e^{\ceil{Rn}}\cdot P_{\bm{Z}}^{\times \ceil{rn}}(\bm{z}^{\ceil{rn}}) \right\}$, $\alpha = \frac{\frac{1}{n}e^{\ceil{Rn}}-1}{\frac{1}{n}e^{\ceil{Rn}}\cdot P^{\times \ceil{rn}}(\cG) -T^{\times \ceil{rn}}(\cG|\bm{s}^{\ceil{rn}})}$, and construct:
\begin{align}
\widetilde{T}^{\rm{NS}}(\bm{z}^{\ceil{rn}}|\bm{s}^{\ceil{rn}}) =
\begin{cases}
  \alpha T^{\times \ceil{rn}}(\bm{z}^{\ceil{rn}}|\bm{s}^{\ceil{rn}}) +(1-\alpha) \frac{1}{n}e^{\ceil{Rn}}\cdot P_{\bm{Z}}^{\times \ceil{rn}}(\bm{z}^{\ceil{rn}}) & \text{if } \bm{z}^{\ceil{rn}}\in \cG, \\
\frac{1}{n}e^{\ceil{Rn}}\cdot P_{\bm{Z}}^{\times \ceil{rn}}(\bm{z}^{\ceil{rn}}) & \text{if } \bm{z}^{\ceil{rn}}\notin \cG.
\end{cases}
\end{align}
We have $\alpha \in [0,1]$ and thus $\widetilde{T}^{\rm{NS}}(\bm{z}^{\ceil{rn}}|\bm{s}^{\ceil{rn}}) \le \frac{1}{n}e^{\ceil{Rn}}\cdot P_{\bm{Z}}^{\times \ceil{rn}}(\bm{z}^{\ceil{rn}}) $ for all $\bm{z}^{\ceil{rn}}$.  Hence, the channel $\widetilde{T}^{\rm{NS}}$ can be seen as a non-signaling synthesized channel for simulating $W^{\ceil{rn}}$ and achieves a simulation error satisfying for any $\beta\ge 0$:
\begin{align}
\mathrm{TV}\left(\widetilde{T}^{\rm{NS}}(\cdot|\bm{s}^{\ceil{rn}}) ,  T^{\times \ceil{rn}}(\cdot|\bm{s}^{\ceil{rn}})\right) &= \frac{1}{2}\sum_{\bm{z}^{\ceil{rn}} \in \cG} (1-\alpha) \left( \frac{1}{n}e^{\ceil{Rn}} P_{\bm{Z}}^{\times \ceil{rn}}(\bm{z}^{\ceil{rn}})-T^{\times \ceil{rn}}(\bm{z}^{\ceil{rn}}|\bm{s}^{\ceil{rn}})\right)
\\&\quad +  \frac{1}{2}\sum_{\bm{z}^{\ceil{rn}} \notin \cG}  \left( T^{\times \ceil{rn}}(\bm{z}^{\ceil{rn}}|\bm{s}^{\ceil{rn}})-\frac{1}{n}e^{\ceil{Rn}} P_{\bm{Z}}^{\times \ceil{rn}}(\bm{z}^{\ceil{rn}})\right)
\\&= \sum_{\bm{z}^{\ceil{rn}} \notin \cG}  \left( T^{\times \ceil{rn}}(\bm{z}^{\ceil{rn}}|\bm{s}^{\ceil{rn}})-\frac{1}{n}e^{\ceil{Rn}} P_{\bm{Z}}^{\times \ceil{rn}}(\bm{z}^{\ceil{rn}})\right)
\\&\le \sum_{\bm{z}^{\ceil{rn}} \notin \cG}   \left(T^{\times \ceil{rn}}(\bm{z}^{\ceil{rn}}|\bm{s}^{\ceil{rn}})\right)^{1+\beta}\left(\frac{1}{n}e^{\ceil{Rn}} P_{\bm{Z}}^{\times \ceil{rn}}(\bm{z}^{\ceil{rn}})\right)^{-\beta}
\\&\le \exp\left(	-\beta n\left(R-  r\mathbb{E}_{s\sim P_{S}} D_{1+\beta}(T_{s}\| P_Z) \right)	+\beta \log(n)	\right)
\\&\le \exp\left(	-\beta n\left(R-   rD_{1+\beta}(P_{S}T_{Z|S}\|P_{S} P_Z) \right)	+\beta \log(n)	\right)
\end{align}
where we used Jensen's inequality in the last inequality.  We shall choose $\beta = \frac{\delta_{n}}{r\var( {T}_{Z|S} )}$. By Taylor expansion \cite[Proposition 11]{hayashi2016correlation}, we have that 
\begin{align}
D_{1+\beta}(P_{S}T_{Z|S}\| P_{S}P_Z) &=  I(P_{S},T_{Z|S}) +\frac{1}{2}\beta \var(P_{S}, {T}_{Z|S} ) +o(\beta)
\le   I(P_{S},T_{Z|S}) + \frac{\delta_{n}}{2r} + o(\delta_{n}).
\end{align}
Since $R\ge r\cdot I(P_{S}, {T}_{Z|S} ) +\delta_{n}$, we deduce that 
\begin{align}
\mathrm{TV}\left(\widetilde{T}^{\rm{NS}}(\cdot|\bm{s}^{\ceil{rn}}) ,  T^{\times \ceil{rn}}(\cdot|\bm{s}^{\ceil{rn}})\right) &\le  \exp\left(	-\beta n\left(R-   rD_{1+\beta}(P_{S}T_{Z|S}\|P_{S} P_Z) \right)	+\beta \log(n)	\right)
\\&\le \exp\left(	-n  \frac{\delta_{n}}{r\var({T}_{Z|S} ) }\left( \delta_{n} -\frac{\delta_{n}}{2} -o(\delta_{n}) \right)	+ \frac{\delta_{n}\log(n)}{\var({T}_{Z|S} )} \right)
\\&= \exp\left( - \frac{n\delta_{n}^2}{ 2r\var({T}_{Z|S} ) } +\frac{\delta_{n}\log(n)}{r\var({T}_{Z|S} )} +o\left({n\delta_{n}^2}\right)\right)
\\&\le \frac{1}{3}\exp\left( -\sqrt{n}\log(n) \right),
\end{align}
where we use in the last inequality $\delta_{n} = 2\frac{\sqrt{\log(n) r \var({T}_{Z|S} )} }{n^{1/4}}$ and $n$ is sufficiently large.

This NS strategy can be rounded to a SR strategy $\widetilde{T}^{\rm{SR}}$ of size $e^{\ceil{Rn}}$ satisfying~\cite{Oufkir2024OctMeta}:
\begin{align}
\widetilde{T}^{\rm{SR}}(\bm{z}^{\ceil{rn}}|\bm{s}^{\ceil{rn}}) \ge (1-e^{-n})\widetilde{T}^{\rm{NS}}(\bm{z}^{\ceil{rn}}|\bm{s}^{\ceil{rn}}), \quad \forall \bm{z}^{\ceil{rn}}\in \cZ^{\ceil{rn}}. 
\end{align}
Therefore, there exists a SR conversion strategy to simulate $T^{\times \ceil{rn}}$ using $\id_2^{\times \ceil{Rn}}$ with an error:
\begin{align}
\epsilon_{\text{simulation}}^{\text{SR}}&\le \mathrm{TV}\left(\widetilde{T}^{\rm{SR}}(\cdot|\bm{s}^{\ceil{rn}}) ,  T^{\times \ceil{rn}}(\cdot|\bm{s}^{\ceil{rn}})\right)
\\&= \sum_{\bm{z}^{\ceil{rn}}} \left( T^{\times \ceil{rn}}(\bm{z}^{\ceil{rn}}|\bm{s}^{\ceil{rn}}) - \widetilde{T}^{\rm{SR}}(\bm{z}^{\ceil{rn}}|\bm{s}^{\ceil{rn}}) \right)_{+}
\\&\le \sum_{\bm{z}^{\ceil{rn}}} \left( T^{\times \ceil{rn}}(\bm{z}^{\ceil{rn}}|\bm{s}^{\ceil{rn}}) - (1-e^{-n})\widetilde{T}^{\rm{NS}}(\bm{z}^{\ceil{rn}}|\bm{s}^{\ceil{rn}}) \right)_{+}
\\&\le  \sum_{\bm{z}^{\ceil{rn}}} \left( T^{\times \ceil{rn}}(\bm{z}^{\ceil{rn}}|\bm{s}^{\ceil{rn}}) - \widetilde{T}^{\rm{NS}}(\bm{z}^{\ceil{rn}}|\bm{s}^{\ceil{rn}}) \right)_{+} + e^{-n}\widetilde{T}^{\rm{NS}}(\bm{z}^{\ceil{rn}}|\bm{s}^{\ceil{rn}})
\\&\le \frac{1}{3}\exp\left( -\sqrt{n}\log(n)  \right)+ e^{-n}. 
\end{align}
The concatenation strategy has then a total conversion error:
\begin{align}
\epsilon^{\text{SR}}&\le \epsilon_{\text{coding}}^{\text{SR}}+\epsilon_{\text{simulation}}^{\text{SR}}
\le \frac{1}{3}\exp\left( -\sqrt{n}\log(n)  \right)+ \frac{1}{3}\exp\left( -\sqrt{n}\log(n)  \right)+e^{-n}\le  \exp\left( -\sqrt{n}\log(n)  \right)
\end{align}
for sufficiently large $n$. 

To prove that  $N$ has marginal $P_{\bm{X}}^{\times n}$, we remark that the random coding uses codewords sampled i.i.d.\ according to $P_{\bm{X}}^{\times n}$ and the rounding strategy of \cite{Oufkir2024OctMeta} is based on standard rejection sampling technique. 
\end{proof}


\section{Improved first order achievability for classical-quantum channel interconversion}\label{app:1order-improved-cq}

In this section, we prove Proposition~\ref{prop:cqfir} that we restate. 
\begin{proposition}[Restatement of Proposition~\ref{prop:cqfir}]
Let  $W_{X \rightarrow Y}$ and  $T_{S \rightarrow Z}$ be  two classical-quantum channels.  Let $n \in \mathbb{N}$, $P_S$ be a type on $\mathcal{S}^n$ and $P_X \in \mathcal{P}(\mathcal{X})$. Let
$\delta_{n} = 3\frac{\sqrt{\log(n) \var( {T}_{S \rightarrow Z} )} }{n^{1/4}}$,   $\gamma_{n} =2\frac{\sqrt{\log(n) \var( W_{X \rightarrow Y})} }{n^{1/4}}$ and $r \geq 0$ such that
\begin{equation}
rI(P_S:T_{S \rightarrow Z})+\delta_n \leq I(P_X:W_{X \rightarrow Y})-\gamma_n.
\end{equation}
Then, there is an entanglement-assisted strategy converting $W^{\otimes n}$ to $T^{\otimes \ceil{rn}}$ under any input of type close to $P_S$ and achieving a conversion $Pur$-error at most $\exp(-\sqrt{n}\log(n))$.
\end{proposition}

\begin{proof}
The proof of this proposition is similar to Proposition~\ref{prop:1order-improved}. It can be divided into a concatenation between channel coding and channel simulation. The proof of the channel coding part is almost the same, relying on the error exponent result for classical-quantum channels \cite{Cheng2025Jul}. 

Hence, we focus on the  channel simulation part in the following.
Let $R$ such that 
\begin{align}\label{eq:R-CQ}
rI(P_S:T_{S \rightarrow Z})+\delta_n \le R \le I(P_X:W_{X \rightarrow Y})-\gamma_n.
\end{align}
\textbf{Channel simulation.}  Let $\sigma_Z=\sum_{s \in \mathcal{S}}P_S(s)T_s$, $\Pi_{s^{\ceil{rn}}}=\left\{ \cE_{\sigma^{\otimes \ceil{rn}}}(T^{\otimes \ceil{rn}}_{s^{\ceil{rn}}}) \mge \tfrac{1}{\nu_{nr}} (\frac{1}{n}e^{\ceil{Rn}}-1)\sigma^{\otimes\ceil{rn}}\right\}$  and construct
\begin{align}\label{eq:construction-smoothing-channel}
    \widetilde{T}^{\mathrm{NS}}_{s^{\ceil{rn}}} = \Pi_{s^{\ceil{rn}}}^c T_{s^{\ceil{rn}}}\Pi_{s^{\ceil{rn}}}^c + \tr{ T_{s^{\ceil{rn}}}\Pi_{s^{\ceil{rn}}} }\sigma^{\otimes \ceil{rn}}.
\end{align}
Clearly $\widetilde{T}^{\mathrm{NS}}_{s^{\ceil{rn}}}$ is a quantum state. Moreover, by pinching inequality~\cite{hayashi2002optimal}, we have that  $ T_{s^{\ceil{rn}}}\mle \nu_{nr} \cE_{\sigma^{\otimes \ceil{rn}}}(T_{s^{\ceil{rn}}})$  so 
\begin{equation}
\begin{split}
    \widetilde{T}^{\mathrm{NS}}_{s^{\ceil{rn}}} &\mle \Pi_{s^{\ceil{rn}}}^c \nu_{nr} \cE_{\sigma^{\otimes \ceil{rn}}}(T^{\otimes \ceil{rn}}_{s^{\ceil{rn}}}) \Pi_{s^{\ceil{rn}}}^c + \tr {T^{\otimes \ceil{rn}}_{s^{\ceil{rn}}} \Pi_{s^{\ceil{rn}}}} \sigma^{\otimes \ceil{rn}} 
    \\
    &\mle \Pi_{s^{\ceil{rn}}}^c (\frac{1}{n}e^{\ceil{Rn}}-1)\sigma^{\otimes \ceil{rn}}   \Pi_{s^{\ceil{rn}}}^c+ \Pi_{s^{ \ceil{rn}}} (\frac{1}{n}e^{\ceil{Rn}}-1)\sigma^{\otimes \ceil{rn}}  \Pi_{s^{\ceil{rn}}}+\sigma^{\otimes \ceil{rn}}
    = \frac{1}{n}e^{\ceil{Rn}}\sigma^{\otimes \ceil{rn}}.
\end{split}
\end{equation}
Hence, $\widetilde{T}^{\mathrm{NS}}$ can be seen as a nonsignaling synthesized channel $T^{\otimes\ceil{rn}}$ and its performance can be evaluated as 
\begin{align}
    F\left(T_{s^{\ceil{rn}}}^{\otimes \ceil{rn}}, \widetilde{T}^{\mathrm{NS}}_{s^{\ceil{rn}}}\right)&= \left(\tr{\sqrt{\sqrt{T_{s^{\ceil{rn}}}^{\otimes \ceil{rn}}} \widetilde{T}^{\mathrm{NS}}_{s^{\ceil{rn}}} \sqrt{T_{s^{\ceil{rn}}}^{\otimes \ceil{rn}}}}}\right)^2
    \\&\ge  \left(\tr{\sqrt{\sqrt{T_{s^{\ceil{rn}}}^{\otimes \ceil{rn}}} \Pi_{s^{\ceil{rn}}}^c T_{s^{\ceil{rn}}}^{\otimes \ceil{rn}} \Pi_{s^{\ceil{rn}}}^c \sqrt{T_{s^{\ceil{rn}}}^{\otimes \ceil{rn}}}}}\right)^2
    \\&= \left(\tr{\sqrt{T_{s^{\ceil{rn}}}^{\otimes \ceil{rn}}} \Pi_{s^{\ceil{rn}}}^c \sqrt{T_{s^{\ceil{rn}}}^{\otimes \ceil{rn}}} }\right)^2
     \\&= \tr{T_{s^{\ceil{rn}}}^{\otimes \ceil{rn}} \Pi_{s^{\ceil{rn}}}^c  }^2. 
\end{align}
Therefore 
\begin{align}
 T_{s^{\ceil{rn}}}^{\otimes \ceil{rn}} &\mle \frac{1}{n}e^{\ceil{Rn}}\sigma^{\otimes \ceil{rn}}
  \; \text{ and }\; 1-F\left(T_{s^{\ceil{rn}}}^{\otimes \ceil{rn}}, \widetilde{T}^{\mathrm{NS}}_{s^{\ceil{rn}}}\right)\le 1-\tr{T_{s^{\ceil{rn}}}^{\otimes \ceil{rn}} \Pi_{s^{\ceil{rn}}}^c  }^2\le2\,\tr{T_{s^{\ceil{rn}}}^{\otimes \ceil{rn}} \Pi_{s^{\ceil{rn}}}}. 
\end{align}
\sloppy Using the definition of the projector $\Pi_{s^{\ceil{rn}}}$, we have that 
\begin{equation}
    \tr{T_{s^{\ceil{rn}}}^{\otimes \ceil{rn}} \Pi_{s^{\ceil{rn}}}}= \tr{\cE_{\sigma^{\otimes \ceil{rn}}}(T_{s^{\ceil{rn}}}^{\otimes \ceil{rn}}) \Pi_{s^{\ceil{rn}}}} \ge \frac{1}{\nu_{nr}}(\frac{1}{n}e^{\ceil{Rn}}-1)\tr{\sigma^{\otimes \ceil{rn}}\Pi_{s^{\ceil{rn}}}},
\end{equation}
thus by the data processing inequality of the sandwiched R\'enyi divergence, we have that for all $\alpha>0$
\begin{align}
    \widetilde{D}_{1+\alpha}(T_{s^{\ceil{rn}}}^{\otimes \ceil{rn}} \|\sigma^{\otimes \ceil{rn}}) 
    \ge &\frac{1}{\alpha}\log \tr{T_{s^{\ceil{rn}}}^{\otimes \ceil{rn}}\Pi_{s^{\ceil{rn}}}}^{1+\alpha}\tr{\sigma^{\otimes\ceil{rn}} \Pi_{s^{\ceil{rn}}}}^{-\alpha}
    \\ \ge  &\frac{1}{\alpha}\log \tr{T_{s^{\ceil{rn}}}^{\otimes\ceil{rn}}\Pi_{s^{\ceil{rn}}}}^{1+\alpha}\left(\tfrac{\nu_{nr}}{\frac{1}{n}e^{\ceil{Rn}}-1}\tr{T_{s^{\ceil{rn}}}^{\otimes \ceil{rn}} \Pi_{s^{\ceil{rn}}}}\right)^{-\alpha}
    \\ = &\frac{1}{\alpha}\log \tr{T_{s^{\ceil{rn}}}^{\otimes \ceil{rn}}\Pi_{s^{\ceil{rn}}}} + \log\left(\tfrac{\frac{1}{n}e^{\ceil{Rn}}-1}{\nu_{nr}}\right).
\end{align}
Therefore for all $\alpha>0$
\begin{equation}
\begin{split}
     1-F\left(T_{s^{\ceil{rn}}}^{\otimes \ceil{rn}}, \widetilde{T}^{\mathrm{NS}}_{s^{\ceil{rn}}}\right)
&\le 2\exp\left(\alpha\big(\widetilde{D}_{1+\alpha}(T_{s^{\ceil{rn}}}^{\otimes \ceil{rn}}\|\sigma^{\otimes \ceil{rn}})-\log(\tfrac{\frac{1}{n}e^{\ceil{Rn}}-1}{\nu_{nr}})\big)\right)\\
&\leq 2\exp\left(\alpha\big(nr \mathbb{E}_{s \sim P_S}\widetilde{D}_{1+\alpha}(T_s\|\sigma)-nR+\log 2n(n+1)^{|Z|}\big)\right) \\
&=2\exp\left(-\alpha n\big(R-r \mathbb{E}_{s \sim P_S}\widetilde{D}_{1+\alpha}(T_s\|\sigma)-\frac{\log 2n(n+1)^{|Z|}}{n}\big)\right) 
\end{split}
\end{equation}
We shall choose $\alpha = \frac{\delta_{n}}{r\var( {T}_{S \rightarrow Z} )}$. By Taylor expansion \cite[Proposition 11]{hayashi2016correlation}, we have that 
\begin{align}
\widetilde{D}_{1+\alpha}(P_{S}{T}_{S \rightarrow Z}\| P_{S} \otimes \sigma) &=  I(P_S:{T}_{S \rightarrow Z}) +\frac{1}{2}\alpha \var(P_{S}, {T}_{S\rightarrow Z} ) +o(\alpha)
\\&\le  I(P_S:{T}_{S \rightarrow Z})  + \frac{\delta_{n}}{2r} + o(\delta_{n}).
\end{align}
Since $R\ge r\cdot I(S:Z)_{P_{S}{T}_{S \rightarrow Z}}  +\delta_{n}$, we deduce that 
\begin{equation}
\begin{split}
1-F\left(T_{s^{\ceil{rn}}}^{\otimes \ceil{rn}}, \widetilde{T}^{\mathrm{NS}}_{s^{\ceil{rn}}}\right)
&\le \exp\left(	-n  \frac{\delta_{n}}{r\var({T}_{S \rightarrow Z} ) }\left( \delta_{n} -\frac{\delta_{n}}{2} -o(\delta_{n}) 	-\frac{\log 2n(n+1)^{|Z|}}{n} \right)\right)
\\&= \exp\left( - \frac{n\delta_{n}^2}{ 2r\var({T}_{S \rightarrow Z} ) } +\frac{\delta_{n}\log(2n(n+1)^{|Z|})}{r\var({T}_{S \rightarrow Z} )} +o\left({n\delta_{n}^2}\right)\right)
\\&\le \frac{1}{9}\exp\left( -2\sqrt{n}\log(n) \right),
\end{split}
\end{equation}
where we use in the last inequality $\delta_{n} = 3\frac{\sqrt{\log(n) r \var({T}_{S \rightarrow Z} )} }{n^{1/4}}$ and $n$ is sufficiently large.

By \cite[Prop 4.1]{Oufkir2024OctMeta}, this NS strategy $\widetilde{T}^{\mathrm{NS}}$ can be rounded to an EA strategy $\widetilde{T}^{\mathrm{EA}}$ of size $e^{\ceil{Rn}}$ satisfying
\begin{equation}
    \widetilde{T}^{\mathrm{EA}}_{s^{\ceil{rn}}} \mge (1-e^{-n}) \widetilde{T}^{\mathrm{NS}}_{s^{\ceil{rn}}}.
\end{equation}
Hence, there exists an EA conversion strategy to simulate $T^{\otimes \ceil{rn}}$ using $\mathrm{id}_2^{\otimes \ceil{Rn}}$ with an error:
\begin{align}
    \begin{split}
    \epsilon^{\mathrm{EA}}_{\mathrm{simulation}}  &\leq P(\widetilde{T}^{\mathrm{EA}}_{s^{\ceil{rn}}},T^{\otimes \ceil{rn}}_{s^{\ceil{rn}}}) \\
  &\leq P(\widetilde{T}^{\mathrm{EA}}_{s^{\ceil{rn}}}, \widetilde{T}^{\mathrm{NS}}_{s^{\ceil{rn}}})+P(\widetilde{T}^{\mathrm{NS}}_{s^{\ceil{rn}}},T^{\otimes \ceil{rn}}_{s^{\ceil{rn}}})  \\
  &\leq \sqrt{e^{-n}}+\sqrt{\frac{1}{9}\exp\left( -2\sqrt{n}\log(n) \right)} 
  \\&\leq \frac{1}{2}\exp\left( -\sqrt{n}\log(n) \right).
    \end{split}
\end{align}
\end{proof}


\section{Method of types}
\label{sec-types}

Let $n$ be an integer and  $\cX$ be a finite alphabet of size $|\cX|$. Let $x^n = x_1\cdots x_n$ be an element of $\cX^n$. For $x\in\cX$ we define $n(x|x^n)$ to be the number of occurrences of $x$ in the sequence $x_1\cdots x_n$: 
\begin{equation}
   n(x|x^n) = \sum_{t=1}^n \bid\{x_t=x\}. 
\end{equation}
A type $T$ is a probability distribution on $\cX$ of the form 
\begin{equation}
    T=\Big\{\frac{n_x}{n}\Big\}_{x\in \cX} \text{ where } n_x\in \mathbb{N} \text{ and } \sum_{x\in \cX} n_x=n. \label{eq:type}
\end{equation}
The set of types of alphabet $\cX$ of length $n$ is denoted $\cT_n(\cX)$. It is a finite set of size satisfying 
\begin{equation}
\label{equ:numbertype}
   |\cT_n(\cX)|\le (n+1)^{|\cX|}. 
\end{equation}
This simple bound can be proved using the simple  observation that each $n_x$ in Eq.~\eqref{eq:type} satisfies $n_x \in \{0,1, \dots, n\}$ and thus it has at most $n+1$ possibilities.
\\We say that $x^n = x_1\cdots x_n$ has type $T$ and write $x^n \sim T$ if for all $x\in \cX$ we have $ {n(x| x^n)} = {n}T_x$.
\\A probability distribution $p$ on $\cX^n$ is permutation invariant if for all permutation $\sigma \in \fS_n$, for all $ x_1\cdots x_n \in \cX^n$ we have 
\begin{equation}
    p(x_1\cdots x_n) = p(x_{\sigma_1}\cdots x_{\sigma_n}).
\end{equation}
Let  $T\in \cT_n(\cX)$ be a type and  $p$ be a permutation invariant  probability distribution on $\cX^n$.
Clearly if two sequences $x_{1}^n$ and $y_{1}^n$ have the same type $T$ then $y_{1}^n$ can be obtained by permuting the elements of $x_{1}^n$ so
\begin{equation}
    p(x^n)= p(y_1^n) , \quad \forall x^n\sim T, \; \forall y_1^n\sim T.
\end{equation}
We denote this value by $\frac{\alpha_T}{|T|}$ where $|T|$ is the number of elements in $\cX^n$ of type $T$. So we can write 
\begin{equation}
    p = \sum_{T\in \cT_n(\cX)} \alpha_T \cU_T
\end{equation}
where $\cU_T$ is the uniform probability distribution supported on $T$:
\begin{equation}
    \cU_T(x^n) = \frac{1}{|T|}\bid\{x^n\sim T\}.
\end{equation}
Since $p$ and $\{\cU_T\}_{T\in \cT_n(\cX)}$ are all probability distributions on $\cX^n$, $(\alpha_T)_{T\in \cT_n(\cX)}$ is a probability distribution on $\cT_n(\cX)$. In particular, since we have $|\cT_n(\cX)|\le (n+1)^{|\cX|}$ there is a type $T^\star$ such that 
\begin{equation}
    \alpha_{T^\star} \ge \frac{1}{ (n+1)^{|\cX|}}
\end{equation}
 because otherwise $\sum_{T\in \cT_n(\cX)}\alpha_T< \sum_{T\in \cT_n(\cX)} \frac{1}{ (n+1)^{|\cX|}} = |\cT_n(\cX)|\cdot\frac{1}{ (n+1)^{|\cX|}} \le 1$ which contradicts the fact that $(\alpha_T)_{T\in \cT_n(\cX)}$ is a probability distribution on $\cT_n(\cX)$.


\section{Variational formulas for strong converse exponents}\label{app:variational}

In this section we prove Proposition \ref{prop:ach-var-u} which we restate:

\begin{proposition}[Restatement of Proposition \ref{prop:ach-var-u}]
Let $\mathbb{A} \subset [0,1]$ be a convex set and let $W_{Y|X}$ and $T_{Z|S}$ be two classical channels over finite alphabets. 
We have that:
\begin{align}
 &  \sup_{\alpha\in \mathbb{A}}\sup_{1<p<\frac{1}{1-\alpha}}	\frac{1-p+\alpha p}{p}\left(r\cdot I_{(1-\alpha)p}(T)-I_{\frac{\alpha p}{p-1}}(W)\right).
 \\&=\inf_{P_{X}} \sup_{P_{S}} \inf_{\widetilde{W}_{Y|X}} \inf_{\widetilde{T}_{Z|S} }\sup_{\alpha\in \mathbb{A}} \;\; \alpha D\left(P_{X}\widetilde{W}_{Y|X}\middle\|  P_{X} W_{Y|X}\right) +  r(1-\alpha)D\left( P_{S}\widetilde{T}_{Z|S}\middle\| P_{S} T_{Z|S}\right) \\&\qquad\qquad\qquad\qquad\qquad +\alpha\left|r \cdot  I(P_{S}, \widetilde{T}_{Z|S} )-  I(P_{X}, \widetilde{W}_{Y|X} ) \right|_{+}.
\end{align}
\end{proposition}

\begin{proof}
Let $P_{Y}^\star = P_{X} \widetilde{W}_{Y|X}$. We have the following chain of equalities:
\begin{align}
&\inf_{P_{X}} \sup_{P_{S}}\inf_{\widetilde{W}_{Y|X}} \inf_{\widetilde{T}_{Z|S} }\sup_{\alpha\in \mathbb{A}} \alpha D\left(P_X\widetilde{W}_{Y|X}\middle\|  p_XW_{Y|X}\right) + (1-\alpha)  rD\left(P_S \widetilde{T}_{Z|S}\middle\| P_S T_{Z|S}\right) 
\\&\qquad \qquad \qquad \qquad \qquad +\alpha \left|r\inf_{P_Z}  D\left( P_S\widetilde{T}_{Z|S}\middle\| P_S P_Z\right) - \inf_{P_{Y}}D\left( P_X \widetilde{W}_{Y|X}\middle\|  P_X P_{Y}\right) \right|_{+}
\\&\overset{(a)}= \inf_{P_{X}} \sup_{P_{S}}\inf_{\widetilde{W}_{Y|X}} \inf_{\widetilde{T}_{Z|S} }\inf_{P_Z} \sup_{\alpha\in \mathbb{A}} \sup_{0\le \delta \le \alpha } \alpha D\left(P_X\widetilde{W}_{Y|X}\middle\|  p_XW_{Y|X}\right) + (1-\alpha)  rD\left(P_S \widetilde{T}_{Z|S}\middle\| P_S T_{Z|S}\right) 
\\&\qquad \qquad \qquad \qquad \qquad +\delta\left(r  D\left( P_S\widetilde{T}_{Z|S}\middle\| P_S P_Z\right) - D\left(P_X\widetilde{W}_{Y|X}\middle\|  P_X P_{Y}^{\star}\right) \right)
\\&\overset{(b)} =  \inf_{P_{X}} \sup_{P_{S}}\inf_{P_Z} \sup_{\alpha\in \mathbb{A}}\sup_{0\le \delta \le \alpha} \inf_{\widetilde{W}_{Y|X}} \inf_{\widetilde{T}_{Z|S} }\sup_{P_{Y}}\alpha D\left(P_X\widetilde{W}_{Y|X}\middle\|  p_XW_{Y|X}\right) + (1-\alpha)  rD\left(P_S \widetilde{T}_{Z|S}\middle\| P_S T_{Z|S}\right) 
\\&\qquad \qquad \qquad \qquad \qquad +\delta\left(r  D\left( P_S\widetilde{T}_{Z|S}\middle\| P_S P_Z\right) - D\left( P_X \widetilde{W}_{Y|X}\middle\|  P_X P_{Y}\right) \right)
\\&\overset{(c)} =   \inf_{P_{X}} \sup_{P_{S}}\inf_{P_Z} \sup_{\alpha\in \mathbb{A}}\sup_{0\le \delta \le \alpha}\sup_{P_{Y}} \inf_{\widetilde{W}_{Y|X}} \inf_{\widetilde{T}_{Z|S} }\alpha D\left(P_X\widetilde{W}_{Y|X}\middle\|  p_XW_{Y|X}\right)  -\delta D\left( P_X \widetilde{W}_{Y|X}\middle\|  P_X P_{Y}\right)
\\&\qquad \qquad \qquad \qquad \qquad + r (1-\alpha)   D\left(P_S \widetilde{T}_{Z|S}\middle\|P_S T_{Z|S}\right)+r\delta   D\left( P_S\widetilde{T}_{Z|S}\middle\| P_S P_Z\right)  
\end{align}
where we used Sion's minimax theorem: in $(a)$ the objective function is concave in $(\delta, \alpha)$ and convex in $P_Z$; in $(b$) the objective function is convex in $(\widetilde{W}_{x}, \widetilde{T}_{s} )$ even though $P_{Y}^\star = P \widetilde{W}_{Y}$, to see this we can write the function's part containing $\widetilde{T}$ as: 
\begin{align}
&\alpha D\left(\widetilde{W}_{x}\middle\|  W_{x}\right) - \delta D\left(\widetilde{W}_{x}\middle\|  P_{Y}^\star\right)= -(\alpha-\delta)H\left(\widetilde{W}_{x} \right) -\alpha \tr{\widetilde{W}_{x}\log W_x}+\delta \tr{\widetilde{W}_{x}\log \textstyle\sum_{x'}P_{X}(x')\widetilde{W}_{x'}};
\end{align}
in $(c)$ we used the fact that $D(P\| Q)$ is convex in $p$ and $q$. 
 
Using the variational property of the R\'enyi divergence \cite{vanErven2014Jun}, we obtain
\begin{align}
\inf_{\widetilde{W}_{x} }\alpha D\left(\widetilde{W}_{x}\middle\|  W_{x}\right)  -\delta D\left(\widetilde{W}_{x}\middle\|  P_{Y}\right)&=-\delta D_{\frac{\alpha}{\alpha-\delta}}\left(W_{x}  \middle\| P_{Y}\right),
\\\inf_{ \widetilde{T}_{s} }  (1-\alpha) D\left( \widetilde{T}_{s}\middle\| T_{s}\right)+\delta  D\left( \widetilde{T}_{s}\middle\| P_Z\right)  &= \delta D_{\frac{1-\alpha}{1+\delta-\alpha}}\left(T_{s}  \middle\| P_Z\right).
\end{align}
Note that because the infimum of linear functions is concave, the right hand sides of the previous equalities remain concave in $(\alpha, \delta)$ even when taking optimizations on the left hand sides on $P_{Y}$, $P_Z$,  and $P_{S}$ (see Lemma \ref{Lem:concave-alpha-delta} for a more precise statement). Therefore by applying successively Sion's minimax theorem, we have the chain of equalities:
\begin{align}
&\inf_{P_{X}} \sup_{P_{S}}\inf_{\widetilde{W}_{Y|X}} \inf_{\widetilde{T}_{Z|S} }\sup_{\alpha\in \mathbb{A}} \alpha D\left(P_X\widetilde{W}_{Y|X}\middle\|  p_XW_{Y|X}\right) + (1-\alpha)  rD\left(P_S \widetilde{T}_{Z|S}\middle\| P_S T_{Z|S}\right) 
\\&\qquad \qquad \qquad \qquad \qquad +\alpha \left|r\inf_{P_Z}  D\left( P_S\widetilde{T}_{Z|S}\middle\| P_S P_Z\right) - \inf_{P_{Y}}D\left( P_X \widetilde{W}_{Y|X}\middle\|  P_X P_{Y}\right) \right|_{+}
\\& =    \inf_{P_{X}} \sup_{P_{S}}\inf_{P_Z} \sup_{\alpha\in \mathbb{A}}\sup_{0\le \delta \le \alpha}\sup_{P_{Y}}  r \delta\mathbb{E}_{s\sim P_{S}}D_{\frac{1-\alpha}{1+\delta-\alpha}}\left(T_{s}  \middle\| P_Z\right) - \delta \mathbb{E}_{x\sim P_{X}}D_{\frac{\alpha}{\alpha-\delta}}\left(W_{x}  \middle\| P_{Y}\right)
\\&\overset{(d)} =    \inf_{P_{X}} \sup_{P_{S}} \sup_{\alpha\in \mathbb{A}}\sup_{0\le \delta \le \alpha}\inf_{P_Z}\sup_{P_{Y}}  r \delta\mathbb{E}_{s\sim P_{S}}D_{\frac{1-\alpha}{1+\delta-\alpha}}\left(T_{s}  \middle\| P_Z\right) - \delta \mathbb{E}_{x\sim P_{X}}D_{\frac{\alpha}{\alpha-\delta}}\left(W_{x}  \middle\| P_{Y}\right)
\\& =    \inf_{P_{X}} \sup_{\alpha\in \mathbb{A}}\sup_{0\le \delta \le \alpha}\sup_{P_{S}} \inf_{P_Z}\sup_{P_{Y}}  r \delta\mathbb{E}_{s\sim P_{S}}D_{\frac{1-\alpha}{1+\delta-\alpha}}\left(T_{s}  \middle\| P_Z\right) - \delta \mathbb{E}_{x\sim P_{X}}D_{\frac{\alpha}{\alpha-\delta}}\left(W_{x}  \middle\| P_{Y}\right)
\\&\overset{(e)} =    \sup_{\alpha\in \mathbb{A}}\sup_{0\le \delta \le \alpha}\inf_{P_{X}} \sup_{P_{S}} \inf_{P_Z}\sup_{P_{Y}}  r \delta\mathbb{E}_{s\sim P_{S}}D_{\frac{1-\alpha}{1+\delta-\alpha}}\left(T_{s}  \middle\| P_Z\right) - \delta \mathbb{E}_{x\sim P_{X}}D_{\frac{\alpha}{\alpha-\delta}}\left(W_{x}  \middle\| P_{Y}\right)
\\&\overset{(f)} =    \sup_{\alpha\in \mathbb{A}}\sup_{0\le \delta \le \alpha}  r \delta I_{\frac{1-\alpha}{1+\delta-\alpha}}\left(T_{Z|S}\right) - \delta I_{\frac{\alpha}{\alpha-\delta}}\left(W_{Y|X}\right)
\end{align}
where in $(d)$ we used Lemma~\ref{Lem:concave-alpha-delta} and that the R\'enyi divergence is convex in $P_Z$ and so is the objective function ; in $(e)$ we used Lemma \ref{Lem:concave-alpha-delta} and that the supremum of linear functions is convex so the objective function is convex in $P_{X}$; in $(f)$ we used the equivalence between these $\alpha$-mutual information \cite[Proposition 1]{csiszar1995generalized}. 

Finally, making the change of variable $ p = \frac{1}{1+\delta-\alpha}\in (1, \frac{1}{1-\alpha})$ concludes the proof.
\end{proof}


\section{Missing proofs and lemmas}\label{sec:lemmas}

\begin{lemma}[Restatement of Lemma \ref{lem:reverse-chernoff}]\label{lem-app:reverse-chernoff}
    Let $P, Q \in \cP(\cX)$. We have that for all $\eta\in (0,\frac{1}{2})$:
    \begin{align}
        \tr{P \wedge Q} \ge (1-2\eta) \sup_{ V \in \cP(\cX)}\inf_{0 < \alpha < 1} \exp\left(-\alpha D_{\max}^{\eta}(V \| P) -(1-\alpha) D_{\max}^{\eta}(V \| Q) \right).
    \end{align}
\end{lemma}
\begin{proof}
    Let $V, V_{1}$ and $V_2$ be probability distributions such that $\mathrm{TV}(V_1 , V) \le \eta$, $\mathrm{TV}(V_2 , V)\le \eta$ and 
    \begin{align}
        D_{\max}^{\eta}(V \| P) = D_{\max}^{}(V_1 \| P), \quad  D_{\max}^{\eta}(V \| Q) = D_{\max}^{}(V_2 \| Q).
    \end{align}
    By the definition of the max divergence, we have that 
    \begin{align}
        P \mge e^{-D_{\max}^{\eta}(V \| P)} V_1 \;\text{ and } \;\; Q \mge e^{-D_{\max}^{\eta}(V \| Q)} V_2.
    \end{align}
    Hence 
    \begin{align}
        \tr{P \wedge Q} &\ge \tr{e^{-D_{\max}^{\eta}(V \| P)} V_1 \ \bigwedge \  e^{-D_{\max}^{\eta}(V \| Q)} V_2} 
        \\&\ge \min \left\{e^{-D_{\max}^{\eta}(V \| P)},  e^{-D_{\max}^{\eta}(V \| Q)} \right\}\tr{ V_1 \wedge V_2}. 
    \end{align}
    Moreover since $\mathrm{TV}(V_1 , V) \le \eta$ and  $\mathrm{TV}(V_2 , V)\le \eta$ we have by the triangle inequality $\mathrm{TV}(V_1 , V_{2})\le 2\eta$ thus 
    \begin{align}
        \tr{ V_1 \wedge V_2} =1-\mathrm{TV}(V_1 , V_{2})\ge 1-2\eta. 
    \end{align}
    Finally
     \begin{align}
        \tr{P \wedge Q} &\ge \min \left\{e^{-D_{\max}^{\eta}(V \| P)},  e^{-D_{\max}^{\eta}(V \| Q)} \right\}\tr{ V_1 \wedge V_2} 
        \\&\ge (1-2\eta) \inf_{0< \alpha < 1} e^{-\alpha D_{\max}^{\eta}(V \| P) -(1-\alpha) D_{\max}^{\eta}(V \| Q) }.
    \end{align}
    Since the choice of $V$ is arbitrary, we take the supremum over $V$ to conclude the proof.
\end{proof}

\begin{lemma}\label{lem:iid}
We have that for all probability distributions $p$ and $q$ and integer $n$:
\begin{align}
     D_{\max}^{\eta}(P^{\times n} \| Q^{\times n}) &\le  nD(P\| Q) + \cO(\sqrt{n}).
\end{align}
\end{lemma}

\begin{proof}
Let $\alpha = 1+\frac{1}{\sqrt{n}}$ and $g(\eta) = -\log(1- \sqrt{1-\eta^2})$. We have that using \cite[Proposition 6.22]{tomamichel2015quantum}:
\begin{align}
     D_{\max}^{\eta}(P^{\times n} \| Q^{\times n}) &\le D_{\alpha}(P^{\times n}\| Q^{\times n}) +\frac{g(\eta)}{\alpha-1}
     \\&=  nD_{\alpha}(P\| Q) +\frac{g(\eta)}{\alpha-1}
      \\&=  nD_{1+\frac{1}{\sqrt{n}}}(P\| Q) +\sqrt{n} g(\eta)
              \\&\le  nD(P\| Q) +\cO(\sqrt{n}),
\end{align}
where we use $D_{\alpha}(P\|Q)\le D(P\| Q)+ \cO(\alpha-1)$ \cite{tomamichel2015quantum}.
\end{proof}

\begin{lemma}\label{lem:al-iid}
Let $P_{1}, \dots, P_{n}\in \cP(\cX)$ and $P^n \in \cP(\cX^n)$ such that  $\mathrm{TV}(P^n, P_{1}\times \cdots \times P_{n})\le \kappa$. Let $Q_{1}, \dots, Q_{n}\in \cP(\cY)$  and $Q_{\min}$ is the smallest  positive part of  all $\{Q_{t}\}_{t\in [n]}$. Then we have that 
\begin{align}
    & D_{\max}^{\eta}(P^n \| Q_{1} \times \cdots \times Q_{n}) \le   \sum_{t=1}^n D(P_{t}\| Q_{t}) + (\sqrt{n}+1)\cdot \kappa\cdot (\tfrac{1}{Q_{\min}})^{{\sqrt{n}}}+\cO(\sqrt{n}).
\end{align}
\end{lemma}

\begin{proof}
We have for all $\alpha \in (0,1)$  and  $g(\eta) = -\log(1- \sqrt{1-\eta^2})$ \cite[Proposition 6.22]{tomamichel2015quantum}:
\begin{align}
    & D_{\max}^{\eta}(P^n \| Q_{1} \times \cdots \times Q_{n}) 
    \\&\le D_{\alpha}(P^n\| Q_{1} \times \cdots \times Q_{n}) +\frac{g(\eta)}{\alpha-1}
     \\&\le D_{\alpha}(P_{1}\times \cdots \times P_{n}\| Q_{1} \times \cdots \times Q_{n}) +\frac{\alpha}{\alpha-1} \log\left( 1+ \kappa\cdot (\tfrac{1}{Q_{\min}^n})^{\frac{\alpha-1}{\alpha}}\right)
+\frac{g(\eta)}{\alpha-1}
\end{align}
where we use the continuity bound of \cite[Corollary 5.7]{Bluhm2024Dec} in the last inequality. We choose $\alpha = 1+\frac{1}{\sqrt{n}}$ and obtain: 
\begin{align}
    & D_{\max}^{\eta}(P^n \| Q_{1} \times \cdots \times Q_{n}) 
    \\&\le  D_{1+\frac{1}{\sqrt{n}}}(P_{1}\times \cdots \times P_{n}\| Q_{1} \times \cdots \times Q_{n}) +(\sqrt{n}+1) \log\left( 1+ \kappa\cdot (\tfrac{1}{Q_{\min}})^{\sqrt{n}}\right)+ \sqrt{n}g(\eta)
            \\&= \sum_{t=1}^n  D_{1+\frac{1}{\sqrt{n}}}(P_{t}\| Q_{t} ) +(\sqrt{n}+1) \log\left( 1+ \kappa\cdot (\tfrac{1}{Q_{\min}})^{\sqrt{n}}\right)+ \sqrt{n}g(\eta)
        \\&\le  \sum_{t=1}^n D(P_{t}\| Q_{t}) +(\sqrt{n}+1) \log\left( 1+ \kappa\cdot (\tfrac{1}{Q_{\min}})^{\sqrt{n}}\right)+\cO(\sqrt{n})
        \\&\le  \sum_{t=1}^n D(P_{t}\| Q_{t}) +(\sqrt{n}+1) \cdot \kappa\cdot (\tfrac{1}{Q_{\min}})^{\sqrt{n}}+\cO(\sqrt{n}).
        \end{align}
\end{proof}

\begin{lemma}[\cite{li2024operational}]\label{Lem:fid-D}
Let $\rho$, $\sigma$ and $\tau$ be quantum states. We have that:
\begin{align}
-\log F(\rho\| \sigma) \le D(\tau\| \rho) + D(\tau\| \sigma). 
\end{align}
\end{lemma}

\begin{lemma}\label{lem:DPI-SC} Let $N_{ZX|YS}$ be a non-signaling map and $P_{S}$ be a probability distribution. Let $\mathbb{D}$ be a divergence satisfying the data-processing inequality. We have that for all channels $(T_{1})_{Y|X}$ and $(T_{2})_{Y|X}$:
\begin{align}
\sup_{P_{X}}\mathbb{D}\left(P_{X}  (T_{1})_{Y|X} \middle\| P_{X}  (T_{2})_{Y|X} \right)\ge \mathbb{D}\left(P_{S} (N\circ T_{1})_{Z|S} \middle\| P_{S} (N\circ T_{2})_{Z|S} \right).
\end{align}
\end{lemma}
\begin{proof}
Define the map 
\begin{align}
M_{ZS|XY}(zs|xy) = \frac{P_{S}(s)N(zx|ys) }{P_{X}(x)}.
\end{align}
Since $N$ is non-signaling, we have that $\sum_{z}N(zx|ys)  = N(x|s)$, so $M$ is a stochastic map iff $P_{X}(x) = \sum_{s} P_{S}(s) N(x|s)$. With such $P_{X}$ we have that by the data-processing inequality:
\begin{align}
&\mathbb{D}\left(P_{X}  (T_{1})_{Y|X} \middle\| P_{X} (T_{2})_{Y|X} \right)
\\&\ge \mathbb{D}\left( \sum_{x,y} M_{ZS,XY=xy} P_{X}(x) (T_{1})_{Y|X}(y|x)\middle\| \sum_{x,y} M_{ZS,XY=xy} P_{X}(x) (T_{2})_{Y|X}(y|x) \right)
\\&=\mathbb{D}\left( P_{S}\sum_{x,y} N_{ZX=x|Y=yS} (T_{1})_{Y|X}(y|x)\middle\| P_{S}\sum_{x,y} N_{ZX=x|Y=yS}  (T_{2})_{Y|X}(y|x) \right)
\\&=\mathbb{D}\left(P_{S} (N\circ T_{1})_{Z|S} \middle\| P_{S} (N\circ T_{2})_{Z|S} \right).
\end{align}
\end{proof} 

\begin{lemma}\label{Lem:concave-alpha-delta}
Let $T_{Z|S}$ and $W_{Y|X}$ be two channels over finite alphabets. The following functions are concave on $\{(\alpha, \delta) \in [0,1]^2: \alpha \ge \delta \}$:
\begin{align}
&f: (\alpha, \delta) \mapsto \delta \sup_{P_{S}}\inf_{P_Z} \mathbb{E}_{s\sim P_{S}}D_{\frac{1-\alpha}{1+\delta-\alpha}}\left(W_s  \middle\| P_Z\right),
\\&g: (\alpha, \delta) \mapsto  - \delta\inf_{P_{Y}}  \mathbb{E}_{x\sim P_{X}}D_{\frac{\alpha}{\alpha - \delta}}\left(T_x  \middle\| P_{Y}\right).
\end{align}
\end{lemma}
\begin{proof}
We have by the variational property of the $\alpha$ R\'enyi divergence \cite{vanErven2014Jun}:
\begin{align}
f(\alpha)&= \delta \sup_{P_{S}}\inf_{P_Z} \mathbb{E}_{s\sim P_{S}}D_{\frac{1-\alpha}{1+\delta-\alpha}}\left(W_s  \middle\| P_Z\right)
\\&= \sup_{P_{S}}\inf_{P_Z} \inf_{\widetilde{T}} (1-\alpha) D\left( P_{S}\widetilde{T}_{Z|S}\middle\| P_{S}P_Z\right) + \delta  D\left(P_{S} \widetilde{T}_{Z|S}\middle\| P_{S}T_{Z|S}\right)
\\&= \sup_{P_{S}}\inf_{\widetilde{T}}\inf_{P_Z} (1- \alpha) D\left( P_{S}\widetilde{T}_{Z|S}\middle\| P_{S}P_Z\right) + \delta  D\left(P_{S} \widetilde{T}_{Z|S}\middle\| P_{S}T_{Z|S}\right)
\\&= \inf_{\widetilde{T}}\sup_{P_{S}}\inf_{P_Z}  (1-\alpha) D\left( P_{S}\widetilde{T}_{Z|S}\middle\| P_{S}P_Z\right) + \delta  D\left(P_{S} \widetilde{T}_{Z|S}\middle\| P_{S}T_{Z|S}\right)
\end{align}
where we use Sion's minimax theorem \cite{sion1958general}: the objective function is concave in $P_{S}$ as an infimum of linear functions and convex in $\widetilde{T}$ since $D$ is jointly convex in its  arguments.  Now, $f$ is an infimum of linear functions in $\alpha$ so it is concave. 

Similarly, we can write the function $g$ using the variational property of the $\alpha$ R\'enyi divergence \cite{vanErven2014Jun}:
\begin{align}
g(\alpha)&= - \delta\inf_{P_{Y}}  \mathbb{E}_{x\sim P_{X}}D_{\frac{\alpha}{\alpha - \delta}}\left(T_x  \middle\| P_{Y}\right)
\\&=\sup_{P_{Y}}\inf_{\widetilde{T}}  \alpha D\left(P_{X}\widetilde{W}_{Y|X}\middle\|  P_{X}W_{Y|X}\right) -\delta D\left(P_{X}\widetilde{W}_{Y|X}\middle\|  P_{X}P_{Y}\right) 
\\&=\inf_{\widetilde{T}}\sup_{P_{Y}}  \alpha D\left(P_{X}\widetilde{W}_{Y|X}\middle\|  P_{X}W_{Y|X}\right) -\delta D\left(P_{X}\widetilde{W}_{Y|X}\middle\|  P_{X}P_{Y}\right) 
\end{align}
where we use Sion's minimax theorem \cite{sion1958general}: the objective function is concave in $P_{Y}$ and convex in $\widetilde{T}$ since it can be written as 
\begin{align}
&\alpha D\left(P_{X}\widetilde{W}_{Y|X}\middle\|  P_{X}W_{Y|X}\right) -\delta D\left(P_{X}\widetilde{W}_{Y|X}\middle\|  P_{X}P_{Y}\right) 
\\&= \mathbb{E}_{x\sim P_{X}} -(\alpha-\delta)H(\widetilde{W}_{x}) -\alpha\tr{\widetilde{W}_{x}\log W_{x}} +\delta \tr{\widetilde{W}_{x}\log P_{Y}}.
\end{align}
Since $g$ is an infimum of linear functions in $\alpha$, it is concave.
\end{proof}

The following Lemma is from~\cite{li2024operational}.
\begin{lemma}[\cite{li2024operational}]
\label{lem:appen1}
Let $\rho, \sigma\in\mc{S}(\mc{H})$, and let $\mc{H}=\bigoplus_{i\in \mc{I}}\mc{H}_i$ decompose into a set of mutually orthogonal subspaces $\{\mc{H}_i\}_{i\in \mc{I}}$. Suppose that $\sigma=\sum\limits_{i \in \mc{I}} \sigma_i$ with $\supp(\sigma_i)\subseteq \mc{H}_i$. Then
\begin{equation}
F\left(\sum_{i \in \mc{I}} \Pi_i \rho \Pi_i, \sigma\right) \leq \sqrt{|\mc{I}|} F(\rho, \sigma),
\end{equation}
where $\Pi_i$ is the projection onto $\mc{H}_i$.
\end{lemma}

The following lemma can be deduced from the same process as~\cite[Lemma~2]{WangWilde2019resource}. It was also implicitly proved in~\cite{LWD2016strong}.

\begin{lemma}
\label{lem:hof}
Let $\rho, \sigma \in \mc{S}(\mc{H})$ and $\tau \in \mc{P}(\mc{H})$, and suppose $\supp(\sigma)\not\perp\supp(\tau)$. Fix $\alpha \in (\frac{1}{2}, 1)$ and $\beta\in(1,+\infty)$
such that $\frac{1}{\alpha}+\frac{1}{\beta}=2$. Then
\begin{equation}
\frac{2\alpha}{1-\alpha} \log F(\rho, \sigma) \leq \widetilde{D}_\beta(\rho\|\tau)-\widetilde{D}_\alpha(\sigma\|\tau).
\end{equation}
\end{lemma}

\end{document}

%% file: preamble.tex
\usepackage{graphicx,epic,eepic,epsfig,amsmath,amsfonts,latexsym,amssymb,color,enumerate } 

\usepackage{mathtools}
\usepackage{fullpage}

\usepackage{amsthm}

\newtheorem{theorem}{Theorem}
\newtheorem{lemma}[theorem]{Lemma}
\newtheorem{remark}[theorem]{Remark}

\newtheorem{proposition}[theorem]{Proposition}

\usepackage{booktabs} 
\usepackage{array} 
\usepackage{caption} 

\DeclareMathOperator{\var}{Var}

\usepackage{mathtools} 
\DeclarePairedDelimiter{\ceil}{\lceil}{\rceil}

\usepackage{bm}
\newcommand{\beq}{\begin{equation}}
\newcommand{\eeq}{\end{equation}}
\newcommand{\mc}{\mathcal}
\newcommand{\ox}{\otimes}

\newcommand*{\id}{\mathrm{id}}
\newcommand*{\tr}[1]{\mathrm{Tr}\left[#1\right]}

\newcommand*{\ket}[1]{| #1 \rangle}
\newcommand*{\bra}[1]{\langle #1 |}

\newcommand*{\mge}{\succcurlyeq}
\newcommand*{\mle}{\preccurlyeq}

\newcommand*{\bid}{\mathbf{1}}

\newcommand*{\cA}{\mathcal{A}}

\newcommand*{\cD}{\mathcal{D}}
\newcommand*{\cE}{\mathcal{E}}

\newcommand*{\cG}{\mathcal{G}}

\newcommand*{\cO}{\mathcal{O}}
\newcommand*{\cP}{\mathcal{P}}

\newcommand*{\cR}{\mathcal{R}}
\newcommand*{\cS}{\mathcal{S}}
\newcommand*{\fS}{\mathfrak{S}}
\newcommand*{\cT}{\mathcal{T}}
\newcommand*{\cU}{\mathcal{U}}

\newcommand*{\cX}{\mathcal{X}}
\newcommand*{\cY}{\mathcal{Y}}
\newcommand*{\cZ}{\mathcal{Z}}
\newcommand{\supp}{{\operatorname{supp}}}
\usepackage{hyperref}

\usepackage{authblk}

\usepackage[maxnames = 99]{biblatex} 